\documentclass[11pt,letterpaper]{article} 
\usepackage[margin=1in]{geometry}

\usepackage[linesnumbered,ruled,vlined]{algorithm2e}

\usepackage{caption}
\usepackage{graphicx}

\SetAlFnt{\small}
\SetAlCapFnt{\small}
\SetAlCapNameFnt{\small}
\SetAlCapHSkip{0pt}
\IncMargin{-\parindent}

\SetCommentSty{mycommfont}
\SetKwInput{KwData}{Input}
\SetKwInput{KwResult}{Output}
\SetKwFunction{Leaf}{isLeaf}
\SetKwFunction{Branch}{isBranching}
\SetKwFunction{HasSummary}{HasSummary}
\SetKwFunction{Summary}{Summary}
\SetKwFunction{RemoveInterior}{RemovePathInterior}
\SetKwFunction{AChildren}{ActiveChildren}
\SetKwFunction{ExtractActivePath}{ExtractActivePath}
\SetKwFunction{PruneAtBranch}{BranchPruning}
\SetKwFunction{BestClosedExtension}{BestExtension}
\SetKwFunction{DFS}{TraverseAndPrune}
\SetKwFunction{AllChildrenAreReduciblePaths}{IsReadyBranch}
\SetKwFunction{QuerySW}{QuerySW}
\SetKwProg{Fn}{}{:}{end}

\usepackage{amsmath}
\usepackage{amsfonts}
\usepackage{amssymb}
\usepackage{amsthm}
\allowdisplaybreaks
\usepackage{mathtools}
\usepackage{tikz}
\usetikzlibrary{positioning}
\usetikzlibrary{shadows,patterns}
\usetikzlibrary{fit,matrix}
\usetikzlibrary{decorations.pathreplacing}
\usepackage{enumitem}
\usepackage{float}

\usepackage{booktabs}
\usepackage{subcaption}

\usepackage[T1]{fontenc}
\usepackage[tt=false]{libertine}
\usepackage[libertine]{newtxmath}
\usepackage{hyperref}
\usepackage[svgnames]{xcolor}
\hypersetup{colorlinks={true},urlcolor={blue},linkcolor={DarkBlue},citecolor=[named]{DarkGreen}}
\usepackage[authoryear,square]{natbib}

\usepackage{doi}
\usepackage{bm}
\usepackage{bbm}
\usepackage[capitalise,nameinlink,noabbrev]{cleveref}
\usepackage{xspace}

\crefname{algocf}{Algorithm}{Algorithms}
\Crefname{algocf}{Subroutine}{Subroutines}
\crefname{algocfline}{line}{lines}
\Crefname{algocfline}{Line}{Lines}
\crefname{subroutine}{subroutine}{subroutines}
\Crefname{subroutine}{Subroutine}{Subroutines}

\usepackage{booktabs} 
\usepackage{authblk} 
\usepackage{graphicx} 

\usepackage{tcolorbox}

\newtheorem{theorem}{Theorem}[section]
\newtheorem{corollary}{Corollary}[section]
\newtheorem{proposition}{Proposition}[section]
\newtheorem{lemma}{Lemma}[section]

\newtheorem{observation}{Observation}[section]
\newtheorem{question}{Question}
\newtheorem{inftheorem}{Informal Theorem}

\theoremstyle{definition}
\newtheorem*{comment*}{Comment}
\newtheorem{definition}{Definition}[section]
\newtheorem{remark}{Remark}

\newcommand{\SW}{\text{\normalfont SW}}

\newcommand{\bv}{\mathbf{v}}
\newcommand{\tr}{\mathtt{tr}}
\newcommand{\Mod}[1]{\ (\mathrm{mod}\ #1)}

\usepackage{xcolor}
\usepackage{xspace}

\newcommand{\RMOWO}{\textnormal{\textsc{Rand-MoWo}}\xspace}

\newlist{condition}{enumerate}{2}
\setlist*[condition,1]{label=\arabic*,ref=\arabic*}
\crefname{conditioni}{Property}{Properties}

\author[1]{Aris Filos-Ratsikas}
\author[2]{Georgios Kalantzis}

\affil[1]{Deparment of Computing, Imperial College London}
\affil[2]{School of Informatics, University of Edinburgh}

\date{}

\title{The Distortion of Stable Matching\thanks{Aris Filos-Ratsikas was supported by the UK Engineering and Physical Sciences Research Council (EPSRC) grant EP/Y003624/1.}}

\begin{document}
\maketitle

\begin{abstract}
We initiate the study of \emph{distortion} in stable matching. Concretely, we aim to design algorithms that have limited access to the agents' cardinal preferences and compute stable matchings of high quality with respect to some aggregate objective, e.g., the social welfare. Following the recent agenda of \citet{amanatidis2021peeking}, we seek to find the best possible tradeoffs between the distortion of an algorithm and the amount of information elicited via means of cardinal \emph{value queries}. 

As a warm-up, we show that purely ordinal algorithms (i.e., algorithms that only use the preference rankings of the agents but no value queries), including the classic Deferred Acceptance (DA) algorithm of \citet{gale1962college}, have unlimited distortion. Despite this strong impossibility, we observe that a simple randomized version of the DA algorithm achieves a distortion of $2$, and show that this is in fact optimal among all randomized stable matching algorithms.

We then consider algorithms that also use value queries on top of the ordinal ranking information. Again, it is easy to see that the same bound of $2$ can be achieved with only $1$ value query per agent, and thus we focus on the question of improving upon this bound. To this end, we provide two strong lower bounds, namely that any stable matching algorithm requires $\Omega(n)$ \emph{non-adaptive} value queries or $\Omega(\log n)$ \emph{adaptive} value queries per agent to achieve a distortion smaller than $2$, where $n$ is the number of agents. Finally, we show an asymptotically tight bound of $\Theta(\log (1+n\varepsilon)/\varepsilon)$ on the number of queries that are necessary and sufficient to achieve near-optimal distortion, i.e., a distortion of $1+\varepsilon$, for any $\varepsilon >0$.  
\end{abstract}

\newpage 


\section{Introduction}

The \emph{distortion} in social choice theory \citep{procaccia2006distortion} measures the deterioration of some aggregate social objective due to limited information about the agents' preferences, which are expressed via cardinal \emph{utilities}. Much like standard notions in algorithm design such as the approximation ratio or the competitive ratio, the distortion quantifies the extent to which an algorithm that operates under some natural restrictions can approximate the best possible outcome, which could be achieved without those restrictions. In the case of the distortion, these restrictions are not due to computational considerations or uncertainty about the future, but rather due to the cognitive limitations of the participants to accurately express their preferences on a cardinal scale.

The first variants of this setting (e.g., see \citep{boutilier2015optimal,caragiannis2017subset}) considered \emph{ordinal} algorithms, i.e., algorithms that only input preference rankings consistent with these utilities, a rather less cognitively demanding preference elicitation device. The objective of those algorithms is to achieve good approximations to the \emph{social welfare}, i.e., the sum of the agents' utilities. More recent works consider algorithms that use a mix of ordinal and cardinal information and study the tradeoffs between their efficiency and the amount (and type) of information they elicit \citep{amanatidis2021peeking,amanatidis2022few,amanatidis2024dont,ebadian2025every,filos2025utilitarian}. Over the past two decades since its inception, the distortion has been studied in a plethora of settings in the epicenter of research in theoretical computer science and artificial intelligence \citep{anshelevich2021distortion}.

Yet, in the midst of this very rich literature, there is still a fundamental setting that has evidently not received enough attention in the context of distortion, namely the \emph{stable matching problem}. Indeed, this problem is one of the most well-studied at the interface of economics and computer science, e.g., see \citep{gale1962college,gusfield1987three,gusfield1989stable,roth1992two,roth1993stable,manlove2013algorithmics}, and captures applications such as college admissions \citep{gale1962college}, placement of medical residents \citep{roth1984evolution,roth1999redesign}, and assignments of students to schools \citep{abdulkadirouglu2005new}, among many others. In this setting, two sets of agents (typically referred to as ``men'' and ``women'') have preferences over the members of the other set, and the goal is to compute a matching without any \emph{blocking pairs}, i.e., any pairs whose members would rather be matched to each other over their assigned partners. 

The existence of a stable matching on any instance with strict preferences was established by \citet{gale1962college} via the famous \emph{Deferred Acceptance} algorithm (also known as the ``Gale-Shapley algorithm''). This algorithm is ordinal, and is known to produce a matching which is the best possible among all stable matchings for the \emph{proposing side}, and the worst possible for the other side. The question of finding the \emph{optimal stable matching}, i.e., the stable matching that maximizes the sum of the utilities that underlie the ordinal preferences of the agents, was studied in classic works, e.g., see \citep{irving1987efficient,feder1989new,roth1993stable} and \citep[Section 3.6]{gusfield1989stable}. These works proposed algorithms that compute an optimal stable matching in polynomial time, but, crucially, require full access to the agents' cardinal utilities.

In the context of distortion, and motivated by the cognitive burden of utility elicitation discussed above, it is very natural to consider the quality of stable matchings that are obtained by algorithms that use ordinal and only limited cardinal information. Concretely, we would like to answer the following general question: 

\begin{question}
What is the distortion of stable matching algorithms that use a limited amount of information about the utilities of the agents?
\end{question}

\subsection{Our Contributions}

We consider algorithms for stable matching that operate under limited information. In particular, we follow the research agenda put forward by \citet{amanatidis2021peeking}, and study algorithms which have access to the (ordinal) preference rankings of the agents for the agents on the other side, and which can also perform a number of \emph{value queries} to the cardinal utilities of the agents; these queries can depend on the ordinal preferences, but also possibly on the information about the utilities obtained by previous queries. The performance of an algorithm in this regime is measured by the tradeoff between the distortion and the number of queries per agent. We quantify the distortion in terms of the social welfare objective, against the best possible \emph{stable} matching.\footnote{Note that our benchmark is the social welfare of the best possible stable matching rather than that of the best possible matching. It can already be inferred by the work of \citet{gale1962college} (see also \citep{anshelevich2013anarchy}) that in some cases, the best possible stable matching (even one obtained with full cardinal information) cannot provide any meaningful approximation to the social welfare of the optimal (non-stable) matching, see \cref{thm:opt-matching-benchmark-infinity}. Since we are interested in the effect of limited information (and not stability) on the social welfare objective, our benchmark is the appropriate one. \label{footnote1}}

\paragraph{Warm-up: 0 queries, i.e., purely ordinal algorithms.} As a warm-up, in \cref{sec:ordinal} we first investigate the distortion of purely ordinal algorithms, i.e., algorithms that do not perform any queries to the agents' utilities. Using a simple construction, we first show a rather negative result, namely that any ordinal algorithm for stable matching (including the Deferred Acceptance algorithm), has unbounded distortion. The proof boils down to exploiting the inability of such a mechanism to distinguish between the better among the \emph{man-optimal} matching and the \emph{woman-optimal} stable matching, one of which can be made to have arbitrarily larger welfare than the other. On the other hand, based on a simple observation about the social welfare of man-optimal or woman-optimal matchings, one can easily design a randomized stable matching algorithm that achieves a distortion of $2$; the algorithm is the natural randomized variant of the Deferred Acceptance algorithm, which selects the proposing side equiprobably. Our main technical result of this section is a matching lower bound, which shows that no ordinal randomized algorithm for stable matching can achieve a better distortion. This bound is in fact quite robust, as it also applies to \emph{ex-ante} stable (or stable \emph{fractional}) matchings, a weaker stability notion investigated by \citet{caragiannis2021stable}. In contrast, the aforementioned randomized Deferred Acceptance algorithm is \emph{ex-post} stable, see \cref{def:randomized-stable-matching}. We remark that all of the lower bounds of \cref{sec:ordinal} hold even when the agents' utilities are \emph{unit-sum normalized}, the standard normalization for ordinal algorithms (without value queries) in the distortion literature, e.g., see \citep{anshelevich2021distortion}, whereas the upper bounds do not require any normalization assumptions.

\paragraph{Tradeoffs between the distortion and the number of queries.} We then turn our attention to algorithms that employ a (limited) amount of cardinal information on top of the preference rankings in \cref{sec:queries}. Such query-enhanced algorithms have been shown to achieve much improved distortion bounds with a relatively small number of queries for a number of different settings, including single-winner voting, one-sided and two-sided (non-stable) matching, resource allocation etc. \citep{amanatidis2021peeking,amanatidis2022few,amanatidis2024dont,ebadian2025every}. 

Similarly to the case of the randomized Deferred Acceptance algorithm, it can be easily observed that a (deterministic) query-enhanced version of the algorithm achieves a distortion of $2$, by using only a single query per agent. Intuitively, this one query allows the algorithm to identify the best matching between the man-optimal and the woman-optimal matchings, which, in turn, provides a $2$-approximation to the social welfare of the overall optimal matching. The more interesting, and much more technically challenging question here is how many queries an algorithm would have to perform in order to beat the distortion $2$ barrier. To this end, we provide crisp positive and negative results. 

First, we show that to achieve any improvement over $2$, the algorithm would need to perform at least $\Omega(\log n)$ queries per agent, and this lower bound becomes stronger (namely, $\Omega(n)$), when these queries are \emph{non-adaptive}, i.e., they are fixed in advance and thus do not rely on the answers to previous queries. We state the corresponding informal theorem below. 

\begin{inftheorem}\label{infthm:query-enhanced-dist-2}
A query-enhanced version of the Deferred Acceptance algorithm achieves a distortion of $2$ using only one query per agent. Furthermore, any stable matching algorithm that achieves a distortion better than $2$ requires $\Omega(\log n)$ adaptive queries or $\Omega(n)$ non-adaptive queries per agent. 
\end{inftheorem}

\noindent The impossibility part of \cref{infthm:query-enhanced-dist-2} above further motivates the following more general question: ``How does the number of queries relate with the best achievable distortion?''. Concretely, we can rephrase this general question as follows.
\begin{question}
``Given an $\varepsilon >0$, how many queries are sufficient and necessary to achieve a distortion of $1+\varepsilon$?''\end{question} \smallskip 

\noindent We provide a conclusive answer to this question, by proving a tight $\Theta(\log (1+n\varepsilon)/\varepsilon)$ bound on the amount of queries that are necessary and sufficient.

\begin{inftheorem}\label{infthm:query-enhanced-dist-1pluseps}
There is a polynomial-time stable matching algorithm which achieves distortion $1+\varepsilon$ using $O(\log (1+n\varepsilon)/\varepsilon)$ queries per agent. Additionally, any stable matching algorithm that achieves a distortion of $1+\varepsilon$ must perform $\Omega(\log (1+n\varepsilon)/\varepsilon)$ queries per agent.
\end{inftheorem}

\subsection{Our Techniques}\label{sec:our-techniques}

Our results provide a comprehensive picture of the distortion landscape for the stable matching problem. In this section, we highlight the ideas behind the proofs of the most important or the most technically and conceptually interesting of these results. 

\paragraph{Number of queries required to beat the distortion barrier of $2$.} In \cref{sec:more-queries-impossibility}, we provide our main lower bounds, namely that $\Omega(\log n)$ queries per agent are required to beat the distortion barrier of $2$, and that this query bound goes up to $\Omega(n)$ when the queries are non-adaptive. Both of these bounds are achieved using the same stable matching instances and the same general methodology, but in the case of adaptive queries, the lower bound argument is substantially more involved. 

First, we identify a particular preference profile, which we refer to as the \emph{cyclic shift profile}. This profile is carefully constructed so that it has precisely $n$ stable matchings of a very specific \emph{diagonal} structure, in an appropriate pictorial representation, see \cref{fig:cyclic-shift-profile}. This structure is established in \cref{lem:cyclic-shift-profile-stable-matchings} via employing the aforementioned concept of rotations and a careful inductive argument. On top of this (ordinal) profile, we add consistent cardinal utilities. These utilities are \emph{dichotomous}, i.e., they assign value either $1$ or $0$ to the partners of the agents on the other side. By setting the values of every man and every woman besides one single (man, woman) pair $(m,w)$ to be $0$, we can focus on the social welfare of $(m,w)$, and construct our distortion lower bound argument in relation to it.

Given the dichotomous nature of the utilities, for each agent in $(m,w)$, there is a \emph{transition point} on his/her preference list, where his/her utility for the agents on the other side transitions from $1$ to $0$. The goal of the algorithm will essentially be to identify this transition point using a given number of queries, and the goal of the adversary will be to ``hide'' this point until the algorithm runs out of queries. Our lower bound boils down to arguing (a) that the algorithm does not have enough queries to identify any of the two transition points (for $m$ and for $w$), and (b) that the (unidentified) transition points are in fact possible stable partners for $m$ and for $w$. Establishing this, if the algorithm chooses to output a matching that matches $m$ and $w$ with those partners, we can set the corresponding values  of $m$ and $w$ for them to $0$ and obtain infinite distortion. If the algorithm does not match them with these partners, we can set these values to $1$ (for a social welfare of $2$), and argue that no other matching achieves a welfare better than $1$. 

For the case of non-adaptive queries, showing (a) and (b) above is relatively easy: the queries are known in advance and hence we can fix the answers of the adversary to the queries in advance as well. For adaptive queries, this is more challenging: we need to carefully choose the adversary's answers based on the queries, which a stable matching algorithm could ask at any position of the preference list of an agent. We make sure that the adversary responds in a way that always maximizes the length of the \emph{uninformed region}, i.e., informally, the region of unknown values after a number of queries, which contains the transition point. This effectively forces the algorithm to, at best, perform binary search over the preference list, looking for the transition point. We refer the reader to \cref{sec:more-queries-impossibility} (and \cref{fig:proof-of-lower-bound} therein) for the details.  

\paragraph{Achieving optimal distortion/query complexity tradeoffs.}
As we mentioned above, we manage to completely characterize the asymptotic number of queries that are necessary and sufficient to obtain a distortion of $1+\varepsilon$, for any $\varepsilon > 0$, as $\Theta(\log (1+n\varepsilon)/\varepsilon)$; these results are presented in \cref{sec:more-queries-positive}. We discuss the main ideas behind the lower and the upper bound of this tight bound below. \medskip

\noindent \textbf{The lower bound.} The lower bound directly extends the ideas of our aforementioned lower bound construction in \cref{sec:more-queries-impossibility}. In particular, it uses the same cyclic-shift profile, and consistent valuation profiles that assign non-zero value to only one man and one woman. To generalize the bound from $2$ to $1+\varepsilon$, the construction is adapted to use $b$ `\emph{bivalued} blocks, rather than two, as in the case of our previous bound, where $b$ is an appropriate function of $\varepsilon$. As before, the adversary ``forces'' the algorithm to essentially perform binary search within each such block, to identify a transition point, as to avoid returning a stable matching of low social welfare. \medskip 

\noindent \textbf{The upper bound.} To prove the upper bound, we design an algorithm (see \cref{alg:recursive-astsf}), using as a starting point an approach proposed by \citet{amanatidis2021peeking, amanatidis2022few} to design query-enhanced algorithms for voting and (non-stable) matching using binary search. The high-level idea of this approach is the following: given an agent $a_i$ and a threshold $\delta$, the algorithm uses binary search to find the last choice of an agent for which the agent has value $\delta \cdot v_{a_i}^*$, where $v_{a_i}^*$ is the value of the agent for its top choice. The algorithm then uses $\lambda$ such \emph{independent} binary searches, effectively partitioning the value space of the agent into $\lambda+1$ ``buckets'', where the value of the agent for a choice in a bucket lies between the upper and lower thresholds (possibly $0$ for the last bucket) that define the bucket. Then, it creates a \emph{simulated} valuation function which sets the value of each choice in each bucket to the lower threshold of the bucket, and then finds an outcome that maximizes the social welfare based on the simulated valuations.

Our approach will be similar, but with some key differences. Our algorithm will differ from those in \citep{amanatidis2021peeking, amanatidis2022few} in two main regards. First, instead of using the whole preference ranking of an agent, we will only use the restriction of the preference ranking to stable partners. This is necessary, as otherwise the simulated social welfare of our algorithm might not correspond to a stable matching. This is also sufficient because only stable pairs can appear in any stable matching. Secondly, we will construct the buckets using \emph{additive} rather than \emph{multiplicative} thresholds; the latter were used in the aforementioned papers. Additive thresholds are better suited to our analysis because the welfare loss depends on the absolute error in each agent’s value, measured relative to their favorite stable partner. Equally spaced thresholds control this error uniformly. Multiplicative thresholds instead become increasingly close together at smaller values, providing more precision than the welfare guarantee requires. Additive thresholds therefore achieve the same distortion bound with fewer buckets, reducing the number of queries needed to construct them. 

Indeed, as we show in \cref{app:simple-stableTSF}, a version of our algorithm that employs multiplicative thresholds can be proven to achieve distortion $1+\varepsilon$ with $O(\log n / \varepsilon^2)$ queries per agent. In contrast, using additive thresholds while keeping the algorithm and its analysis virtually the same, we obtain an improved bound of $O(\log n / \varepsilon)$ in the query complexity. This is an asymptotic improvement, but does not quite match the lower bound of $\Omega(\log (1+n\varepsilon)/\varepsilon)$ of \cref{infthm:query-enhanced-dist-1pluseps}. To achieve the tight upper bound, we further modify the algorithm to use a more refined search procedure to construct the buckets, rather than a number of independent binary searches. In particular, the procedure maintains a set of intervals (consisting of stable partners of the queried agent) with the invariant that the endpoints of those intervals have already been placed in the appropriate bucket. Then, it uses a binary-search type process within each interval. Notice that whenever both endpoints of an interval belong to the same bucket, every partner between them must also belong to that bucket, so no further queries are needed there. In a nusthsell, while independent binary searches locate each bucket boundary separately, even though a query reveals a partner’s value and therefore determines which bucket that partner belongs to, the refined search procedure uses this information across all bucket boundaries simultaneously.   

\subsection{Further Related Work and Discussion}

The study of the distortion in social choice dates back to the work of \citet{procaccia2006distortion}, and then later \citep{caragiannis2011voting,boutilier2015optimal,caragiannis2017subset}; we refer the reader to the survey of \citet{anshelevich2021distortion} for a more detailed exposition of the first line of works in the area. These early works considered solely (deterministic or randomized) ordinal algorithms for single- or multi-winner voting and proved upper and lower bounds on their distortion. Over the past half-decade, following the work of \citet{amanatidis2021peeking}, the focus has partially shifted to the study of query-enhanced algorithms, with the goal of achieving good tradeoffs between the distortion and the amount of elicited information via these queries, see also \citep{amanatidis2022few,amanatidis2024dont,caragiannis2024beyond,ebadian2025every}. The performance of such algorithms has been investigated in matching settings, but without the stability condition, which makes the problem markedly different. 

In the context of stable matchings, cardinal queries could be further motivated by the concept of \emph{interviews}; an interview might reveal important information that will help employers and employees uncover the intensity of their preferences, but, given their costly logistics, the number of interviews should be kept to a minimum. The stable matching problem with interviews has been considered in the literature e.g., see \citep{babaioff2026efficient,ashlagi2025stable,rastegari2016preference,rastegari2013two,drummond2014preference}, but only for settings \emph{without cardinal values}, where the interview is meant to reveal part of the \emph{ordinal} preference rankings of the agents. Similar settings with oracle access to the agents' ordinal rankings have also been considered in the context of query and communication complexity, see \citep{ng1990lower,segal2007communication,gonczarowski2019stable}. 

As we explained earlier, the natural benchmark for studying the distortion of stable matching algorithms is the social welfare of the best stable matching, rather than the social welfare of the best matching, which might not be stable. The ratio of the latter two quantities\footnote{Somewhat ambiguously, these works refer to the ratio of the social welfare of the optimal matching over that of any (resp. the best) stable matching as the \emph{Price of Anarchy} (resp. the \emph{Price of Stability}), two notions that are typically used in the context of equilibrium performance in strategic games \citep{koutsoupias1999worst,anshelevich2013anarchy}.} was studied first by \citet{anshelevich2013anarchy}, who also observed what we also remark in \cref{footnote1} above (and, formally, in \cref{thm:opt-matching-benchmark-infinity}), namely that this ratio can be unbounded in the worst case. The authors then proceeded to consider special cases such as symmetric values, for which constant approximation guarantees are possible. In a similar vein, \citet{emek2015price} considered stable matchings with metric costs, and proved tight bounds on the aforementioned ratio. We remark that in the case of symmetric preferences like those studied in \citep{emek2015price}, when the preferences are strict, the stable matching is unique (e.g., see \citep{arkin2009geometric}); hence distortion investigations (with respect to our benchmark) in those settings are not meaningful.  

Finally, we note that the stable matching literature in computer science is very extensive, with most works being concerned with obtaining efficient algorithms or computational hardness results for finding stable matchings in several variants of the problem. We refer the interested reader to the classic textbooks of \citet{gusfield1989stable} and \citet{manlove2013algorithmics} for more details. 
\section{Preliminaries}\label{sec:prelims}
We consider the \emph{stable matching problem}\footnote{The problem is often referred to as the \emph{stable marriage problem}, to differentiate from other variants where the matching is not necessarily on a bipartite graph. We do not consider any of those variants in our work, so we use the term stable matching instead.}, in which there are two disjoint sets of {\em agents}; we refer to these sets as {\em men}, denoted by $M$, and {\em women}, denoted by $W$, with $|M| = |W| = n$. Each man $m_i \in M$ has a valuation function $v_{m_i}:W \rightarrow \mathbb{R}$ assigning non-negative values to different women. Respectively, each woman $w_i \in W$ has a valuation function $v_{w_i}:M \rightarrow \mathbb{R}$ assigning non-negative values to different men. Each valuation function $v_{m_i}$ can also be interpreted as a \emph{valuation vector} $\mathbf{v}_{m_i} = (v_{m_i}(w_1), \ldots, v_{m_i}(w_n))$, (and similarly for $v_{w_i}$). We will let $\bv_M = (\bv_{m_1}, \ldots, \bv_{m_n})$ be a \emph{men's valuation profile}, $\bv_W = (\bv_{w_1}, \ldots, \bv_{w_n})$ be a \emph{women's valuation profile}, and $\bv = (\bv_M,\bv_W)$ be a valuation profile, consisting of the valuations of both men and women. 

These valuations induce \emph{(ordinal) preference rankings}, or simply \emph{preferences}, of the agents on one side for the agents on the other. We will let $\succ_{m_i}$ (respectively $\succ_{w_i}$) denote the preference ranking of a man $m_i \in M$ (respectively a woman $w_i \in W$) for the women in $W$ (respectively the men in $M$). Intuitively, one can think of $\succ_{m_i}$ as a permutation of the elements of the set $W$. Let $$\succ_M = (\succ_{m_1}, \ldots, \succ_{m_n}) \ \ \text{ and } \ \ \succ_W = (\succ_{w_1}, \ldots, \succ_{w_n})$$ be a \emph{men's preference profile} and  a \emph{women's preference profile}, respectively. Also, let $\succ = (\succ_M,\succ_W)$ be a preference profile, consisting of the preferences of both men and women, and let $(\succ)^n$ be the set of all preference profiles. To emphasize that a given preference profile is induced by a valuation profile $\bv$, we will write $\succ_\bv$ (and $\succ_{M,\bv}$ and $\succ_{W,\bv}$ for the men's and women's preference profiles, respectively). 

\begin{remark}
Following the standard convention in the literature of distortion, we will not require the valuation functions to be bijective, i.e., a man (resp. a woman) can have the same value for multiple women (resp. men). Still, again following the literature, we will assume that the preferences induced by those functions are strict, by applying an arbitrary tie-breaking rule when constructing preferences consistent with the valuation functions. Our positive results will hold for any such tie-breaking rule. The negative results will be presented with ties in the valuations. This is merely for presentational convenience; all the negative results can be modified to use instances that have distinct values.\footnote{The literature on stable matching sometimes considers preference rankings with ties, and the corresponding notions of stability, namely superstability, strong stability, and weak stability, e.g., see \citep[Chapter 1.4.3]{gusfield1989stable}. In the context of our work, the former two are arguably less meaningful, as they are not guaranteed to exist. All of our results extend to the case of preference rankings with ties and weakly stable matchings.}
\end{remark}

\paragraph{Terminology and Notation.} Sometimes it will be more convenient to refer to some agent without specifying if that agent is a man or a woman. In that case, we will use $a_i$ to denote the agent and $A \in \{M,W\}$ to denote the set that the agent belongs to. We will be referring to the other set as $\bar{A} = M\cup W \setminus A$. Also, when referring to agents generically, we will use ``it'', rather than the pronouns ``he'' and ``she'' that we will use for men and women, respectively. Additionally, given the preference rankings of the agents, it will make sense to give names to some designated agents.

\begin{definition}[Favorites and Suitors]
Consider any agent $a_i \in A$ with preference ranking $\succ_{a_i}$. We will call an agent $a_j \in \bar{A}$ the \emph{favorite} of agent $a_i$, and we will denote it by $f(a_i)$, if $a_j$ appears first in agent $a_i$'s ranking, i.e., $a_j \succ_{a_i} a_{j'}$, for all $a_{j'} \in \bar{A}\setminus\{a_{j}\}$. We will call an agent $a_j \in \bar{A}$ the \emph{suitor} of agent $a_i$, and we will denote it by $s(a_i)$, if $a_i$ appears first in agent $a_j$'s ranking, i.e., $a_i \succ_{a_j} a_{i'}$, for all $a_{i'} \in A\setminus\{a_{i}\}$. 
\end{definition}

\paragraph{Matchings and Stable Matchings.} An outcome is a (one-to-one) {\em matching} $\mu$ between the men and the women; we denote by $\mu(m_i)$ the woman that is assigned to man $m_i$, and by $\mu(w_i)$ the man that is assigned to woman $w_i$. Let $\mathcal{M}(\bv)$ be the set of all possible matchings with $n$ men and $n$ women, when the valuation profile is $\bv$; when $\bv$ is clear from the context, we will simply write $\mathcal{M}$ for the set of matchings. The social welfare of a matching $\mu$ for profile $\bv$ is 
\begin{align*}
    \SW(\mu\mid\bv) = \SW_M(\mu\mid\bv) + \SW_W(\mu\mid\bv)
\end{align*}
where $\SW_M(\mu\mid\bv) = \sum_{m_i \in M} u_{m_i}(\mu \mid \bv)$ and $\SW_W(\mu\mid\bv) = \sum_{w_i \in W} u_{w_i}(\mu \mid \bv)$.
We will be interested in matchings that are \emph{stable}, i.e., matchings that are robust to unilateral deviations by pairs of men and women. To define a matching formally, we first define the notion of a blocking pair.
\begin{definition}[Blocking Pair]
Let $\mu \in \mathcal{M}$ be a matching. We will say that a (man,woman) pair $(m_i,w_j)$ is a \emph{blocking pair} for $\mu$, if both of the following are true:
\begin{itemize}
    \item[-] Man $m_i$ prefers $w_j$ to the woman he is matched with in $\mu$, i.e., $w_j \succ_{m_i} \mu(m_i)$. 
    \item[-] Woman $w_j$ prefers $m_i$ to the man she is matched with in $\mu$, i.e., $m_i \succ_{w_j} \mu(w_j)$ 
\end{itemize}
\end{definition}

\begin{definition}[Stable Matching]\label{def:stable-matching}
A matching $\mu \in \mathcal{M}$ is a \emph{stable matching}, if there does not exist any blocking pair for $\mu$. 
\end{definition}
\noindent We will let $\mathcal{M}_s(\succ_\bv) \subseteq \mathcal{M}(\bv)$ denote the set of stable matchings on the ordinal preference profile induced by the valuation profile $\bv$. Notice that since stability is uniquely defined by the preference rankings of the agents, we denote the set by $\mathcal{M}_s(\succ_\bv)$ (or simply $\mathcal{M}_s(\succ)$, when $\bv$ is clear from context) rather than $\mathcal{M}_s(\bv)$; indeed for any valuation profile $\bv$ consistent with the same preference profile $\succ$, the set of stable matchings is the same.  Like before, we will simply write $\mathcal{M}_s$ when $\succ$ is clear from context.\medskip

\noindent We next present the definition of an \emph{optimal stable matching}. 

\begin{definition}[Optimal Stable Matching]\label{def:optimal-stable-matching}
A matching $\mu^*$ is an \emph{optimal stable matching} if it is stable, and it has the highest social welfare among all stable matchings, i.e., $\mu^* \in \mathcal{M}_s(\succ_\bv)$ and $\SW(\mu^*\mid\bv) = \max_{\mu \in \mathcal{M}_s(\succ_\bv)}\SW(\mu\mid\bv)$.
\end{definition} 

\noindent Given a valuation profile $\bv$, an optimal stable matching can be computed in polynomial time \citep{irving1987efficient}; we refer to the corresponding algorithm that achieves that as $\mathcal{A}_\text{ilg}$. \medskip 

\noindent It is not true that every pair $(m_i,w_j)$ can be part of a stable matching; see for example the preference profile of \cref{fig: gale-shapley}. It will be useful to consider, for an agent $a_i \in A$, only those agents $a_j \in \bar{A}$ that $a_i$ can be matched with in some stable matching. Formally, given a preference profile $\succ$, an agent $a_i \in A$ and an agent $a_j \in \bar{A}$, the pair $(a_i,a_j)$ is a \emph{stable pair} if and only if $(a_i,a_j)$ is part of some stable matching $\mu \in \mathcal{M}_s(\succ)$. In that case, we will say that $a_i$ and $a_j$ are \emph{stable partners}. The following well-known theorem establishes that the set of all stable pairs, and hence the set of all stable partners, can be found in time $O(n^2)$.

\begin{theorem}\citep{gusfield1987three}\label{lem:all-stable-pairs}
Given a preference profile $\succ$, there is a polynomial-time algorithm $\mathcal{A}_{\text{gus}}$ that finds the set of all stable pairs in time $O(n^2)$.
\end{theorem} 

\noindent For an agent $a_i \in A$, will let $S_a \subseteq \bar{A}$ denote the set of its stable partners, and we will use $\succ_{a_i}^s$ to denote the preference ranking of the agent over only the elements of $S_a$. We will also use $f^s(a_i)$ to denote the \emph{stable favorite} of agent $a_i$, i.e., an agent $a_j \in \bar{A}$ such that $a_j \succ_{a_i} a_{j'}$, for all $a_j' \in S_a \setminus \{a_j\}$.\\

\noindent It will also be useful to consider matchings that are optimal for one of the two sides only.
\begin{definition}[Man-Optimal and Woman-Optimal Stable Matching]\label{def:man-optimal-woman-optimal}
Consider a matching $\mu$.
\begin{itemize}
    \item[-] We will say that $\mu$ is a \emph{man-optimal} matching (and we will denote it by $\mu_M^*$) if it is stable, and, for every man $m_i$, it holds that $\mu_M^*(m_i) \succ_{m_i} w$, for all $w \in S_{m_i}$.
    \item[-] We will say that $\mu$ is a \emph{woman-optimal} matching (and we will denote it by $\mu_W^*$) if it is stable, and, for every woman $w_i$, it holds that $\mu_W^*(w_i) \succ_{w_i} m$, for all $m \in S_{w_i}$.
\end{itemize}
It follows by \cref{def:man-optimal-woman-optimal} that a man-optimal (resp. woman-optimal) stable matching maximizes the social welfare of the men (resp. women). Man-optimal and woman-optimal stable matchings are in fact stronger, as they are best among all stable matchings for all agents of the corresponding side \emph{simultaneously}. Such matchings exist, and can be computed by the Deferred Acceptance algorithm of \citet{gale1962college}.
\end{definition}

\paragraph{Randomized Stable Matchings.} We will also be interested in randomized stable matchings, which match each man $m_i$ with each woman $w_i$ with probability $p_{m_i,w_i} \in [0,1]$; a deterministic matching is then simply a randomized matching where $p_{m_i,w_i} \in \{0,1\}$. Given a randomized matching $\mu$, the \emph{expected utility} of an agent $a_i \in A$ is given by $u_{a_i}(\mu \mid \bv) = \sum_{a_j \in \bar{A}}p_{a_i,a_j} v_{a_i}(a_j)$. The social welfare of a randomized matching is defined analogously to before, namely 
\begin{align*}
    \SW(\mu\mid\bv) = \SW_M(\mu\mid\bv) + \SW_W(\mu\mid\bv)
\end{align*}
where $\SW_M(\mu\mid\bv) = \sum_{m_i \in M} u_{m_i}(\mu \mid \bv)$ and $\SW_W(\mu\mid\bv) = \sum_{w_i \in W} u_{w_i}(\mu \mid \bv)$.
Note that, by the Birkhoff-von Neumann decomposition \citep{birkhoff1946tres}, a randomized matching $\mu$ can alternatively be interpreted as a probability distribution over (at most $n^2$) deterministic matchings $\mu_d$. The set of such matchings that are outputted with positive probability is called the \emph{support} of $\mu$.

For a randomized (or \emph{fractional}) matching $\mu$, the literature has considered several different notions of stability. In this work we focus on perhaps the most natural of those, namely \emph{ex-post stability} \citep{roth1993stable}, which, just like our notion of stability in \cref{def:stable-matching}, can be defined solely based on the preference rankings $\succ$. We refer the reader to \citep{aziz2019random}, as well as \cite[Appendix A]{caragiannis2021stable} for a discussion of other notions of stability. 

\begin{definition}[(Ex-post) randomized stable matching]\label{def:randomized-stable-matching}
A randomized matching $\mu$ is an \emph{(ex-post) randomized stable matching}, if every matching $\mu_d$ in its support is a stable matching. 
\end{definition}
\noindent From now on, we will omit the ``ex-post'' part and refer to such matchings simply as ``randomized stable matchings''. Notice that, by definition, the social welfare of a randomized matching is upper bounded by the maximum social welfare of any matching in its support; this implies that within the set of randomized stable matchings as defined in \cref{def:randomized-stable-matching}, the optimal stable matching is still deterministic, just like in \cref{def:optimal-stable-matching}. We will use the notation $\Delta(\mathcal{M}_s(\succ))$ (or simply $\Delta(\mathcal{M}_s)$) to denote the set of all randomized stable matchings on preference profile $\succ$.

\subsection{Stable Matching Algorithms and Distortion}\label{subsec:ordinal-algos-prelims}

We consider algorithms that have access to the (ordinal) preference rankings of the agents, and elicit (limited) cardinal information about the valuation functions of the agents via \emph{queries}. For this, we employ the query model of \citet{amanatidis2021peeking}, in which the algorithm, on top of the ordinal information, is equipped with a \emph{query oracle} $\mathcal{Q}$ which it can use to ask the agents a set of \emph{value queries}. A value query for an agent $a_i \in A$ inputs an ordinal preference profile $\succ$, an agent $a_j \in \bar{A}$, and possibly the answers (of all the agents) to the previous queries, and returns the value $v_{a_i}(a_j)$ of agent $a_i$ for agent $a_j$; we refer to these queries as \emph{adaptive} queries. A special case is that of \emph{non-adaptive} queries, in which the input to the query does not depend on the answers to the previous queries by the agents. Obviously, regardless of its type, an answer to a query must be consistent with the values revealed by previous queries (i.e., it must be consistent with the preference ranking of the agent). \medskip 

\noindent To define the query model formally, we first define the notion of a \emph{partial valuation profile}. 

\begin{definition}[Partial valuation profile $\bv^p$]
A partial valuation profile $\bv^p$ is a valuation profile in which only some \emph{entries} are known, and the others are unknown. Formally, the valuation functions $v_{m_i}: W \rightarrow \mathbb{R}$ and $v_{w_i}:M \rightarrow \mathbb{R}$ are not surjective, and thus define \emph{partial valuation vectors} that have fewer than $n$ elements. A partial valuation profile is defined as a vector of those partial valuation vectors. We will say that $\bv^p$ is a \emph{partial restriction} of $\bv$ if $\bv$ can be recovered from $\bv^p$ by completing its entries. 
\end{definition}

\noindent Let $q^t(\bv)$ denote the set of answers to the first $t$ queries performed by the algorithm, when the valuation profile is $\bv$, and notice that $q^t(\bv)$ is a partial valuation profile which is partial restriction of $\bv$. We define:
\begin{itemize}
    \item[-] An \emph{adaptive} query for agent $a_i \in A$ as a function $\mathcal{Q}_i$ with input $\succ_\bv$, $q^t(\bv)$, and an agent $a_j \in \bar{A}$, and output $v_{a_i}(a_j)$. 
    \item[-] A \emph{non-adaptive} query for agent $a_i \in A$ as a function $\mathcal{Q}_i$ with input $\succ_\bv$ and an agent $a_j \in \bar{A}$, and output $v_{a_i}(a_j)$.
\end{itemize}

\noindent Formally, a (deterministic) stable matching algorithm $\mathcal{A}$ is a function that inputs a preference profile $\succ$ and a query oracle $\mathcal{Q}$ and outputs a stable matching $\mu \in \mathcal{M}_s$; the number and the type of queries that the oracle is allowed to perform is implicit in the definition and will appear in the statements of the corresponding theorems. When the oracle is allowed to perform $0$ queries, we will refer to $\mathcal{A}$ as an \emph{ordinal} algorithm. For ordinal algorithms in particularly, we will also be interested in their randomized counterparts. Formally, an ordinal randomized stable matching algorithm $\mathcal{A}$ is a function $\mathcal{A}:(\succ)^n  \rightarrow \Delta(\mathcal{M}_s)$, where $(\succ)^n$ is the set of all possible preference orderings with $n$ men and $n$ women. A deterministic ordinal stable matching algorithm is an ordinal randomized stable matching algorithm whose output contains only one matching in its support. An example of a well-known deterministic ordinal stable matching algorithm is the Deferred Acceptance algorithm of \citet{gale1962college}, which, depending on the proposing side, computes either the man-optimal matching $\mu_M^*$ or the woman-optimal matching $\mu_W^*$. We present the men-proposing version of the algorithm in \cref{alg:da}.

\paragraph{Revealed Social Welfare.} A quantity that will be useful for our analyses will be the social welfare of a matching $\mu$ on $\bv$, restricted only to the answers of the queries $q^t(\bv)$. Formally, given a partial valuation profile $\bv^p$, the \emph{revealed social welfare} of a matching $\mu$ on $\bv^p$ is defined as
\begin{equation*}
    \SW(\mu\mid\bv^p) = \sum_{m_i \in M} v_{m_i}(\mu(m_i)) \cdot \mathbbm{1}_{v_{m_i}(\mu(m_i)) \in \bv^p} + \sum_{w_i \in W} v_{w_i}(\mu(w_i))\cdot \mathbbm{1}_{v_{w_i}(\mu(w_i)) \in \bv^p}
\end{equation*}
where $\mathbbm{1}$ is the indicator function which specifies for a given agent $i \in M \cup W$, whether the value for the agent's partner in $\mu$ is in the partial profile $\bv^p$ or not.

\paragraph{Distortion.} The distortion of a stable matching algorithm is defined as the worst-case ratio (over all valuation profiles $\bv$) of the maximum social welfare on $\bv$, over the (expected) social welfare of the matching $\mathcal{A}(\succ_\bv)$ outputted by the algorithm on input the preference profile $\succ_\bv$ induced by $\bv$, and with access to the query oracle $\mathcal{Q}$. As we mentioned in the Introduction, we are concerned with the best possible distortion \emph{within the set of stable matchings}, and therefore the numerator in the definition of the distortion will be the maximum social welfare of any stable matching. Formally, 
we have:
\begin{equation*}
    \text{dist}(\mathcal{A}) = \sup_{\bv,n} \frac{\max_{\mu \in \mathcal{M}_s(\succ_\bv)}\SW(\mu\mid\bv)}{\SW(\mathcal{A}(\succ_{\bv})\mid\bv)}
\end{equation*}

\begin{remark}\label{rem:unit-sum}
Typically in the literature, the (utilitarian) distortion of ordinal algorithms is studied under the standard \emph{unit-sum normalization} assumption, e.g., see \citep{anshelevich2021distortion}). This is because, without any such normalization, the values of certain agents can ``dominate'' the social welfare, and these values cannot be identified without access to any cardinal information. As a result, these imbalances often create trivial artificial lower bounds. To avoid such artificial bounds in our setting, we diligently construct our lower bounds for ordinal algorithms to only use instances that are unit-sum normalized, see also \cref{rem:general-ordinal-lower} and \cref{rem:randomized-lower}. Our upper bounds do not require any normalization and are thus stronger. We refer to \citet{aziz2020justifications} for a more philosophical justification of the unit-sum normalization. 

For query-enhanced algorithms, the standard convention is to study them without any normalization assumptions, as one cardinal query is typically enough to avoid such pathological instances, see \citep{amanatidis2021peeking,amanatidis2022few,amanatidis2024dont,caragiannis2024beyond,ebadian2025every}. We adopt the same approach in this work as well.   
\end{remark}

\begin{algorithm}[t]
\caption{\textsc{Deferred Acceptance} \citep{gale1962college} (\emph{men-proposing version})}
\label{alg:da}
\DontPrintSemicolon

\KwIn{A preference profile $\succ = (\succ_M, \succ_W)$.}
\KwOut{A stable matching $\mu \in \mathcal{M}_s(\succ)$.}

Let $e: A \rightarrow \bar{A}$ \tcp*{Engagement function}
Initialize $E_M = \emptyset$ and $E_W=\emptyset$ \tcp*{Sets of engaged agents} 
\tcp{Assign each man and woman to be initially not engaged to anyone}
Let $P_{m}$ be the set of women that man $m$ has \emph{proposed} to\\

\While{$\exists m \in M\setminus E_M$ such that $P_{m}\neq W$}{
    $P_m \gets P_m \cup \{w\}$, where $w \succ_{m_i} w'$ for all $w' \in W\setminus P_m$\\
    \tcp{While there is a man that is not engaged, and who has not proposed to every woman, the man proposes to his favorite woman among those he has not proposed to yet}
    \BlankLine
    \If{$w \notin E_W$\tcp*{If woman $w$ is not engaged}}
    {
    Let $e(m)=w$ and $e(w)=m$ \tcp*{assign $m$ and $w$ to be engaged to each other} 
    $E_M \gets E_M \cup \{m\}$\\
    $E_W \gets E_W \cup \{w\}$
    
    }
    \Else{
    \tcp{Woman $w$ is engaged to some man $m'$}
    Let $m'=e(w)$ \tcp*{$m'$ is the man that woman $w$ is currently engaged to}
    \If{$m\succ_{w} m'$}{
        Let $e(m)=w$ and $e(w)=m$ \tcp*{assign $m$ and $w$ to be engaged to each other}
        $E_M \gets E_M \cup \{m\} \setminus \{m'\}$
    }
    \Else
    {
    \tcp{Do Nothing}
    \tcp{Woman $w$ rejects the proposal of man $m$, and man $m$ remains not engaged}
    }
    }
}
\ForEach{$m \in M$}{
Let $\mu(m) = e(m)$
}
\Return $\mu$
\end{algorithm}

\noindent One might wonder what would happen if we considered the social welfare of the best (not necessarily stable) matching as our optimality benchmark, i.e., if the numerator in the definition of the distortion was instead $\max_{\mu \in \mathcal{M}(\bv)}\SW(\mu\mid\bv)$. In this case, it is fairly easy to show that no stable matching algorithm (even one that employs randomization and has full access to the cardinal valuation functions of the agents) can achieve any meaningful distortion, see also \citep{anshelevich2013anarchy}. For completeness, we present the following observation, which establishes this bound via proving that there are instances in which the best stable matching has a social welfare of $0$, whereas the best (non-stable) matching has a positive social welfare, therefore resulting in an infinite distortion.

\begin{observation}\label{thm:opt-matching-benchmark-infinity}
There exists a valuation profile $\bv$ such that 
$$\max_{\mu \in \mathcal{M}(\bv)}\SW(\mu\mid\bv)> 0 \text{ and } \max_{\mu \in \mathcal{M}_s(\succ_\bv)}\SW(\mu\mid\bv)=0.$$ 
\end{observation}

\begin{proof}
Consider the preference profile in \cref{fig: gale-shapley}; this is the preference profile in \citep[Example 2]{gale1962college}. It is not hard to verify that the only stable matching is the one indicated by the circled entries, in which no agent is matched with their top choice among the agents of the opposite side. Now, for every $i \in \{1,2,3,4\}$, let $v_{m_i}(f(m_i))=1$ and $v_{m_i}(w_j)=0$ for all $w_j \neq f(m_i)$, and likewise $v_{w_i}(f(w_i))=1$ and $v_{w_i}(m_j)=0$ for all $m_j \neq f(w_i)$. Clearly, for the single matching $\mu \in \mathcal{M}_s(\succ_\bv)$, we have that $\SW(\mu\mid\bv)=0$. At the same time, any matching $\mu' \in \mathcal{M}(\bv)$ that matches some agent to their top choice has $\SW(\mu'\mid\bv)>0$, which proves the claim. 
\end{proof}

\begin{figure}[h!]
\centering
\begin{tikzpicture}[
    scale=0.7, transform shape,
    cell/.style={
        minimum width=1.8cm,
        minimum height=1.2cm,
        align=center,
        font=\fontsize{17}{20}\selectfont
    },
    rowlbl/.style={anchor=east, font=\fontsize{15}{20}\selectfont},
    collbl/.style={anchor=south, font=\fontsize{15}{20}\selectfont},
]

\draw[
    line width=1.5pt,
    rounded corners,
    fill=white,
    drop shadow
]
  (-1.0,-1.0) rectangle (10.0,5.5);

\node[collbl] at (0,5.7) {$w_1$};
\node[collbl] at (3,5.7) {$w_2$};
\node[collbl] at (6,5.7) {$w_3$};
\node[collbl] at (9,5.7) {$w_4$};

\node[rowlbl] at (-1.5,4.5) {$m_1$};
\node[rowlbl] at (-1.5,3.0) {$m_2$};
\node[rowlbl] at (-1.5,1.5) {$m_3$};
\node[rowlbl] at (-1.5,0.0) {$m_4$};

\node[cell] (a00) at (0,4.5) {$1,3$};
\node[cell] (a01) at (3,4.5) {$2,3$};
\node[cell] (a02) at (6,4.5) {$3,2$};
\node[cell] (a03) at (9,4.5) {$4,3$};

\node[cell] (a10) at (0,3.0) {$1,4$};
\node[cell] (a11) at (3,3.0) {$4,1$};
\node[cell] (a12) at (6,3.0) {$3,3$};
\node[cell] (a13) at (9,3.0) {$2,2$};

\node[cell] (a20) at (0,1.5) {$2,2$};
\node[cell] (a21) at (3,1.5) {$1,4$};
\node[cell] (a22) at (6,1.5) {$3,4$};
\node[cell] (a23) at (9,1.5) {$4,1$};

\node[cell] (a30) at (0,0.0) {$4,1$};
\node[cell] (a31) at (3,0.0) {$2,2$};
\node[cell] (a32) at (6,0.0) {$3,1$};
\node[cell] (a33) at (9,0.0) {$1,4$};

\draw[blue!60, line width=2.5pt] (a02) circle (0.7);
\draw[blue!60, line width=2.5pt] (a13) circle (0.7);
\draw[blue!60, line width=2.5pt] (a20) circle (0.7);
\draw[blue!60, line width=2.5pt] (a31) circle (0.7);

\end{tikzpicture}

\caption{An example of a preference profile in which no agent is matched with their top choice, taken from \citep[Example 2]{gale1962college}. The entry $(m_i, w_j)$ consists of a pair of numbers; the first indicates the rank of $w_j$ for $m_i$ and the second indicates the rank of $m_i$ for $w_j$. For example, the entry $(m_2,w_3)$ is $(3,3)$ which indicates that $m_2$ ranks $w_3$ third among the women, and $w_3$ ranks $m_2$ third among the men. The unique stable matching is indicated by the entries circled in blue.}
\label{fig: gale-shapley}
\end{figure}

\subsection{The Structure of Stable Matchings}\label{app:background-stable-matchings}
In this section, we review the well-established theory of \textit{partially ordered sets} and explain how it relates to the stable matching problem. We first provide a series of key definitions:

\begin{definition}[Poset]\label{def:poset}
A {\em partially ordered set} ( or {\em poset} ) $\mathcal{P} = (V,\preceq)$ is a set $V$ equipped with a binary relation $\preceq$ on $V$ satisfying:
\begin{itemize}
  \item[-] Antisymmetry: for all distinct $u,v\in V$, if $u\prec v$ then $v\not\prec u$.
  \item[-] Transitivity: for all $u,v,w\in V$, if $u\preceq v$ and $v\preceq w$ then $u\preceq w$.
\end{itemize}
\end{definition}

\begin{definition}[Incomparable Elements]\label{def: incomp-elems}
    Given a poset $\mathcal{P} = (V,\preceq)$, two elements $x$ and $y$ in $V$ are called \emph{incomparable} based on the binary relation $\prec$ if neither $x \preceq y$ nor $y \preceq x$ holds.
\end{definition}

\begin{definition}[Antichain]\label{def: antichain}
    Given a poset $\mathcal{P} = (V,\preceq)$. An antichain $A \subseteq V$ is a subset of $V$ such that every two
    elements of $A$ are incomparable as in \cref{def: incomp-elems}.
\end{definition}

\begin{definition}[Chain]\label{def: chain}
    Given a poset $\mathcal{P} = (V,\preceq)$. A chain $C \subseteq V$ is a subset of $V$ such that every two elements of $C$ are comparable.
\end{definition}

\begin{definition}[Closed Subset]\label{def: closed-subset}
    Given a poset $\mathcal{P} = (V,\preceq)$. A subset $X \subseteq V$ is a {\em closed subset} of $\mathcal{P}$ if for every $x \in X$ and $y \in V$ such that $y \preceq x$ then $y \in X$. That is, if $x \in X$ then every predecessor $y$ of $x$ is also in $X$.
\end{definition}

\begin{definition}[Hasse Diagram]\label{def: hasse-diagram}
    Given a poset $\mathcal{P} = (V,\preceq)$ the {\em Hasse Diagram} of $\mathcal{P}$ is the directed acyclic graph (DAG) which we denote as $H(\mathcal{P}) = (V,E)$ with
    \[
    E= \{(v,w) \in V \times V \hspace{1mm} \mid \hspace{1mm} v \prec w \text{ and there is no } z \in V \text{ such that } v \prec z \prec w\}
    \]
    We will also denote the relation as $v \prec: w $ if there is no $z$ in $V$ such that $v \prec z \prec w$
\end{definition}

\begin{comment}
\begin{definition}[Saturated Chain] \label{def: saturated-chain}
    Given a poset $\mathcal{P} = (V,\preceq)$. Let positive integer $k$ and $z_i \in V$ for $i \in [k]$. A chain
    \[
    z_0 \prec z_1 \prec z_2 \prec \dots \prec z_k
    \]
    is \emph{saturated} for all $i \in [k-1]$ if it holds that $z_i \prec: z_{i+1}$.
\end{definition}

\noindent We refer the interested reader to the textbook of \cite{trotter2002combinatorics} for a more detailed exposition of partially ordered sets and their properties. \medskip 
\end{comment}

\noindent We now proceed to define a poset which has a one-to-one correspondence between its closed subsets and the stable matchings on any preference profile $\succ$. First, we define the very central notion of a \emph{rotation}, which is the basic element for constructing the binary relation of the poset.

\begin{definition}[Rotation]\label{def:rotation}
Let $k \geq 2$. A \emph{rotation} $\rho$ is an ordered list of pairs
\[
\rho = \left( (m_0, w_0), (m_1, w_1), \ldots, (m_{k-1}, w_{k-1}) \right)
\]
that are matched in some stable matching $\mu$ with the property that for every $i$ such that $0 \leq i \leq k-1$, woman $w_{i+1\Mod k}$ is the \textit{highest} ranked woman on $m_i$’s preference list satisfying:
\begin{itemize}
    \item[-] man $m_i$ prefers $w_i$ to $w_{i+1\Mod k}$, and
    \item[-] woman $w_{i+1\Mod k}$ prefers $m_i$ to $m_{i+1}$.
\end{itemize}

In this case, we say $\rho$ is {\em exposed} in $\mu$. We will also refer to this woman as $\text{succ}_\mu(m_i)$.
\end{definition}

\noindent Next, we present two known properties of rotations, via two lemmas proved in \citep{irving1986complexity} and \citep{gusfield1989stable}.  

\begin{lemma}[\cite{irving1986complexity}]\label{lem: pair-in-one-rotation}
    A pair $(m, w)$ can appear in at most one rotation.
\end{lemma}

\begin{lemma}[\citep{gusfield1989stable}]
     If $\rho$ is a rotation with consecutive pairs $(m_i, w_i)$, and $(m_{i+1}, w_{i+1})$, and $w$ is a woman between $w_i$ and $w_{i+1}$ in $m_i$’s preference list, then there is no stable matching containing the pair $(m_i, w)$.
\end{lemma}

\noindent Next we define the \emph{elimination} of a rotation.
\begin{definition}[Elimination of a Rotation]\label{def:elimination-of-rotation}
Let $\rho = \left((m_0, w_0), \ldots, (m_{k-1}, w_{k-1})\right)$ be a rotation exposed in stable matching $M$. The rotation $\rho$ is \emph{eliminated from} $M$ by matching $m_i$ to $w_{(i+1) \bmod k}$ for all $0 \leq i \leq k - 1$, leaving all other pairs in $M$ unchanged, i.e., matching $M$ is replaced with matching $M'$, where
\[
M' := M \setminus \rho \, \cup \, \{(m_0, w_1), (m_1, w_2), \ldots, (m_{k-1}, w_0)\}.
\]
Note that when we eliminate a rotation from $M$, the resulting matching $M'$ is stable and in the following sections we simply denote the elimination as $M' := M \setminus \rho$. 
\end{definition}
\noindent To informally see why $M'$ is stable, notice that when switching from $M$ to $M'$, all the partners of the women in $\rho$ are better off, and all the partners of the men in $\rho$ are worse off. It is easy to check that this switch cannot create a blocking pair in the set $\rho$. The only other possibility is for a blocking pair to involve a man in $\rho$ and a woman outside $\rho$. For $(m_i \in \rho, w \notin \rho)$ to become a blocking pair, $m_i$ would have to prefer $w$ to $w_{i+1}$, but by the definition of rotation, $w_{i+1}$ was the first woman on $m$’s list who would prefer to be matched to him, so he cannot prefer $w$ to $w_{i+1}$. See also \citep[Lemma 2.5.2]{gusfield1989stable}. \medskip

\noindent In fact, every stable matching $\mu$ can be obtained by starting from the man-optimal stable matching $\mu_M^*=\mu_0$, and then eliminating a sequence of rotations. 

\begin{lemma}[\cite{irving1986complexity}]\label{lemma:unique-rotation-subset}
    Let $\mu$ be a stable matching. Then there exists a sequence of rotations $\rho_0, \rho_1,\ldots, \rho_{k-1}$ and stable matchings $\mu_0, \mu_1,\ldots, \mu_k = \mu$ such that $\mu_0$ is the man-optimal stable matching and, for each $i$, $\rho_i$ is a rotation exposed in $\mu_i$, and $\mu_{i+1} = \mu_i \setminus \rho_i$. Moreover, the set $\{\rho_0, \rho_1, \ldots, \rho_{k-1}\}$ uniquely specifies $\mu$.    
\end{lemma}

\noindent We denote by $R(\succ)$ the set of all rotations exposed in any stable matching on the preference profile $\succ$. \cref{lemma:unique-rotation-subset} shows that every stable matching $\mu$ on $\succ$ can be associated with a unique subset $R_{\mu} \subseteq R(\succ)$ consisting of rotations that must be eliminated to obtain $\mu$ from $\mu_0$ (the man-optimal matching $\mu_M^*$). For example, $\mu_0$ corresponds to the empty set, while the woman-optimal stable matching, $\mu_z=\mu_W^*$, corresponds to $R(\succ)$. Not all subsets of $R(\succ)$, however, correspond to some stable matching on $\succ$. In particular, a rotation $\pi$ is said to be an explicit predecessor of a rotation $\rho = (m_0, w_0), \ldots, (m_{k-1}, w_{k-1})$ if there exists an index $i$ with $0 \leq i \leq k-1$ and a woman $x \neq w_i$ such that $\pi$ is the eliminating rotation for the pair $(m_i, x)$ and $m_i$ prefers $x$ to $w_{i+1}$. In particular, a rotation cannot become exposed until all of its explicit predecessors have been eliminated. In particular, \citet{irving1986complexity} defined a poset structure on $R(\succ)$ to characterize the subsets of $R(\succ)$ that correspond to a stable matching as follows:

\begin{definition}[Rotation Precedence Relation]\label{def:rotation-precedes}
Suppose $\rho, \rho' \in R(\succ)$. We say $\rho$ \emph{precedes} $\rho'$ and write $\rho \prec \rho'$ if for every stable matching $\mu$ in which $\rho'$ is exposed, we have $\rho \in R_\mu$. That is, $\rho \prec \rho'$ if $\rho$ was eliminated in every stable matching in which $\rho'$ is exposed.
\end{definition}

\noindent Thus, in the following sections we will make use of the below definition. 

\begin{definition}[Stable Matching induced by a Closed Subset]\label{def:matching-from-R}
For every closed subset $R\subseteq R(\succ)$, let $\mu(R)$ denote the stable matching obtained from the man-optimal stable matching $\mu_0$ by eliminating the rotations in $R$ following \cref{thm:rotation-poset}.
\end{definition}

\noindent Furthermore, we make use of the following consequences of the rotation precedence relation of \cref{def:rotation-precedes}. \cref{prop:shared-agent-precedence} is known from the literature, see \cite{gusfield1989stable}, though we include proof in \cref{app: appendix-structure-stable-matching} for completeness.
\begin{proposition}[]
\label{prop:shared-agent-precedence}
Let $\rho,\rho'\in\mathcal R(\succ)$ be two distinct rotations. If there exists an agent that appears in both $\rho$ and $\rho'$, then $\rho$ and $\rho'$ are comparable in the rotation poset. That is,
$\rho\prec\rho'$ or $\rho'\prec\rho$.
\end{proposition}

\noindent Also we can infer the following corollary.

\begin{corollary} \label{lem:fixed-agent-chain} For every man $m\in M$, the set $\mathcal R(m):=\{\rho\in\mathcal R(\succ):m\text{ appears in }\rho\}$ is a chain in the rotation poset. The same holds for every woman $w\in W$ and the set $\mathcal R(w):=\{\rho\in\mathcal R(\succ):w\text{ appears in }\rho\}.$
\end{corollary}

Having established the binary relation from \cref{def:rotation-precedes}, \citet{irving1986complexity} proved the existence of a partially ordered set, whose closed subsets are in bijection with the stable matchings induced in a preference profile. Namely,
\begin{theorem}[\cite{irving1986complexity}]\label{thm:rotation-poset}
Let $\succ$ be a preference profile. Then $\mathcal{R} = (R(\succ), \preceq)$ is a poset. Moreover, there is a one-to-one correspondence between the stable matchings on $\succ$ and the closed subsets of $\mathcal{R}$.
\end{theorem}

\noindent We refer to $\mathcal{R} =(R(\succ), \preceq)$ as the \emph{rotation poset} of the preference profile $\succ$. In addition, we have the following characterization related with the antichains induced in $\mathcal{R}$. Namely, consider an antichain $A \subseteq R(\succ)$ and define the set of rotations $A^{\star}=\{\pi \in R(\succ): \pi \preceq \rho \text{ for } \rho \in A\}$. As in \cite{irving1986complexity}, for any closed subset $X \subseteq R(\succ)$ there is a unique antichain $A\subseteq R(\succ)$ such that $A^{\star}=X$. Namely, it holds that $A = \max(C)$, where $\max(C)$ denotes the maximal elements of $C$, i.e. elements $\rho \in \max(C)$ for which there are no elements $\pi \in R(\succ)$ such that $\pi \succ \rho$. Therefore, we have the following bijection between the stable matchings of $\succ$ and the antichains of $\mathcal{R}$.
\begin{theorem}[\cite{irving1986complexity}]\label{thm:rotation-poset-antichains}
Let $\succ$ be a preference profile. There is a one-to-one correspondence between the stable matchings of the profile $\succ$ and the antichains of its rotation poset $\mathcal{R}$.
\end{theorem}

Lastly, consider $P=(V,E)$ the directed acyclic graph (DAG) whose node set $V$ consists of one node for every rotation in $R(\succ)$ and a directed edge $(\pi, \rho) \in E$ from rotation $\pi$ to rotation $\rho$ if and only if $\pi \prec \rho$ as in \cref{def:rotation-precedes}. \cite{irving1987efficient} showed that a DAG whose transitive closure\footnote{The transitive closure of a DAG $G=(V,E)$ is another DAG $H=(V,F)$ such that $(u,v) \in F$ if and only if there is a path from $u$ to $v$ in G.} is $P$ can be computed in time $O(n^3)$. This runtime was improved to $O(n^2)$ by \citet{gusfield1987three}. Namely,
\begin{theorem}[\citet{gusfield1987three}]\label{thm:rotation-digraph}
    Let $\succ$ be a preference profile. A directed acyclic graph (DAG) $G(\succ)$ whose transitive closure is $P$ can be computed in $O(n^2)$ time.
\end{theorem}
\noindent We refer to $G(\succ)$ as the \emph{rotation digraph}. Since the transitive closure of $G(\succ)$ is $P$, we note that the Hasse diagram as in \cref{def: hasse-diagram} of the rotation poset $\mathcal{R}$ is in fact a subgraph of $G(\succ)$. In particular, the Hasse diagram is a transitive reduction\footnote{The transitive reduction of a DAG $G=(V,E)$ is another DAG $H=(V,F)$ such that $(u,v) \in F$ if and only if $(u,v) \in E$ and there is no directed path from $u$ to $v$ which does not include the edge $(u,v)$.} of $G(\succ)$.

\section{Warm-up: The Distortion of (Purely) Ordinal Algorithms}\label{sec:ordinal}
Before we consider our main question of the best possible tradeoffs between the number of queries per agent and the achievable distortion, we first investigate the performance of stable matching algorithms that perform $0$ queries, i.e., purely ordinal algorithms. It is not very hard to see that any such deterministic algorithm is bound to have unbounded distortion, even when the valuations of the agents are unit-sum normalized (see \cref{rem:unit-sum}). This is captured by the following theorem. 

\begin{theorem}\label{thm:any-ordinal-distortion-unbounded}
    Let $\mathcal{A}$ be any deterministic ordinal stable algorithm. Then the distortion of $\mathcal{A}$ is unbounded. 
\end{theorem}
\begin{figure}[h!]
\centering
\begin{tikzpicture}[
    scale=0.7, transform shape,
    cell/.style={
        minimum width=1.8cm,
        minimum height=1.2cm,
        align=center,
        font=\fontsize{17}{20}\selectfont
    },
    rowlbl/.style={anchor=east, font=\fontsize{15}{20}\selectfont},
    collbl/.style={anchor=south, font=\fontsize{15}{20}\selectfont},
]

\draw[
    line width=1.5pt,
    rounded corners,
    fill=white,
    drop shadow
]
  (-1.3,-0.9) rectangle (10.3,5.4);

\node[collbl] at (0,5.5) {$w_1$};
\node[collbl] at (3,5.5) {$w_2$};
\node[collbl] at (6,5.5) {$w_3$};
\node[collbl] at (9,5.5) {$w_4$};

\node[rowlbl] at (-1.8,4.5) {$m_1$};
\node[rowlbl] at (-1.8,3.0) {$m_2$};
\node[rowlbl] at (-1.8,1.5) {$m_3$};
\node[rowlbl] at (-1.8,0.0) {$m_4$};

\node[cell] (a00) at (0,4.5) {\Large $2,2$};
\node[cell] (a01) at (3,4.5) {\Large $1,4$};
\node[cell] (a02) at (6,4.5) {\Large $4,1$};
\node[cell] (a03) at (9,4.5) {\Large $3,1$};

\node[cell] (a10) at (0,3.0) {\Large $4,1$};
\node[cell] (a11) at (3,3.0) {\Large $3,2$};
\node[cell] (a12) at (6,3.0) {\Large $1,3$};
\node[cell] (a13) at (9,3.0) {\Large $2,3$};

\node[cell] (a20) at (0,1.5) {\Large $1,3$};
\node[cell] (a21) at (3,1.5) {\Large $4,1$};
\node[cell] (a22) at (6,1.5) {\Large $2,2$};
\node[cell] (a23) at (9,1.5) {\Large $3,4$};

\node[cell] (a30) at (0,0.0) {\Large $3,4$};
\node[cell] (a31) at (3,0.0) {\Large $2,3$};
\node[cell] (a32) at (6,0.0) {\Large $1,4$};
\node[cell] (a33) at (9,0.0) {\Large $4,2$};

\draw[blue!60, line width=2pt] (a00) circle (0.7);
\draw[blue!60, line width=2pt] (a13) circle (0.7);
\draw[blue!60, line width=2pt] (a22) circle (0.7);
\draw[blue!60, line width=2pt] (a31) circle (0.7);

\draw[red!60, line width=2pt, dashed] (a00) circle (0.6);
\draw[red!60, line width=2pt, dashed] (a11) circle (0.6);
\draw[red!60, line width=2pt, dashed] (a22) circle (0.6);
\draw[red!60, line width=2pt, dashed] (a33) circle (0.6);

\end{tikzpicture}
        \caption{An example on which any ordinal stable algorithm has unbounded distortion. The entry $(m_i, w_j)$ consists of a pair of numbers; the first indicates the rank of $w_j$ for $m_i$ and the second indicates the rank of $m_i$ for $w_j$. For example, the entry $(m_2,w_3)$ is $(1,3)$ which indicates that $m_2$ ranks $w_3$ first among the women, and $w_3$ ranks $m_2$ third among the men. In this preference profile there are only two stable matchings, the man-optimal matching, consisting of entries circled in blue with solid lines, and the woman-optimal matching, consisting of entries circled in red with dashed lines.}
        \label{fig: lower-bound}
    \end{figure}
\begin{proof}
Consider the preference profile $\succ$ shown in \cref{fig: lower-bound}, in which in the entry $(m_i,w_j)$, the first number indicates the rank of $w_j$ for $m_i$ and the second number indicates the rank of $m_i$ for $w_j$. First, notice that on this profile, there are only two stable matchings, indicated by the circled entries. The entries circled by blue solid lines in fact correspond to the man-optimal matching $\mu_M^*$, and the entries circled by red dashed lines correspond to the woman-optimal matching $\mu_W^*$. It can be verified by inspection that, for any of these matchings, there are no blocking pairs, and hence they are stable. It can be verified by inspection that these are the only stable matchings on $\succ$. For completeness, we provide a brief argument that employs the induced rotation poset of the instance below: \medskip 


\noindent Let $\mu_0 = \mu_{M}^*$ be the man-optimal matching and let $\mu_{z} = \mu_{W}^*$ be the woman-optimal matching on $\succ$, which are computed by the men-proposing and women-proposing variants of the \textsc{Deferred Acceptance} algorithm, respectively. Consider the rotation (\cref{def:rotation}) $\rho = ((m_2,w_4),(m_4,w_2))$. To see why this is a rotation, observe that among the women ranked below $w_4$ in $m_2$'s preference ranking, the most preferred one is $w_2$. Additionally, we have $\mu_0(w_2)=m_4$, and since $m_2 \succ_{w_2} m_4$, it follows that $\text{succ}_{\mu_0}(m_2)=w_2$. Similarly, the women ranked below $w_2$ in $m_4$'s preference ranking are $w_1$ and then $w_4$. We have $\mu_0(w_1)=m_1$ and $m_1 \succ_{w_1} m_4$, as well as $\mu_0(w_4) = m_2$ and $m_4 \succ_{w_2} m_2$. Therefore, we have that $\text{succ}_{\mu_0}(m_4)=w_4$. \medskip

\noindent Now notice that the elimination of $\rho$ corresponds to $m_2$ and $m_4$ switching partners, which results in matching $\mu_z$. In other words, it is possible to move from the man-optimal to the woman-optimal matching by eliminating a single rotation $\rho$. Since the pairs $(m_1,w_1)$ and $(m_3,w_3)$ are in both $\mu_0$ and $\mu_z$, it follows that $\rho$ is in fact the only rotation in the rotation poset $\mathcal{R}(\succ)$, whose closed subsets are thus $\emptyset$ and $\{\rho\}$. By \cref{thm:rotation-poset}, these are in one-to-one correspondence with the stable matchings on $\succ$, specifically the man-optimal and the woman-optimal stable matching, respectively. Hence, there does not exist any other stable matching on $\succ$.  \medskip


\noindent Now consider two valuation profiles, both consistent with the preference ranking $\succ$.
\begin{itemize}[leftmargin=*]
    \item[] \textbf{Profile $\bv_1$:} For each man $m_i \in M$, we have $v_{m_i}(f(m_i)) = 1$, and $v_{m_i}(w_j) = 0$, for any woman $w_j \neq f(m_i)$. For the women, for $i \in \{1,3\}$, we have $v_{w_i}(f(w_i)) = 1$, and $v_{w_i}(m_j) = 0$, for any man $m_j \neq f(w_i)$; we also have $v_{w_2}(m_3)=v_{w_2}(m_2)=1/2$, and $v_{w_2}(m_1)=v_{w_2}(m_4)=0$, and $v_{w_4}(m_1)=v_{w_4}(m_4)=1/2$, and $v_{w_4}(m_2)=v_{w_4}(m_3)=0$. 

    In other words, every man has value $1$ for his top choice and $0$ for all other women. Women $w_1$ and $w_3$ have value $1$ for their top choice and $0$ for all other men, whereas women $w_2$ and $w_4$ have value $1/2$ for the two top choices and value $0$ for the remaining two men. 
    
    \item[] \textbf{Profile $\bv_2$:} For each woman $w_i \in W$, we have $v_{w_i}(f(w_i)) = 1$, and $v_{w_i}(m_j) = 0$, for any man $m_j \neq f(w_i)$. For the men, for $i \in \{1,3\}$, we have $v_{m_i}(f(m_i)) = 1$, and $v_{m_i}(w_j) = 0$, for any woman $w_j \neq f(m_i)$; we also have $v_{m_2}(w_1)=v_{m_2}(w_4)=1/2$, and $v_{m_2}(w_2)=v_{m_2}(w_3)=0$, and $v_{m_4}(w_4)=v_{m_4}(w_2)=1/2$, and $v_{m_4}(w_1)=v_{m_4}(w_3)=0$. 

    In other words, every woman has value $1$ for her top choice and $0$ for all other men. Men $m_1$ and $m_3$ have value $1$ for their top choice and $0$ for all other women, whereas men $m_2$ and $m_4$ have value $1/2$ for the two top choices and value $0$ for the remaining two women. 
\end{itemize}
Let $\mu$ be the stable matching outputted by $\mathcal{A}$ on $\succ$. If $\mu=\mu_M^*$, then consider the valuation profile $\bv_1$, and observe that $\SW(\mu\mid\bv_1)=0$, whereas $\SW(\mu_W^*\mid\bv_1)=1$; this implies that the distortion is unbounded. Similarly, if $\mu = \mu_W^*$, then consider the valuation profile $\bv_2$, and observe that $\SW(\mu\mid\bv_2)=0$, whereas $\SW(\mu_M^*\mid\bv_2)=1$. Again, this implies that the distortion is unbounded. Since there are no other stable matchings in $\succ$, this proves the claim. 
\end{proof}

\begin{remark}\label{rem:general-ordinal-lower}
We remark that the valuation profiles that we construct in the proof of \cref{thm:any-ordinal-distortion-unbounded} are unit-sum normalized, which is the standard normalization assumption for ordinal algorithms in the literature, see \cref{rem:unit-sum}. It is well-documented that, in the absence of such a normalization (and any cardinal information), unbounded distortion results can be obtained straightforwardly, and this is the case for our setting as well.  
\end{remark}

\noindent Since we do not have any cardinal queries at our disposal, we ask the natural follow-up question: Is it possible to achieve better distortion bounds by using randomization?. To this end, it is fairly easy to see that the simple algorithm that outputs the man-optimal and the woman-optimal matching equiprobably (by running the men-proposing and women-proposing versions of the Deferred Acceptance algorithm with probability $1/2$ each) achieves a distortion of $2$; furthermore, the algorithm is clearly randomized stable, as it randomizes between two stable matchings. We present the algorithm and the brief proof of its distortion next, see \cref{alg:rand-mowo}. \medskip

\begin{algorithm}[t]
\caption{\textsc{Random Man-Optimal or Woman-Optimal} (\RMOWO)}
\label{alg:rand-mowo}

\KwIn{A preference profile $\succ = (\succ_M, \succ_W)$.}
\KwOut{An (ex-post) randomized stable matching $\mu \in \mathcal{M}_s(\succ)$.}
\BlankLine

$\mu_M^* = \textsc{Deferred Acceptance}(\succ)$ (men proposing version)\\
$\mu_W^* = \textsc{Deferred Acceptance}(\succ)$ (women proposing version) 
\BlankLine

\Return $\mu^*_M$ with probability $1/2$ and $\mu^*_W$ with probability $1/2$

\end{algorithm}

\begin{theorem}
The distortion of the \RMOWO algorithm is at most $2$.
\end{theorem}

\begin{proof}
Let $\bv$ be any valuation profile and let $\mu^*(\bv)$ be the optimal stable matching on $\bv$, i.e., $\mu^*(\bv) \in \arg\max_{\mu \in \mathcal{M}_s}\SW(\mu\mid\bv)$. The expected welfare of the algorithm is
\[
\frac12 \cdot \SW(\mu^*_M\mid\bv) + \frac12 \cdot \SW(\mu^*_W\mid\bv) = \frac12 \cdot \left(\SW(\mu^*_M\mid\bv)+\SW(\mu^*_W\mid\bv)\right) \geq \frac12  \SW(\mu^*(\bv)\mid \bv),
\]
where the last inequality follows from the fact that in any stable matching, any man (resp. any woman) is matched with a woman (resp. a man) that is not better than his (resp. her) partner in the man-optimal (resp. woman-optimal) matching.
\end{proof}

\noindent The rather straightforward nature of the \RMOWO algorithm motivates the question of whether a more involved algorithm could achieve a better distortion bound. In the following theorem, we prove that this is not the case. The theorem establishes that in fact, the simple \RMOWO algorithm is the best possible among all randomized stable algorithms for the problem.

\begin{theorem}\label{thm:lower-bound-ordinal-randomized}
Let $\mathcal{A}$ be an ordinal randomized stable algorithm. Then $\text{dist}(\mathcal{A}) \geq 2$.
\end{theorem}

\begin{proof}
For any agent $a_i \in M \cup W$, recall the definitions of its \emph{favorite} and its \emph{suitor} from \cref{sec:prelims}. We will consider any preference profile $\succ$ with the following properties:
\begin{itemize}
    \item[-] Every agent has a different favorite.
    \item[-] Every agent has a different suitor.
    \item[-] Every man ranks his suitor second.
    \item[-] Every woman ranks her suitor last.
\end{itemize}
Notice that the second property is in fact implied by the first. A concrete preference profile that achieves the properties above is the following \emph{reverse cycle shift profile}, shown in \cref{fig:reverse-cyclic-shift-profile}: 
\begin{itemize}
    \item[-] $m_i$ has preference ranking $w_i \succ_{m_i} w_{i+1} \succ_{m_i} w_{i+2} \succ_{m_i} \ldots \succ_{m_i} w_n \succ_{m_i} w_{1} \succ_{m_i }\ldots \succ_{m_i} w_{i-1}$
    \item[-] $w_i$ has preference ranking $m_{i-1} \succ_{w_i} m_{i-2} \succ_{w_i} m_{i-3} \succ_{w_i} \ldots \succ_{w_i} m_1 \succ_{w_i} m_{n} \succ_{w_i }\ldots \succ_{w_i} m_i$
\end{itemize}

\usetikzlibrary{matrix, shadows}

\begin{figure}[t]
\centering
\begin{tikzpicture}[
    scale=0.8, transform shape,
    cell/.style={
        minimum width=1.8cm, 
        minimum height=1.0cm,
        align=center,
        font=\fontsize{13}{16}\selectfont
    },
    rowlbl/.style={anchor=east, font=\fontsize{13}{16}\selectfont},
    collbl/.style={anchor=south, font=\fontsize{13}{16}\selectfont},
]


\draw[line width=1.2pt, rounded corners]
  (-1.2, -1.8) rectangle (11.5, 4.2);

\node[collbl] at (0,    4.4) {$w_1$};
\node[collbl] at (2.2,  4.4) {$w_2$};
\node[collbl] at (4.4,  4.4) {$\dots$};
\node[collbl] at (7.0,  4.4) {$w_{n-1}$};
\node[collbl] at (10.0, 4.4) {$w_n$}; 

\node[rowlbl] at (-1.6, 3.6) {$m_1$};
\node[rowlbl] at (-1.6, 2.4) {$m_2$};
\node[rowlbl] at (-1.6, 1.2) {$\vdots$};
\node[rowlbl] at (-1.6, 0.0) {$m_{n-1}$};
\node[rowlbl] at (-1.6, -1.2) {$m_n$};

\node[cell] at (0,    3.6) {$1,n$};
\node[cell] at (2.2,  3.6) {$2,1$};
\node[cell] at (4.4,  3.6) {$\dots$};
\node[cell] at (7.0,  3.6) {$n-1,n-2$};
\node[cell] at (10.0, 3.6) {$n,n-1$};

\node[cell] at (0,    2.4) {$n,n-1$};
\node[cell] at (2.2,  2.4) {$1,n$};
\node[cell] at (4.4,  2.4) {$\dots$};
\node[cell] at (7.0,  2.4) {$n-2,n-3$};
\node[cell] at (10.0, 2.4) {$n-1,n-2$};

\node[cell] at (0,    1.2) {$\vdots$};
\node[cell] at (2.2,  1.2) {$\vdots$};
\node[cell] at (4.4,  1.2) {$\ddots$};
\node[cell] at (7.0,  1.2) {$\vdots$};
\node[cell] at (10.0, 1.2) {$\vdots$};

\node[cell] at (0,    0.0) {$3,2$};
\node[cell] at (2.2,  0.0) {$4,3$};
\node[cell] at (4.4,  0.0) {$\dots$};
\node[cell] at (7.0,  0.0) {$1,n$};
\node[cell] at (10.0, 0.0) {$2,1$};

\node[cell] at (0,   -1.2) {$2,1$};
\node[cell] at (2.2, -1.2) {$3,2$};
\node[cell] at (4.4, -1.2) {$\dots$};
\node[cell] at (7.0, -1.2) {$n,n-1$};
\node[cell] at (10.0,-1.2) {$1,n$};

\end{tikzpicture}
    \caption{The \emph{reverse cyclic shift profile} used in the proof of \cref{thm:lower-bound-ordinal-randomized}. The entry $(m_i, w_j)$ consists of a pair of numbers; the first indicates the rank of $w_j$ for $m_i$ and the second indicates the rank of $m_i$ for $w_j$. For example, the entry $(m_2,w_1)$ is $(n,n-1)$ which indicates that $m_2$ ranks $w_1$ last among the women, and $w_1$ ranks $m_2$ second to last among the men. Notice that the preference of each man $m_i$ (in the rows of the table) is formed by placing $w_i$ first, $w_{i+1}$ second, and so on, with the preference \emph{cycling around} once the index reaches $n$. The preference of each woman is formed similarly, but in a shifted and reverse way: woman $w_i$ places $m_{i-1}$ first, $m_{i-2}$ second and so on, with preference cycling around once the index reaches $1$.} 
    \label{fig:reverse-cyclic-shift-profile}
\end{figure}

\noindent We will construct a valuation profile $\bv$ consistent with $\succ$, which will feature only two types of agents, which we refer to as either \emph{selective} or \emph{indifferent}. Specifically,

\begin{itemize}
\item[-] A selective agent $a_i \in A$ has value $1$ for its favorite, and $0$ for all the other agents, i.e., $v_{a_i}(f(a_i))=1$, and $v_{a_i}(a_j)=0$ for all $a_j \in \bar{A} \setminus \{f(a_i)\}$.
\item[-] An indifferent agent $a_i \in A$ has value $1/n$ for all of the agents on the other side, i.e., $v_{a_i}(a_j)=1/n$ for all $a_j \in \bar{A}$.
\end{itemize}
\noindent Whether an agent $a_i$ will be selective or indifferent will depend on the probabilities assigned by $\mu$ to certain pairs in which $a_i$ is part of; this is possible, since $\mathcal{A}$ is ordinal, and hence these probabilities are unchanged for any consistent underlying valuation profile. 

More specifically, consider any man $m_i \in M$; $m_i$ will be selective if and only if $m_i$ is matched with $f(m_i)$ with probability at most $1/2$, i.e., if $p_{m_i,f(m_i)} \leq 1/2$. Otherwise, the man will be indifferent. If man $m_i$ is selective, then his suitor $s(m_i)$ will be indifferent, and if he is indifferent, then $s(m_i)$ will be selective. Notice that, by definition of $\bv$ and the preference profile $\succ$, the following properties hold:
\begin{condition}
    \item[-] A selective man is matched with his favorite with probability at most $1/2$. \label{enum:property1}
    \item[-] A selective woman is matched with her favorite with probability at most $1/2$.\label{enum:property2}
    \item[-] Any selective agent $a_i$ is matched with its favorite $f(a_i)$ with probability at most $1/2$.\label{enum:property3}
    \item[-] For any $i$, either $m_i$ or $s(m_i)$ will be selective, but not both. \label{enum:property4}
\end{condition}
To be more specific, the first property follows by the way in which selective men are constructed. The second property follows from the following chain of facts:
\begin{quote}
   $w_i$ is selective $\Rightarrow$ $f(w_i)$ is not selective $\Rightarrow$ $p_{f(w_i),f(f(w_i))} > 1/2$ $\Rightarrow$ $p_{f(w_i),w_i} < 1/2$,
\end{quote}
since by construction of $\succ$, $w_i \neq f(f(w_i))$ i.e., woman $w_i$ is not the favorite of her favorite man. The third property follows directly from property one and two. 
The fourth property follows again directly by the way in which selective women are constructed. \medskip 

\noindent Now let us consider the expected social welfare of $\mu$ on $\bv$. Consider any pair $(m_i, s(m_i))$ of a man and his suitor. Since every man has a different suitor (and hence, obviously, every woman is the suitor of some man, since $|M|=|W|=n$), these pairs are disjoint and span the whole set of agents. Therefore, the expected social welfare of $\mu$ can be bounded as
\begin{align*}
\SW(\mu\mid\bv) &= \sum_{(m_i, s(m_i)) \in M \times W} \left(\sum_{w_j \in W}p_{m_i,w_j}v_{m_i}(w_j) + \sum_{m_j \in M}p_{s(m_i),m_j}v_{s(m_i)}(m_j)\right) \\
&\leq \sum_{(m_i, s(m_i)) \in M \times W} \left(\frac12 + 2/n\right) = \frac{n}{2} + 2,
\end{align*}
where the inequality above follows from the third and fourth property; Indeed, by the fourth property only one agent of the pair $(m_i, s(m_i))$ is selective, and by property three, its expected contribution to the social welfare of that agent is at most $1/2 \cdot 1 = 1/2$. The other agent of the pair is not selective and hence its expected contribution to the social welfare is $1/n$.  \medskip

Now consider the following matching $\mu^*$. For any man $m_i \in M$: 
\begin{itemize}
    \item[-] If $m_i$ is selective, let $\mu^*(m_i)=f(m_i)$, i.e., the man is matched to his favorite in $\mu^*$. 
\item[-] If $m_i$ is not selective (in which case, by property four, $s(m_i)$ is selective), let $\mu^*(m_i)=s(m_i)$, the man is matched to his suitor in $\mu^*$.
\end{itemize}
Similarly to above, we can express the social welfare of $\mu^*$ as the sum of the contributions to the social welfare by each pair $(m_i,s(m_i))$ of a man and his suitor, as follows:
\begin{align*}
\SW(\mu^*\mid\bv) &= \sum_{(m_i, s(m_i)) \in M \times W} \left(v_{m_i}(\mu^*(m_i)) + v_{s(m_i)}(\mu^*(s(m_i)))\right)\\
& = \sum_{(m_i, s(m_i)) \in M \times W} \left(1+\frac{1}{n}\right)\\
&= n+1
\end{align*}
The second equation follows from the fact that the selective agent in the pair $(m_i, s(m_i))$ is matched with its favorite, and hence contributes $1$ to the social welfare; the other agent is indifferent, and contributes $1/n$. \medskip

\noindent We have that $\SW(\mu^*\mid\bv) / \SW(\mu\mid\bv) \geq 2-\frac{3}{n/2+2}$, which goes to $2$ as $n \rightarrow \infty$. Therefore, to establish the desired distortion bound, it suffices to establish that $\mu^*$ is a stable matching.  

To see this, consider any pair $(m_i,w_j)$ as a potential blocking pair. Notice that if either $m_i$ or $w_j$ is selective, then it is matched to its favorite in $\mu^*$ by construction, and hence it cannot be part of any blocking pair. This implies that both $m_i$ and $w_j$ are indifferent. Again, by construction, in $\mu^*$ they are both matched to their suitors $s(m_i)$ and $s(w_j)$. By construction of the preference profile $\succ$, man $m_i$ ranks his suitor $s(m_i)$ second; therefore, the only way in which $m_i$ could be part of a blocking pair is if $w_j$ is $m_i$'s favorite $f(m_i)$. But this is the same as saying that $m_i$ is $w_j$'s suitor $s(w_j)$, whom, by construction of $\succ$, $w_j$ ranks last. This establishes that $(m_i, w_j)$ cannot be a blocking pair, and hence $\mu^*$ is stable.  
\end{proof}

\begin{remark}\label{rem:randomized-lower} The proof of \cref{thm:lower-bound-ordinal-randomized} actually does not use the stability of the matching produced by the algorithm as a property. In that sense, it in fact shows something quite stronger, namely that any ordinal randomized matching algorithm (which might even produce \emph{unstable} matchings) cannot approximate the best stable matching within a factor better than $2$. In turn, this implies a lower bound of $2$ for any of the randomized stability notions of the literature, even the weaker ones. We also remark, that, similarly to the proof of \cref{thm:any-ordinal-distortion-unbounded}, our lower bound in \cref{thm:lower-bound-ordinal-randomized} uses valuation profiles which are unit-sum normalized (see \cref{rem:general-ordinal-lower} and \cref{rem:unit-sum}); without this standard assumption, a lower bound of $2$ can be proven rather straightforwardly. 
\end{remark}
\section{Tradeoffs Between the Distortion and the Number of Queries}\label{sec:queries}

After settling the distortion of purely ordinal algorithms, we now turn to our main investigation, that of the tradeoffs between the distortion and the number of queries required to achieve it. Before we present the more technical results of the section, we start with the following easy observation, which serves as a useful point of reference, when navigating the space of best possible tradeoffs: A simple query-enhanced variant of Deferred Acceptance (\cref{alg:da}), which selects, among the man-optimal and the woman-optimal matching, the one with the highest revealed social welfare, achieves a distortion of $2$. We present the algorithm in \cref{alg:1-mowo}, and the simple proof of its distortion in \cref{thm:1-MOWO-2-distortion} below.


\begin{algorithm}[t]
\caption{\textsc{1-Query Man-Optimal or Woman-Optimal (1-MoWo)}}
\label{alg:1-mowo}

\KwIn{A preference profile $\succ = (\succ_M, \succ_W)$.}
\KwOut{A stable matching $\mu \in \mathcal{M}_s(\succ)$.}

\BlankLine
$\mu_M^* = \textsc{Deferred Acceptance}(\succ)$ (men proposing version)\\
$\mu_W^* = \textsc{Deferred Acceptance}(\succ)$ (women proposing version) 
\BlankLine

\BlankLine
\ForEach{man $m \in M$}{
    Query man $m$ for the woman he is matched with in $\mu^*_M$\;
}

\ForEach{woman $w \in W$}{
    Query woman $w$ for the man she is matched with in $\mu^*_W$\;
}

\BlankLine
\Return $\displaystyle 
\arg\max\!\left\{
    \SW\!\left(\mu^*_M \mid q^1(\mathbf{v})\right),
    \SW\!\left(\mu^*_W \mid q^1(\mathbf{v})\right)
\right\}$\;
\tcc{The matching with the maximum revealed social welfare among the two}

\end{algorithm}

\begin{theorem}\label{thm:1-MOWO-2-distortion}
The distortion of the \textsc{1-MoWo} algorithm is at most $2$.
\end{theorem}

\begin{proof}
In the man-optimal matching $\mu^*_M$, every man has a value for his assigned partner that is at least as high as in any other stable matching, in particular also in the optimal stable matching. Likewise, in the woman-optimal matching $\mu^*_W$, every woman has a value for her assigned partner that is at least as high as her value for her partner in the optimal stable matching. The social welfare of the men in  $\mu^*_M$ is precisely $\SW(\mu^*_M\mid q^1(\bv))$ since every man is queried for his match in $\mu^*_M$. Likewise, the social welfare of the women in  $\mu^*_W$ is precisely $\SW(\mu^*_W\mid q^1(\bv))$ since every woman is queried for her match in $\mu^*_W$. Therefore, for any stable matching $\mu \in \mathcal{M}_s$, we have
\begin{align*}
\SW(\mu^*_M\mid q^1(\bv)) \geq \sum_{m_i \in M} v_{m_i}(\mu(m_i)) \  \text{ and } \ \SW(\mu^*_W\mid q^1(\bv)) \geq \sum_{w_i \in W} v_{w_i}(\mu(w_i))
\end{align*}
This immediately implies that 
\[
\max\{\SW(\mu^*_M\mid q^1(\bv)), \SW(\mu^*_W\mid q^1(\bv)) \geq \frac{1}{2} \cdot \max_{\mu \in \mathcal{M}_s}\SW(\mu\mid \bv)\}
\]
and the bound follows. 
\end{proof}

\subsection{Beating the Barrier of 2: Impossibility Results} \label{sec:more-queries-impossibility}
Given \cref{thm:1-MOWO-2-distortion}, our task would be to design algorithms that beat the distortion bound of $2$, and ask how many queries these algorithms would require to achieve that. To this end, we prove a rather strong impossibility: to beat the bound of $2$, any deterministic algorithm must perform $\Omega(\log n)$ \emph{adaptive} queries or $\Omega(n)$ \emph{non-adaptive} queries. For non-adaptive queries, this essentially means that the algorithm would need to query \emph{asymptotically everything}. The bound also implies that a relatively large (superconstant) number of queries is necessary to improve on the bound of $2$, even if these queries can employ the whole power of adaptivity. Before we state and prove the corresponding theorems, we present the necessary machinery that will be used for the proofs.



\paragraph{The Cyclic Shift Profile.} Our lower bounds in this section will make use of the following preference profile $\succ$, which we refer to as the \emph{cyclic shift profile}.\footnote{Incidentally, the same preference profile was identified by \citet{knuth1997stable} as a profile on which the \textsc{Deferred Acceptance} algorithm achieves the worst possible running time, see \citep[pp. 15]{gusfield1989stable}. We note that, similarly to the preference profile used in the proof of \cref{thm:lower-bound-ordinal-randomized}, the description of the profile is slightly informal. A fully rigorous definition of the profile would set the $j$-th most preferred woman of $m_i$ to be $w_{(i-2+j)\pmod n +1}$, and the $j$-th most preferred man of $w_i$ to be $m_{(i-1+j)\pmod n +1}$. To avoid the cumbersome notation, we elected to go with the slightly more informal definition instead.}

\begin{definition}[The Cyclic Shift Profile]\label{def:cyclic-shift-profile}
In the \emph{cyclic shift preference profile} $\succ$, for all $i=1,\ldots,n$, 
\begin{itemize}
    \item[-] $m_i$ has preference ranking $w_i \succ_{m_i} w_{i+1} \succ_{m_i} w_{i+2} \succ_{m_i} \ldots \succ_{m_i} w_n \succ_{m_i} w_{1} \succ_{m_i }\ldots \succ_{m_i} w_{i-1}$
    \item[-] $w_i$ has preference ranking $m_{i+1} \succ_{w_i} m_{i+2} \succ_{w_i} m_{i+3} \succ_{w_i} \ldots \succ_{w_i} m_n \succ_{w_i} m_{1} \succ_{w_i }\ldots \succ_{w_i} m_i$
\end{itemize}
See also \cref{fig:cyclic-shift-profile} for a pictorial representation.
\end{definition}

\begin{figure}[t]
\centering

\begin{tikzpicture}[
        scale=0.7, transform shape,
    cell/.style={minimum width=1.1cm, minimum height=0.6cm, align=center},
    rowlbl/.style={anchor=east, font=\Large},
    collbl/.style={anchor=south, font=\Large}
]

\draw[line width=1.5pt, rounded corners]
  (-1.1,-4.3) rectangle (14.6,5.5);

\node[collbl] at (0.25,6) {$w_1$};
\node[collbl] at (3,6) {$w_2$};
\node[collbl] at (6,6) {$w_3$};
\node[collbl] at (8,6) {$\ldots$};
\node[collbl] at (10,6) {$w_{n-1}$};
\node[collbl] at (13,6) {$w_n$};

\node[rowlbl] at (-1.5,4.5) {$m_1$};
\node[rowlbl] at (-1.5,3.0) {$m_2$};
\node[rowlbl] at (-1.5,1.5) {$m_3$};
\node[rowlbl] at (-1.5,0.0) {$\vdots$};
\node[rowlbl] at (-1.5,-1.5) {$m_{n-1}$};
\node[rowlbl] at (-1.5,-3) {$m_n$};

\node[cell] (a00) at (0.25,4.5) {\Large $1,n$};
\node[cell] (a01) at (3,4.5) {\Large $2,n-1$};
\node[cell] (a02) at (6,4.5) {\Large $3,n-2$};
\node[cell] (a03) at (8,4.5) {\Large $\ldots$};
\node[cell] (a04) at (10,4.5) {\Large $n-1,2$};
\node[cell] (a05) at (13,4.5) {\Large $n,1$};

\node[cell] (a10) at (0.25,3.0) {\Large $n,1$};
\node[cell] (a11) at (3,3.0) {\Large $1,n$};
\node[cell] (a12) at (6,3.0) {\Large $2,n-1$};
\node[cell] (a13) at (8,3.0) {\Large $\ldots$};
\node[cell] (a14) at (10,3.0) {\Large $n-2,3$};
\node[cell] (a15) at (13,3.0) {\Large $n-1,2$};

\node[cell] (a20) at (0.25,1.5) {\Large $n-1,2$};
\node[cell] (a21) at (3,1.5) {\Large $n,1$};
\node[cell] (a22) at (6,1.5) {\Large $1,n$};
\node[cell] (a23) at (8,1.5) {\Large $\ldots$};
\node[cell] (a24) at (10,1.5) {\Large $n-3,4$};
\node[cell] (a25) at (13,1.5) {\Large $n-2,3$};

\node[cell] (a30) at (0.25,0) {\Large $\vdots$};
\node[cell] (a31) at (3,0) {\Large $\vdots$};
\node[cell] (a32) at (6,0) {\Large $\vdots$};
\node[cell] (a33) at (8,0) {\Large $\vdots$};
\node[cell] (a34) at (10,0) {\Large $\vdots$};
\node[cell] (a35) at (13,0) {\Large $\vdots$};

\node[cell] (a40) at (0.25,-1.5) {\Large $3,n-2$};
\node[cell] (a41) at (3,-1.5) {\Large $4,n-3$};
\node[cell] (a42) at (6,-1.5) {\Large $5,n-4$};
\node[cell] (a43) at (8,-1.5) {\Large $\ldots$};
\node[cell] (a44) at (10,-1.5) {\Large $1,n$};
\node[cell] (a45) at (13,-1.5) {\Large $2,n-1$};

\node[cell] (a50) at (0.25,-3) {\Large $2,n-1$};
\node[cell] (a51) at (3,-3) {\Large $3,n-2$};
\node[cell] (a52) at (6,-3) {\Large $4,n-3$};
\node[cell] (a53) at (8,-3) {\Large $\ldots$};
\node[cell] (a54) at (10,-3) {\Large $n,1$};
\node[cell] (a55) at (13,-3) {\Large $1,n$};

\draw[blue!60, line width=2pt] (a00) circle (0.7);
\draw[blue!60, line width=2pt] (a11) circle (0.7);
\draw[blue!60, line width=2pt] (a22) circle (0.7);
\draw[blue!60, line width=2pt] (a44) circle (0.7);
\draw[blue!60, line width=2pt] (a55) circle (0.7);

\draw[red!60, line width=2pt, dashed] (a01) ellipse (1 and 0.7);
\draw[red!60, line width=2pt, dashed] (a12) ellipse (1 and 0.7);
\draw[red!60, line width=2pt, dashed] (a45) ellipse (1 and 0.7);
\draw[red!60, line width=2pt, dashed] (a50) ellipse (1 and 0.7);

\draw[green!60!black, line width=2pt, dotted] (a02) ellipse (1 and 0.7);
\draw[green!60!black, line width=2pt, dotted] (a40) ellipse (1 and 0.7);
\draw[green!60!black, line width=2pt, dotted] (a51) ellipse (1 and 0.7);


\end{tikzpicture}
        \caption{The \emph{cyclic shift profile} used in the proofs of \cref{sec:more-queries-impossibility} and \cref{sec:more-queries-positive}. The entry $(m_i, w_j)$ consists of a pair of numbers; the first indicates the rank of $w_j$ for $m_i$ and the second indicates the rank of $m_i$ for $w_j$. For example, the entry $(m_2,w_1)$ is $(n,2)$ which indicates that $m_2$ ranks $w_1$ last among the women, and $w_1$ ranks $m_2$ second among the men. Notice that the preference of each man $m_i$ (in the rows of the table) is formed by placing $w_i$ first, $w_{i+1}$ second, and so on, with the preference \emph{cycling around} once the index reaches $n$. The preference of each woman is formed similarly, but in a shifted way: woman $w_i$ places $m_{i+1}$ first, $m_{i+2}$ second and so on, with preference cycling around once the index reaches $n$.
        The set of stable matchings on this instance consists of the ``cyclic diagonals''. For example, referring to the statement of \cref{lem:cyclic-shift-profile-stable-matchings}, $\mu_0$ corresponds to the main diagonal (shown in solid blue lines), $\mu_1$ corresponds to the first superdiagonal, together with the $n$-th subdiagonal (shown in dashed red lines), $\mu_2$ corresponds to the second superdiagonal, together with the $(n-1)$-th subdiagonal (shown in dotted green lines), etc. } 
        \label{fig:cyclic-shift-profile}
    \end{figure}

\noindent The following lemma characterizes the possible stable matchings on the cyclic shift profile of \cref{def:cyclic-shift-profile}, see also \cref{fig:cyclic-shift-profile}. Intuitively, the set of stable matchings on this profile is consists of the ``cyclic diagonals'', i.e., the main diagonal, the first superdiagonal together with the $n$-th subdiagonal, the second superdiagonal together with the $(n-1)$-th subdiagonal etc. Also notice that for each agent $a_i \in A$, every agent $a_j \in \bar{A}$ is a stable partner, and hence $S_{a_i} = \bar{A}$. We state the corresponding lemma formally below. The proof of the lemma requires more advanced machinery related to the structure of stable matchings, which we present in \cref{app:background-stable-matchings}. 

\begin{lemma} \label{lem:cyclic-shift-profile-stable-matchings}
On the cyclic shift profile of \cref{def:cyclic-shift-profile}, the set of stable matchings $\mathcal{M}_s$ consists of stable matchings $\mu_0, \ldots, \mu_{n-1}$, where $\mu_k = \{(m_i,w_{(i-1+k)\pmod n +1})\}$, for $k=0,\ldots,n-1$. 
\end{lemma}

\begin{proof}
For ease of notation throughout the proof, we will define the operator $\oplus$, such that for $k\in\{0,1,\dots,n-1\}$,
\[
i\oplus k := ((i-1+k)\bmod n)+1.
\]
Using this notation for the cyclic shift profile $\succ$ specifically, we have that for $i = 1, \dots,n$,
\begin{itemize}
    \item[-] $m_i$ has preference ranking $w_i \succ_{m_i} w_{i\oplus 1} \succ_{m_i} w_{i\oplus 2} \succ_{m_i} \cdots \succ_{m_i} w_{i\oplus (n-1)}$ 
    \item[-] $w_i$ has preference ranking $m_{i\oplus 1} \succ_{w_i} m_{i\oplus 2} \succ_{w_i} \cdots \succ_{w_i} m_{i\oplus (n-1)} \succ_{w_i} m_i$
\end{itemize}
Thus, what we need to prove is that the set of stable matchings is exactly
\[
\mathcal M_s=\{\mu_0,\mu_1,\dots,\mu_{n-1}\},
\qquad
\mu_k := \{(m_i, w_{i\oplus k}) : i\in[n]\}.
\]
To this end, we leverage the structure of the rotations, see \cref{def:rotation}. Notice that, given the definition of $\mathcal{M}_s$ above, it can be verified that
\[
\mu_0 = \mu_M^* \ \text{ (man-optimal) } \ \text{ and } \ \mu_{n-1} = \mu_W^* \  \text{ (woman-optimal),}
\]
e.g., by applying the men-proposing and women-proposing versions of the \textsc{Deferred Acceptance} algorithm to $\succ$. Hence, specifically for $\mu_0$, we have that $\mu_0=\{(m_i,w_i): i\in[n]\}$.

\noindent For $t\in\{1,2,\dots,n-1\}$, define the rotation
\[
\rho_t := \bigl( (m_1,w_{1\oplus (t-1)}), (m_2,w_{2\oplus (t-1)}), \dots, (m_n,w_{n\oplus (t-1)}) \bigr).
\]
\noindent Note that $\rho_t$ is exactly the list of pairs of $\mu_{t-1}$. In order to prove the statement, we first show that rotation $\rho_t$ is exposed in $\mu_{t-1}$ and its elimination yields $\mu_t$ for each $t\in\{1,2,\dots,n-1\}$; we prove this by induction on $t$.

\medskip
\noindent\textit{Base case : Rotation $\rho_1$ is exposed in $\mu_0$ and eliminating it yields $\mu_1$, i.e $\mu_1 := \mu_0 \setminus \rho_1$}

\noindent Let $i\in[n]$. In $\mu_0$, man $m_i$ is matched to $w_{i\oplus 0}=w_i$. The next woman on $m_i$'s preference profile is $w_{i\oplus 1}$.
In $\mu_0$, woman $w_{i\oplus 1}$ is matched to $m_{i\oplus 1}$.
By the preference profile for women, $w_{i\oplus 1}$ ranks men in the following order
\[
m_{(i\oplus 1)\oplus 1} \succ m_{(i\oplus 1)\oplus 2} \succ \cdots \succ m_i \succ m_{i\oplus 1},
\]
so in particular $w_{i\oplus 1}$ prefers $m_i$ to $m_{i\oplus 1}$, i.e.
\[
m_i \succ_{w_{i\oplus 1}} m_{i\oplus 1}.
\]
Moreover, since $w_{i\oplus 1}$ is the immediate successor of $w_i$ on $m_i$'s list, it is the highest-ranked woman after $w_i$ in $m_i$'s list satisfying this property. \medskip

\noindent Thus, for every $i \in [n]$, the successor of $(m_i,w_i)$ considering \cref{def:rotation} is $(m_{i\oplus 1},w_{i\oplus 1})$, and therefore the rotation
\[
\rho_1=\bigl((m_1,w_1),(m_2,w_2),\dots,(m_n,w_n)\bigr)
\]
is exposed in $\mu_0$. Eliminating $\rho_1$ matches each $m_i$ to the successor woman $w_{i\oplus 1}$, hence produces exactly the matching
\[
\mu_1=\{(m_i,w_{i\oplus 1}): i\in[n]\}
\]
Finally, $\rho_1$ is the only rotation exposed in $\mu_0$ because $\rho_1$ contains all pairs of $\mu_0$ and by \cref{lem: pair-in-one-rotation} a pair can appear in at most one rotation.

\medskip
\noindent\textit{Inductive step : For each $t\in \{1,\dots,n-1\}$, rotation $\rho_t$ is exposed in $\mu_{t-1}$ and eliminating it yields $\mu_t$.}

\noindent Assume for some $t\in\{1,\dots,n-1\}$ that $\mu_{t-1}$ is a stable matching and equals
\[
\mu_{t-1}=\{(m_i,w_{i\oplus (t-1)}): i\in[n]\}.
\]

\noindent Let $i\in[n]$. Similarly as before, in $\mu_{t-1}$, man $m_i$ is matched to $w_{i\oplus (t-1)}$. On $m_i$'s preference profile, the woman immediately after $w_{i\oplus (t-1)}$ is $w_{i\oplus t}$.
In $\mu_{t-1}$, man $m_{i\oplus 1}$ is matched to woman 
\[
w_{(i\oplus 1)\oplus (t-1)} = w_{i\oplus t},
\]
And therefore, woman $w_{i\oplus t}$ is matched to man
\[
\mu_{t-1}(w_{i\oplus t}) = m_{i\oplus 1}.
\]
Now compare $m_i$ and $m_{i\oplus 1}$ in $w_{i\oplus t}$'s preferences. By the women preference rule, $w_{i\oplus t}$ ranks men in the following order
\[
m_{(i\oplus t)\oplus 1} \succ m_{(i\oplus t)\oplus 2} \succ \cdots \succ m_{i} \succ m_{i\oplus 1} \succ \cdots \succ m_{i\oplus t},
\]
so in particular
\[
m_i \succ_{w_{i\oplus t}} m_{i\oplus 1}.
\]
Thus $w_{i\oplus t}$ prefers $m_i$ to her current partner in $\mu_{t-1}$. \medskip

\noindent Since $w_{i\oplus t}$ is the immediate successor of $w_{i\oplus (t-1)}$ on $m_i$'s list, it is the highest-ranked woman after $\mu_{t-1}(m_i)$ with this property.
Therefore, by \cref{def:rotation}, the rotation
\[
\rho_t=\bigl((m_1,w_{1\oplus (t-1)}),\dots,(m_n,w_{n\oplus (t-1)})\bigr)
\]
is exposed in $\mu_{t-1}$ and eliminating $\rho_t$ matches each $m_i$ to the successor $w_{i\oplus t}$, yielding exactly
\[
\mu_t=\{(m_i,w_{i\oplus t}): i\in[n]\}.
\]
As in the base case, $\rho_t$ is the only exposed rotation in $\mu_{t-1}$ because it contains all pairs of $\mu_{t-1}$.

\medskip
\noindent Next we show that the rotation poset is a chain, i.e. $\rho_1 \prec \rho_2 \prec \cdots \prec \rho_{n-1}$, see \cref{def: chain}. Namely, from the induction above, eliminating $\rho_t$ produces all the pairs of $\mu_t$, i.e.\ it creates each pair $(m_i,w_{i\oplus t})$.
But $(m_i,w_{i\oplus t})$ is precisely the $i$th pair of $\rho_{t+1}$ (since $\rho_{t+1}$ consists of $(m_i,w_{i\oplus t})$ for all $i$).
Hence $\rho_t$ produces a pair that is eliminated by $\rho_{t+1}$, so $\rho_t$ must precede $\rho_{t+1}$ in the rotation poset.

\medskip
\noindent Finally, by \cref{thm:rotation-poset} the stable matchings of an instance are in bijection with the closed subsets of the rotation poset, i.e. starting from the man-optimal matching, eliminating precisely the rotations in a closed subset yields a stable matching, which is uniquely identified by this subset, see \cref{lemma:unique-rotation-subset}. In this instance, the rotation poset is a chain $\rho_1 \prec \cdots \prec \rho_{n-1}$, so the closed subsets are the following
\[
\varnothing,\ \{\rho_1\},\ \{\rho_1,\rho_2\},\ \dots,\ \{\rho_1,\dots,\rho_{n-1}\}.
\]
Eliminating the first $k$ rotations yields $\mu_k$ by the induction above. Therefore the set of stable matchings is exactly
\[
\mathcal M_s = \{\mu_0,\mu_1,\dots,\mu_{n-1}\},
\qquad
\mu_k=\{(m_i,w_{i\oplus k}): i\in[n]\}.
\]
This completes the proof.
\end{proof}


\paragraph{Dichotomous Valuation Profiles.} For the lower bounds of the section, we will consider valuation profiles $\bv$ consistent with the cyclic shift profile $\succ$. In fact, those will be \emph{dichotomous} valuation profiles, defined formally below. 

\begin{definition}[Dichotomous Valuation Function and Profile]\label{def:dichotomous}
Let $a_i \in A$ be an agent. A valuation function $v_{a_i}$ is \emph{dichotomous} if its image is contained in $\{0,1\}$, i.e., if for each agent $a_j \in \bar{A}$, it holds that $v_{a_i}(a_j) \in \{0,1\}$. A valuation profile is dichotomous if it is defined by dichotomous valuation functions. 
\end{definition}

\paragraph{Transition Points.} For convenience, when $v_{a_i}$ is dichotomous, we will use an alternative equivalent representation of the function, rather than the valuation vector representation introduced in \cref{sec:prelims}. We will represent the function $v_{a_i}$ by a single number $\tr(v_{a_i})$, which is the index of the first agent $a_j \in \bar{A}$ in $\succ_{a_i}$ for which $v_{a_i}(a_j)=0$. We will refer to $\tr(v_{a_i})$ as the \emph{transition point} of $v_{a_i}$. For completeness, we may define $\tr(v_{a_i})=\varnothing$ when $v_{a_i}(a_j)=1$ for all $a_j \in \bar{A}$, but such valuation functions will not appear in any of the valuation profiles that we will construct.  

\paragraph{Uninformed Regions.} Before we proceed, we introduce some further terminology, which will make the exposition of the results of the section easier. Let $\mathcal{A}$ be some query-enhanced algorithm applied to some valuation profile $\bv$ and let $\bv^p$ be a partial restriction of $\bv$. Given an agent $a_i \in A$, we define the \emph{uninformed set} of agent $a_i$, denoted $U_{a_i,\bv^p}$ to be the set of agents $a_j \in \bar{A}$ for which the values $v_{a_i}(a_j)$ are not known, i.e., $U_{a_i,\bv^p}:=\{a_i \in \bar{A}: v_{a_i}(a_j) \notin \bv^p\}$. Notice that for dichotomous valuation profiles, $U_{a_i,\bv^p}$ is an interval, meaning that it contains agents $a_j \in \bar{A}$ that appear in consecutive positions in the preference ranking of agent $a_i$.  
We will refer to $U_{a_i}$ as the \emph{uninformed region} of agent $a_i$. When it is clear from the context, we will drop $\bv^p$ from the notation and simply write $U_{a_i}$. \medskip 

\noindent We are now ready to present our first lower bound of the section, which applies to algorithms that use non-adaptive queries. The bound establishes that in the worst case, we need to query asymptotically the whole valuation profile $\bv$ in order to improve over the distortion of $2$, which is achieved by the \textsc{1-MoWo} algorithm. 

\begin{theorem}\label{thm:lower-bound-queries-non-adaptive}
    Let $\mathcal{A}$ be any deterministic stable algorithm that uses fewer than $n/2$ non-adaptive queries. Then $\text{dist}(\mathcal{A}) \geq 2$.
\end{theorem}

\begin{proof}
We will consider dichotomous valuation profiles $\bv$ consistent with the cyclic shift profile $\succ$ of \cref{def:cyclic-shift-profile}. In all of these profiles, for any agent $a_i \notin \{m_1,w_1\}$, we will have $\tr(v_{a_i}) = 1$, i.e., the agent will have value $0$ for all the agents of the opposite side. For man $m_1$ and woman $w_1$, the transition points $\tr(v_{m_1})$ and $\tr(v_{w_1})$ will depend on the positions where the algorithm queries the agents. Since these queries are non-adaptive, we can construct the valuation profile $\bv$, given the positions of the queries in advance. We will let $\bv^p=q^{k}(\bv)$ be the restriction of $\bv$ to the answers of the $k$ queries per agent that $\mathcal{A}$ is allowed to make (with $k < n/2$), and the uninformed regions $U_{m_1}$ and $U_{w_1}$ for man $m_1$ and woman $w_1$, respectively. We first prove the following claim:
\begin{quote}
    Regardless of the positions of the queries for $m_1$ and $w_1$, there exists a stable matching $\tilde{\mu}$ such that $\tilde{\mu}(m_1) \in U_{m_1}$ and $\tilde{\mu(w_1)} \in U_{w_1}$. 
\end{quote}
By the pigeonhole principle, $|U_{m_1}| > n/2$ and $|U_{w_1}| > n/2$, as the algorithm uses fewer than $n/2$ queries per agent. By \cref{lem:cyclic-shift-profile-stable-matchings}, all women in $U_{m_1}$ are potential stable partners for $m_1$, and furthermore, for each $j=1,\ldots,|U_{m_1}|$, there is a unique stable matching $\mu^j$ such that $\mu^j(m_1)=w_j$, see also \cref{fig:cyclic-shift-profile}. Consider the set $S_{w_1}=\{\mu_1(w_1),\ldots,\mu_{k}(w_1)\}$ of the assigned partners of woman $w_1$ in those stable matchings. Since $|S_{w_1}|=|U_{m_1}| > n/2$ and $|M\setminus U_{w_1}|< n/2$, it follows that $|S_{w_1} \cap M\setminus U_{w_1}| > 1$, i.e., there is a man in $S_{w_1}$ that woman $w_1$ was not queried for. By the way that $S_{w_1}$ was constructed, this implies that there exists some $j \in 1, \ldots, |U_{m_1}|$ such that $w_1$ was not queried for man $\mu^j(w_1)$. Together with the fact that $\mu^j(m_1) \in W \setminus T_{m_1}$, it follows that $m_1$ was not queried for woman $\mu^j(m_1)$ either, and hence $\mu^j(m_1) \in U_{m_1}$ and $\mu^j(w_1) \in U_{w_1}$. \medskip

\noindent We are now ready to define the valuation profile $\bv$. By the discussion in the first paragraph of the proof, it suffices to define $\tr(v_{m_1})$ and $\tr(v_{w_1})$. For any agent $a_i \in A$, let $r(a_i)$ be the position of $\tilde{\mu}(a_i)$ in its preference ranking $\succ_{a_i}$, where $\tilde{\mu}$ is the matching $\mu^j$ identified above. Let $\mu$ be the matching computed by $\mathcal{A}$ on input $\succ$ and the query oracle $\mathcal{Q}$. We consider two cases:
\begin{itemize}[leftmargin=*]
    \item[-] \emph{Case 1: $\mu=\tilde{\mu}$.} In this case, let $\tr(v_{m_1})=r(\tilde{\mu}(m_1))$, and $\tr(v_{w_1})=r(\tilde{\mu}(w_1))$. In other words, both $m_1$ and $w_1$ have value $1$ for any agent on the other side which is preferred to their matched partner in $\tilde{\mu}$, and $0$ for any other agent, including their partner in $\tilde{\mu}$. In this case, $\SW(\mu \mid \bv) = 0$. %
    At the same time, there exist stable matchings that have a positive welfare; for example for the man-optimal matching $\mu^{*}_M$, we have that $\SW(\mu^{*}_M \mid \bv) = 1$, and hence the distortion of $\mathcal{A}$ is infinite. 
    \item[-] \emph{Case 2: $\mu\neq \tilde{\mu}$.} In this case, let $\tr(v_{m_1})=r(\tilde{\mu}(m_1))+1$ and $\tr(v_{w_1})=r(\tilde{\mu}(w_1))+1$. In other words both $m_1$ and $w_1$ have value $1$ for their matched partner in $\tilde{\mu}$, and any agent on the other side that is preferred to their matched partner, and $0$ for any agent on the other side which is not preferred to their matched partner. $\tilde{\mu}$ is a stable matching for which $\SW(\tilde{\mu} \mid \bv) = 2$. By construction, for any other matching $\mu$, it holds that 
\[
\text{either }\mu(m_1) \succ_{m_1} \tilde{\mu}(m_1) \ \text{ and } \  \tilde{\mu}(w_1) \succ_{w_1} \mu(w_1), \ \text{ or }\  
\mu(w_1) \succ_{w_1} \tilde{\mu}(w_1) \ \text{ and } \ \tilde{\mu}(m_1) \succ_{m_1} \mu(m_1)
\]
 In either case, we have that $\SW(\mu \mid \bv) = 1$, and the distortion bound follows.
 \end{itemize}
This completes the proof.
\end{proof}

\noindent Next, we prove our lower bound on algorithms that use adaptive queries. In contrast to the proof of \cref{thm:lower-bound-queries-non-adaptive}, our construction of the valuation profiles used for the bound has to be dynamic, taking into account the answers to the previous queries as well.

\begin{theorem}\label{thm:lower-bound-queries-adaptive}
    Let $\mathcal{A}$ be any deterministic stable algorithm that uses fewer than $(\log n) / 2$ adaptive queries. Then $\text{dist}(\mathcal{A}) \geq 2$.
\end{theorem}

\begin{proof}
Similarly to the proof of \cref{thm:lower-bound-queries-non-adaptive}, we will consider dichotomous valuation profiles $\bv$ consistent with the cyclic shift profile $\succ$ of \cref{def:cyclic-shift-profile}. In all of these profiles, for any agent $a_i \notin \{m_1,w_1\}$, we will have $\tr(v_{a_i}) = 1$, i.e., the agent will have value $0$ for all the agents of the opposite side. For man $m_1$ and woman $w_1$, the transition points $\tr(v_{m_1})$ and $\tr(v_{w_1})$ will depend on the positions where the algorithm queries the agents. Contrary to the proof of \cref{thm:lower-bound-queries-non-adaptive}, the queries are now adaptive; this means that we cannot provide a ``bad'' valuation profile consistent with their answers at the end of the query process, but, rather, we will have to construct such a profile dynamically, updating it after each individual query. Therefore, the answers $q^t(\bv)$ to the first $t$ queries will define a partial valuation profile and corresponding uninformed regions $U_{m_1}^t$ and $U_{w_1}^t$, which will be updated for each value of $t=1,\ldots, k$, where $k < (\log n)/2$ is the number of queries performed by $\mathcal{A}$. Initially, we have $U_{m_1}^0 = W$ and $U_{w_1}^0=M$. 

Recall the set of stable matchings $\mu_0, \ldots, \mu_{n-1}$ defined in the statement of \cref{lem:cyclic-shift-profile-stable-matchings}; by the lemma, these are all the stable matchings on $\succ$, with $\mu_0 = \mu_M^*$ and $\mu_{n-1} = \mu_{W}^*$. Specifically for $m_1$ and $w_1$, this implies that the set of stable partners are
\begin{align}
&(\mu_0(m_1), \mu_1(m_1), \ldots, \mu_{n-2}(m_1), \mu_{n-1}(m_1)) = (w_1, w_2, \ldots, w_{n-1}, w_n) \text { and } \label{eq:women-matchings}\\
&(\mu_0(w_1), \mu_1(w_1), \ldots, \mu_{n-2}(w_1), \mu_{n-1}(w_1)) = (m_1, m_n, \ldots, m_3, m_2) \label{eq:men-matchings}
\end{align}
In other words, for both agents, all agents of the other side are potential stable partners, and, furthermore, the sequence of matchings induces a sequence of stable partners for the agents; for $m_1$ this is the same as his preference ranking, and for $w_1$ it is the reverse of her preference ranking, see also \cref{fig:proof-of-lower-bound}. Given this, we 
have an alternative (equivalent) interpretation of the uninformed region, via the set of stable matchings that match agent $a$ with an agent $b \in \bar{A}$ of the other side, for whom the value $v_a(b)$ is not known. Formally, we have:
\[
\hat{U}_a^t = \{\mu \in \mathcal{M}_s \mid \mu(a) \in U_a^t\}
\]
Notice that, given the discussion above, there is a one-to-one mapping between the elements of $U_a^t$ and the elements of $\hat{U}_a^t$. \medskip 

\noindent Now consider agent $a \in \{m_1,w_1\} \subset A$ such that $\mathcal{A}$ performs one query for $v_{a}(b)$, for some agent $b \in \bar{A}$. If the answer to the query is $1$, this means that for any agent $b'$ such that $b' \succ_{a} b$, it holds that $v_{a}(b')=1$ as well. If the answer is $0$, this means that for any agent $b'$ such that $b \succ_{a} b'$, it holds that $v_{a}(b')=0$ as well. In terms of the uninformed region $U_{a}^t$, this effectively means that after the query we have the following update rule:
\begin{equation}
\label{eq:update-rule}
    U_a^t=
    \begin{cases}
    U_a^{t-1}\setminus
    \bigl(\{b\}\cup\{b'\in\bar{A}\mid b'\succ_a b\}\bigr),
        &\text{if }v_a(b)=1,\\
    U_a^{t-1}\setminus
    \bigl(\{b\}\cup\{b'\in\bar{A}\mid b\succ_a b'\}\bigr),
        &\text{if }v_a(b)=0.
    \end{cases}
\end{equation}
It will be useful to keep track of the set of agents in $\bar{A}$ for whom we learned the value of agent $a$ after the $t$-th query; formally, we let $L_a^t = \{U_a^{t-1} \setminus U_a^t\}$. It follows from the update rule in \cref{eq:update-rule} that $L_a^t$ is an interval which is \emph{adjacent} to $U_a^t$, i.e., it either ``precedes'' or ``succeeds'' $U_a^t$ and does not intersect it. We also define $\hat{L}_a^t = \{\hat{U}_a^{t-1} \setminus \hat{U}_a^t\}$ for the corresponding set of ``learned'' stable matchings. 

\medskip 

\noindent We maintain the invariant that at least one stable matching remains uninformed for both agents. In particular, similarly to the proof of \cref{thm:lower-bound-queries-non-adaptive}, the valuation profile $\bv$ that we will construct will ensure that after $k \leq (\log n)/2$ queries, the following holds:
\begin{quote}
    Regardless of the positions of the $k$ queries for $m_1$ and $w_1$, $|\hat{U}_{m_1}^k \cap \hat{U}_{w_1}^k|\geq 1$, i.e., there exists a stable matching $\tilde{\mu}$ such that $\tilde{\mu}(m_1) \in U_{m_1}^k$ and $\tilde{\mu}(w_1) \in U_{w_1}^k$.
\end{quote}
Given queries $t=1,\ldots,k$, we will define the answer to those queries in a way that ensures that the uninformed region of the agent \emph{shrinks as little as possible}. To simplify our exposition, we will assume that queries to men $m_1$ shrink the uniformed region of woman $w_1$ and vice versa; if the statement above holds in this case, it certainly also holds in the case where the regions are shrunk only by the queries to the agent itself. In particular, suppose that $\mathcal{A}$ has already asked $\ell$ queries to agent $b \in \bar{A}$, and $t-1$ queries to agent $a \in A$, and is now asking the $t$-th query to agent $a$. After the update of $U_a^{t-1}$ to $U_a^t$ according to \cref{eq:update-rule}, we will also update $U_b^\ell$ (and correspondingly $\hat{U}_b^\ell$) as follows:

\begin{figure}[ht!]

\begin{tcolorbox}[
    standard jigsaw,
    opacityback=0,  
]
\centering
\scalebox{1}{
    \begin{tikzpicture}[
    font=\large,
    every node/.style={inner sep=1pt},
]

\matrix (top) [matrix of math nodes,
    row sep=6mm,
    column sep=2mm] {
m_1: &
w_1 & \succ & w_2 & \succ &
w_3 & \succ & w_4 & \succ & w_5 & \succ & \cdots & \succ & w_{n-3} &
\succ & w_{n-2} & \succ & w_{n-1} & \succ & w_n \\
};

\node[
    draw=red,
    rounded corners,
    fill=red,
    fill opacity=0.15,   
    fit=(top-1-6)(top-1-16),
    inner sep=4pt
] {};

\node[red!70!black,above=3mm of top-1-2] {$\mu_0$};
\node[red!70!black,above=3mm of top-1-4] {$\mu_1$};
\node[red!70!black,above=3mm of top-1-6] {$\mu_2$};
\node[red!70!black,above=3mm of top-1-8] {$\mu_3$};
\node[red!70!black,above=3mm of top-1-10] {$\mu_4$};
\node[red!70!black,above=3mm of top-1-12] {$\cdots$};
\node[red!70!black,above=3mm of top-1-14] {$\mu_{n-4}$};
\node[red!70!black,above=3mm of top-1-16] {$\mu_{n-3}$};
\node[red!70!black,above=3mm of top-1-18] {$\mu_{n-2}$};
\node[red!70!black,above=3mm of top-1-20] {$\mu_{n-1}$};

\node[blue!70!black,below=2mm of top-1-2] {$1$};
\node[blue!70!black,below=2mm of top-1-4] {$1$};
\node[blue!70!black,below=2mm of top-1-18] {$0$};
\node[blue!70!black,below=2mm of top-1-20] {$0$};

\draw[
    decorate,
    decoration={brace,mirror,raise=4pt,amplitude=10pt}
]
(top-1-6.south west) -- (top-1-16.south east)
node[midway,below=15pt] {$U_{m_1}^{t-1}$};

\matrix (bot) [matrix of math nodes,
    below=20mm of top,
    column sep=2mm] {
w_1: &
m_2 & \succ & m_3 & \succ &
m_4 & \succ & m_5 & \succ & m_6 & \succ & \cdots & \succ & m_{n-2} &
\succ & m_{n-1} & \succ & m_n & \succ & m_1\\
};

\node[
    draw=red,
    rounded corners,
    fill=red,
    fill opacity=0.15,   
    fit=(bot-1-6)(bot-1-16),
    inner sep=4pt
] {};

\node[red!70!black,above=3mm of bot-1-2] {$\mu_{n-1}$};
\node[red!70!black,above=3mm of bot-1-4] {$\mu_{n-2}$};
\node[red!70!black,above=3mm of bot-1-6] {$\mu_{n-3}$};
\node[red!70!black,above=3mm of bot-1-8] {$\mu_{n-4}$};
\node[red!70!black,above=3mm of bot-1-10] {$\mu_{n-5}$};
\node[red!70!black,above=3mm of bot-1-12] {$\cdots$};
\node[red!70!black,above=3mm of bot-1-14] {$\mu_{3}$};
\node[red!70!black,above=3mm of bot-1-16] {$\mu_{2}$};
\node[red!70!black,above=3mm of bot-1-18] {$\mu_{1}$};
\node[red!70!black,above=3mm of bot-1-20] {$\mu_{0}$};

\node[blue!70!black,below=2mm of bot-1-2] {$1$};
\node[blue!70!black,below=2mm of bot-1-4] {$1$};
\node[blue!70!black,below=2mm of bot-1-18] {$0$};
\node[blue!70!black,below=2mm of bot-1-20] {$0$};

\draw[
    decorate,
    decoration={brace,mirror,raise=4pt,amplitude=10pt}
]
(bot-1-6.south west) -- (bot-1-16.south east)
node[midway,below=15pt] {$U_{w_1}^{\ell}$};

\end{tikzpicture} 
}
\end{tcolorbox} 

\begin{tcolorbox}[
    standard jigsaw,
    opacityback=0,  
]
\centering
\scalebox{1}{

    \begin{tikzpicture}[
    font=\large,
    every node/.style={inner sep=1pt},
]

\matrix (top) [matrix of math nodes,
    row sep=6mm,
    column sep=2mm] {
m_1: &
w_1 & \succ & w_2 & \succ &
w_3 & \succ & w_4 & \succ & w_5 & \succ & \cdots & \succ & w_{n-3} &
\succ & w_{n-2} & \succ & w_{n-1} & \succ & w_n \\
};

\node[
    draw=red,
    rounded corners,
    fill=red,
    fill opacity=0.15,   
    fit=(top-1-6)(top-1-16),
    inner sep=4pt
] {};

\node[
    draw=green,
    rounded corners,
    pattern=north east lines,   
    pattern color=green,
    fill opacity=0.8,   
    fit=(top-1-6)(top-1-8),
    inner sep=4pt
] {};

\node[red!70!black,above=3mm of top-1-2] {$\mu_0$};
\node[red!70!black,above=3mm of top-1-4] {$\mu_1$};
\node[red!70!black,above=3mm of top-1-6] {$\mu_2$};
\node[red!70!black,above=3mm of top-1-8] {$\mu_3$};
\node[red!70!black,above=3mm of top-1-10] {$\mu_4$};
\node[red!70!black,above=3mm of top-1-12] {$\cdots$};
\node[red!70!black,above=3mm of top-1-14] {$\mu_{n-4}$};
\node[red!70!black,above=3mm of top-1-16] {$\mu_{n-3}$};
\node[red!70!black,above=3mm of top-1-18] {$\mu_{n-2}$};
\node[red!70!black,above=3mm of top-1-20] {$\mu_{n-1}$};

\node[blue!70!black,below=2mm of top-1-2] {$1$};
\node[blue!70!black,below=2mm of top-1-4] {$1$};
\node[blue!70!black,below=2mm of top-1-18] {$0$};
\node[blue!70!black,below=2mm of top-1-20] {$0$};

\draw[
    decorate,
    decoration={brace,mirror,raise=4pt,amplitude=10pt}
]
(top-1-9.south west) -- (top-1-16.south east)
node[midway,below=15pt] {$U_{m_1}^{t}$};

\draw[
    decorate,
    decoration={brace,mirror,raise=4pt,amplitude=10pt}
]
(top-1-6.south west) -- (top-1-8.south east)
node[midway,below=15pt] {$L_{m_1}^{t}$};

\matrix (bot) [matrix of math nodes,
    below=20mm of top,
    column sep=2mm] {
w_1: &
m_2 & \succ & m_3 & \succ &
m_4 & \succ & m_5 & \succ & m_6 & \succ & \cdots & \succ & m_{n-2} &
\succ & m_{n-1} & \succ & m_n & \succ & m_1\\
};

\node[
    draw=red,
    rounded corners,
    fill=red,
    fill opacity=0.15,   
    fit=(bot-1-6)(bot-1-16),
    inner sep=4pt
] {};

\node[
    draw=green,
    rounded corners,
    pattern=north east lines,   
    pattern color=green,
    fill opacity=0.8,   
    fit=(bot-1-14)(bot-1-16),
    inner sep=4pt
] {};

\node[red!70!black,above=3mm of bot-1-2] {$\mu_{n-1}$};
\node[red!70!black,above=3mm of bot-1-4] {$\mu_{n-2}$};
\node[red!70!black,above=3mm of bot-1-6] {$\mu_{n-3}$};
\node[red!70!black,above=3mm of bot-1-8] {$\mu_{n-4}$};
\node[red!70!black,above=3mm of bot-1-10] {$\mu_{n-5}$};
\node[red!70!black,above=3mm of bot-1-12] {$\cdots$};
\node[red!70!black,above=3mm of bot-1-14] {$\mu_{3}$};
\node[red!70!black,above=3mm of bot-1-16] {$\mu_{2}$};
\node[red!70!black,above=3mm of bot-1-18] {$\mu_{1}$};
\node[red!70!black,above=3mm of bot-1-20] {$\mu_{0}$};

\node[blue!70!black,below=2mm of bot-1-2] {$1$};
\node[blue!70!black,below=2mm of bot-1-4] {$1$};
\node[blue!70!black,below=2mm of bot-1-18] {$0$};
\node[blue!70!black,below=2mm of bot-1-20] {$0$};

\draw[
    decorate,
    decoration={brace,mirror,raise=4pt,amplitude=10pt}
]
(bot-1-6.south west) -- (bot-1-13.south east)
node[midway,below=15pt] {$U_{w_1}^{\ell}$};

\draw[
    decorate,
    decoration={brace,mirror,raise=4pt,amplitude=10pt}
]
(bot-1-14.south west) -- (bot-1-16.south east)
node[midway,below=15pt] {$\{\mu(w_1) \mid \mu \in \hat{L}_{m_1}^{t}$\}};
\end{tikzpicture}
}
\end{tcolorbox}
\caption{The argument used in the proof of \cref{thm:lower-bound-queries-adaptive}. The top part of the figure depicts a situation in which $t-1$ queries have been asked to man $m_1$ and $\ell$ queries have been asked to woman $w_1$; the corresponding uninformed regions $U_{m_1}^{t-1}$ and $U_{w_1}^\ell$ are shaded in red. Observe also the one-to-one correspondence between stable partners of agent $a \in \{m_1,w_1\}$ and the possible stable matchings for this instance, established by \cref{lem:cyclic-shift-profile-stable-matchings}. In the figure, we have that $\hat{U}_{m_1}^{t-1} = \hat{U}_{w_1}^\ell=\{\mu_2, \mu_3, \mu_4, \ldots, \mu_{n-4}, \mu_{n-3}\}$. The bottom part of the picture shows the situation after query $t$ has been asked to man $m_1$ for woman $w_4$. This partitions the previous uninformed region $U_{m_1}^{t-1}$ into the new uninformed region $U_{m_1}^t$, and the region of the values learned via this query $L_{m_1}^t$. As we mention in the proof, we assume that this query also updates the uninformed region $U_{w_1}^\ell$ - note that the index in the superscript does not change because there was no query for woman $w_1$ in this round. $U_{w_1}^\ell$ is updated by removing from it the men that are matched with $w_1$ in the matchings in $\hat{L}_{m_1}^t$, namely $m_{n-2} = \mu_3(w_1)$ and $m_{n-1} = \mu_2(w_1)$.}
\label{fig:proof-of-lower-bound}
\end{figure}

\begin{equation}
\label{eq:update-ub}
U_b^\ell = U_b^\ell \setminus \{\mu(b) \mid \mu \in \hat{L}_a^t\}
\end{equation}
In other words, for every matching $\mu$ that assigns agent $a$ to a partner with a known value, the partner $\mu(b)$ of agent $b$ will also be considered to be known, see \cref{fig:proof-of-lower-bound}.

Now observe that the set $\{\mu(b) \mid \mu \in \hat{L}_a^t\}$ is an interval adjacent to $U_b^\ell$, after the update of $U_b^\ell$ in \cref{eq:update-ub}; this follows from the fact that $L_a^t$ is an interval adjacent to $U_q^t$, and the structure of the correspondence of agents and stable matchings in \cref{eq:women-matchings} and \cref{eq:men-matchings} above. Since $U_a^t$ and $U_b^t$ are updated ``together'', $U_a^t$ fully defines $U_b^t$ and vice versa; therefore it is equivalent to assume that $\mathcal{A}$ performs $2k$ queries for the values of a ``proxy'' agent $c \in \{m_1,w_1\}$, and prove that $|U_{c}^{2k}| \geq 1$. In that case, by construction, for any stable matching $\mu \in \hat{U}_{c}^{2k}$, we will have $\mu(b) \in U_b^{k}$.\medskip 

\noindent We will construct the profile $\bv$ given the following rule: for $t=1,\ldots,2k$ and a queried partner
$\bar{b}\in U_c^{t-1}$, set
\[
v_c(\bar{b})=
\begin{cases}
0,
&\text{if }
|\{b'\in U_c^{t-1}\mid b'\succ_c\bar{b}\}|
\geq
|\{b'\in U_c^{t-1}\mid\bar{b}\succ_c b'\}|,\\
1,
&\text{otherwise}.
\end{cases}
\]
Queries outside $U_c^{t-1}$ return their previously fixed values
and leave the region unchanged. In simple words, the answer to the query is $0$ if it is asked at the bottom half of the current uninformed region $U_c^{t-1}$, and $1$ if it is asked at the top half. Observe that by definition, we have that $|U_c^t|+1 \geq \frac12\bigl(|U_c^{t-1}|+1\bigr)$. This implies that $|U_c^{2k}|+1
\geq \frac{1}{2^{2k}}\bigl(|U_c^0|+1\bigr)
= \frac{n+1}{2^{2k}}$. Since $k < (\log_2 n)/2$, we have $|U_c^{2k}|+1>1$, and hence $|U_c^{2k}|\geq1$. This implies that $|\hat{U}_{m_1}^k \cap \hat{U}_{w_1}^k|\geq 1$.  \medskip

\noindent To complete the proof, we work similarly as in the proof of \cref{thm:lower-bound-queries-non-adaptive}. The valuation profile $\bv$ has been partially defined, apart from the unexplored regions $U_{m_1}^k$ and $U_{w_1}^k$. Specifically, from the discussion above, we know that there exists at least one matching $\tilde{\mu} \in \{\mu_0,\ldots,\mu_{n-1}\}$ such that $\tilde{\mu}(m_1) \in U_{m_1}^k$ and $\tilde{\mu}(w_1) \in U_{w_1}^k$. For any agent $a \in A$, let $r(a)$ be the position of $\tilde{\mu}(a)$ in its preference ranking $\succ_{a}$.
Let $\mu$ be the matching computed by $\mathcal{A}$ on input $\succ$ and the query oracle $\mathcal{Q}$. We consider two cases:
\begin{itemize}[leftmargin=*]
    \item[-] \emph{Case 1: $\mu=\tilde{\mu}$.} In this case, let $\tr(v_{m_1})=r(m_1)$, and $\tr(v_{w_1})=r(w_1)$. In other words, both $m_1$ and $w_1$ have value $1$ for any agent on the other side which is preferred to their matched partner in $\tilde{\mu}$, and $0$ for any other agent, including their partner in $\tilde{\mu}$. In this case, $\SW(\mu \mid \bv) = 0$. %
    At the same time, there exist stable matchings that have a positive welfare; for example for the man-optimal matching $\mu^{*}_M$, we have that $\SW(\mu^{*}_M \mid \bv) = 1$, and hence the distortion of $\mathcal{A}$ is infinite. 
    \item[-] \emph{Case 2: $\mu\neq \tilde{\mu}$.} In this case, let $\tr(v_{m_1})=r(m_1)+1$ and $\tr(v_{w_1})=r(w_1)+1$. In other words both $m_1$ and $w_1$ have value $1$ for their matched partner in $\tilde{\mu}$, and any agent on the other side that is preferred to their matched partner, and $0$ for any agent on the other side which is not preferred to their matched partner. $\tilde{\mu}$ is a stable matching for which $\SW(\tilde{\mu} \mid \bv) = 2$. By construction, for any other matching $\mu$, it holds that 
\[
\text{either }\mu(m_1) \succ_{m_1} \tilde{\mu}(m_1) \ \text{ and } \  \tilde{\mu}(w_1) \succ_{w_1} \mu(w_1), \ \text{ or }\  
\mu(w_1) \succ_{w_1} \tilde{\mu}(w_1) \ \text{ and } \ \tilde{\mu}(m_1) \succ_{m_1} \mu(m_1)
\]
 In either case, we have that $\SW(\mu \mid \bv) = 1$, and the distortion bound follows.
 \end{itemize}
This completes the proof.
\end{proof}

\subsection{Optimal Distortion/Query Complexity Tradeoffs}\label{sec:more-queries-positive}

Our lower bounds in \cref{sec:more-queries-impossibility} establish that although a distortion of $2$ can rather easily be achieved with $1$ query, achieving a better distortion seems quite challenging, even if we are allowed to employ more queries. It is natural to consider the question of how many queries are sufficient to beat the bound of $2$, and, more generally, what the best possible tradeoff between the number of queries and the distortion is. Given \cref{thm:lower-bound-queries-non-adaptive}, we are only concerned with algorithms that use adaptive queries. The general question that we aim to answer in this subsection is the following:

\begin{quote}
    \emph{Given an $\varepsilon > 0$, how many queries does a deterministic stable algorithm need to use in order to achieve distortion at most $1+\varepsilon$?}
\end{quote}

\noindent We start with the lower bound. Building on our machinery developed in \cref{sec:more-queries-impossibility}, we can actually parameterize the bound of \cref{thm:lower-bound-queries-adaptive} by some $\varepsilon$, to lower bound the number of queries required to achieve a distortion of $1+\varepsilon$, instead of $2$. The proof again uses the cyclic shift profile, and valuation profiles that assign non-zero value to only one man and one woman, similarly to the proof of \cref{thm:lower-bound-queries-adaptive}. The main difference is the answers to the queries by the two agents of interest, which need not be in $\{0,1\}$. Hence, the valuation profiles that we construct to be consistent with the cyclic shift profile are not dichotomous; instead, we use $b=\Theta(\min\{n,\varepsilon^{-1}\})$ blocks, each with its own pair of values. This construction forces any deterministic algorithm to incur distortion greater than $1+\varepsilon$. Besides that, the arguments in the proof are very similar to those of \cref{thm:lower-bound-queries-adaptive}. We present the corresponding theorem and its proof below; given the similarities, the proof often refers to the proof of \cref{thm:lower-bound-queries-adaptive}, highlighting the differences and developing the new arguments needed.

\begin{theorem}\label{thm:lower-bound-queries-adaptive-optimal}
Let $n\geq 2$ and $\varepsilon\in(0,1)$. 
Let $\mathcal{A}$ be any deterministic stable algorithm that uses fewer than $\frac{1}{4\varepsilon}\log_2(1+n\varepsilon)$
adaptive value queries per agent.
Then $\text{dist}(\mathcal{A})>1+\varepsilon$.
\end{theorem}

\begin{proof}
Let
$b=\min\{n,\lceil\varepsilon^{-1}\rceil-1\}$ and
$L=\lfloor n/b\rfloor$.
We will show that any deterministic stable algorithm $\mathcal{A}$
that performs fewer than $\frac{b}{2}\log_2(L+1)$ queries per agent
has distortion at least $1+1/b \geq 1+\varepsilon$. To see why this implies the bound stated in the theorem, notice that
\begin{itemize}
    \item[-] If $b<n$, then $b\geq1/(2\varepsilon)$ and $\lfloor n/b\rfloor\geq n/(2b)$. Consequently,
\[
\frac b2\log_2\!\left(1+\frac{n}{2b}\right)
\geq\frac{1}{4\varepsilon}\log_2(1+n\varepsilon),
\]
where the second inequality uses that \(x\log_2(1+n/x)\) is increasing in \(x>0\) (e.g., see \citep[Problem~2.79, p.~76]{mitrinovic1964elementary}).
    \item[-] If $b=n$, the number of queries is at least $n/2$. Since \(\ln(1+x)\leq x\) (e.g., see \citep[Problem~2.79, p.~76]{mitrinovic1964elementary}), we have that 
\[
\frac n2 \geq \frac{n}{4\ln2} \geq \frac{1}{4\varepsilon}\log_2(1+n\varepsilon).
\]
\end{itemize}

\noindent First, we consider valuation profiles $\bv$ consistent with the cyclic shift profile of \cref{def:cyclic-shift-profile}. As in the proof of \cref{thm:lower-bound-queries-adaptive}, all agents other than $m_1,w_1$ have value zero for every agent of the opposite side. Now, however, we partition these matchings into $b$ consecutive blocks
$\mathcal{M}_s^1,\ldots,\mathcal{M}_s^b$, each containing at least $L$
matchings. For every $\mu\in\mathcal{M}_s^j$, we will have
\[
 v_{m_1}(\mu(m_1))\in\{b-j,b-j+1\}
 \quad\text{and}\quad
 v_{w_1}(\mu(w_1))\in\{j-1,j\}.
\]
Within each block, the values will respect the agents' rankings. These choices also respect the rankings across blocks; every possible value of $m_1$ in block $j$ is at least every possible value in block $j+1$, and the reverse holds for $w_1$. Also, we use the uninformed sets $U_a^t$ and their corresponding sets of
stable matchings $\hat{U}_a^t$, as in the proof of \cref{thm:lower-bound-queries-adaptive}.
Initially, $U_{m_1}^0=W$ and $U_{w_1}^0=M$ and let $k<\frac{b}{2}\log_2(L+1)$ be the number of queries allowed for algorithm $\mathcal{A}$ to perform per agent.
We may assume that each agent receives exactly $k$ queries, if we not we simply pad them with additional queries whose answers do not affect the outcome. 

Next, we construct the query answers separately in each block, following the rule of \cref{thm:lower-bound-queries-adaptive}. Namely, for a query to $m_1$ in block $j$, replace
the answers $0$ and $1$ in that rule by $b-j$ and $b-j+1$,
respectively. For a query to $w_1$ in the same block, replace them
by $j-1$ and $j$, respectively.
Thus, a query in the bottom half of the current uninformed region
receives the lower of the two possible values, and a query in the
top half receives the higher value. In either case, we retain a
larger remaining portion of that region. Queries to already known
entries return their fixed values. As in \cref{thm:lower-bound-queries-adaptive}, we update the two agents' uninformed regions
together. Whenever a query determines one agent's value at a stable
matching $\mu$, we also fix the other agent's value so that $v_{m_1}(\mu(m_1))+v_{w_1}(\mu(w_1))=b.$ Consequently, after $t$ queries to $m_1$ and $\ell$
queries to $w_1$, we have, in every block $j$, that $\hat{U}_{m_1}^{t}\cap\mathcal{M}_s^j = \hat{U}_{w_1}^{\ell}\cap\mathcal{M}_s^j.$ Next, let $t_j$ be the total number of queries to $m_1$ and
$w_1$ in block $j$. The halving argument of \cref{thm:lower-bound-queries-adaptive}, then yields
\begin{equation}\label{eq:interestion-uninformed-regions}
    \bigl|
 \hat{U}_{m_1}^{k}\cap\hat{U}_{w_1}^{k}
 \cap\mathcal{M}_s^j
 \bigr|+1
 \geq
 \frac{|\mathcal{M}_s^j|+1}{2^{t_j}}
 \geq
 \frac{L+1}{2^{t_j}}.
\end{equation}
Since
\[
 \sum_{j=1}^{b}t_j\leq2k<b\log_2(L+1),
\]
there exists a block $j$ with $t_j<\log_2(L+1)$ and hence, from \cref{eq:interestion-uninformed-regions} we get that  $\bigl|\hat{U}_{m_1}^{k}\cap\hat{U}_{w_1}^{k} \cap\mathcal{M}_s^j \bigr|\geq1.$
In particular, there exists a stable matching
$\tilde{\mu}\in\mathcal{M}_s^j$ such that
$\tilde{\mu}(m_1)\in U_{m_1}^k$ and
$\tilde{\mu}(w_1)\in U_{w_1}^k$.

To complete the proof, we work similarly as in the final two cases
of \cref{thm:lower-bound-queries-adaptive}. First, in every block $i\ne j$, assign all remaining
undefined values of $m_1$ to $b-i+1$, and the corresponding values
of $w_1$ to $i-1$. These assignments preserve the rankings and
previous answers, and give every matching in those blocks welfare
$b$. For each $a\in\{m_1,w_1\}$, let $r(a)$ be the position of
$\tilde{\mu}(a)$ in its preference ranking $\succ_a$.
Let $\mu$ be the matching computed by $\mathcal{A}$.

\begin{itemize}[leftmargin=*]
    \item[-] \textit{Case 1: $\mu=\tilde{\mu}$.}
    Within block $j$, give $m_1$ value $b-j+1$ for partners
    preferred to $\tilde{\mu}(m_1)$, and value $b-j$ for
    $\tilde{\mu}(m_1)$ and all worse partners.
    Likewise, give $w_1$ value $j$ for partners preferred to
    $\tilde{\mu}(w_1)$, and value $j-1$ for
    $\tilde{\mu}(w_1)$ and all worse partners. In this case, $\operatorname{SW}(\mu\mid\bv)=b-1$.
    Every other stable matching has welfare $b$, since within
    block $j$ exactly one of the two agents receives its higher
    value. Such a matching exists because $n\ge2$.
    Therefore, if $b>1$, the distortion is at least
    \[
        \frac{b}{b-1}\geq1+\frac1b.
    \]
    If $b=1$, the output has welfare zero and another stable
    matching has positive welfare, giving infinite distortion.

    \item[-] \textit{Case 2: $\mu\ne\tilde{\mu}$.}
    Within block $j$, give $m_1$ value $b-j+1$ for
    $\tilde{\mu}(m_1)$ and every preferred partner, and value
    $b-j$ for every worse partner.
    Likewise, give $w_1$ value $j$ for
    $\tilde{\mu}(w_1)$ and every preferred partner, and value
    $j-1$ for every worse partner.
    Equivalently, each agent receives its higher block value
    up to and including position $r(a)$, and its lower block
    value afterwards.

    In this case,
    $\operatorname{SW}(\tilde{\mu}\mid\bv)=b+1$.
    Every other stable matching, including the output $\mu$,
    has welfare $b$. The distortion is therefore at least
    \[
        \frac{b+1}{b}=1+\frac1b.
    \]
\end{itemize}

\noindent Both completions respect the preference rankings and preserve every
query answer. Consequently, guaranteeing distortion at most $1+\varepsilon$
requires
\[
 k\geq
 \frac b2\log_2\!\left(1+\left\lfloor\frac nb\right\rfloor\right)
 =
 \Omega\!\left(\frac{\log(1+n\varepsilon)}{\varepsilon}\right).
\]
\end{proof}

\noindent We now move on to prove a matching upper bound, i.e., a stable algorithm that achieves a distortion of $1+\varepsilon$ using $O(\log(1+n\varepsilon)/\varepsilon)$ (adaptive) queries per agent, see \cref{alg:recursive-astsf}. As we explained in \cref{sec:our-techniques}, we first present a simpler version of the algorithm, one that achieves a distortion of $1+\varepsilon$ with slightly more queries asymptotically, namely $O(\log n/ \varepsilon)$ queries per agent, see \cref{alg:astsf}. As we mentioned earlier, the use of \emph{additive thresholds} enables us to achieve this query complexity bound. With multiplicative thresholds, as in e.g., \citep{amanatidis2022few}, one can obtain a weaker bound of $O(\log n / \varepsilon^2)$ on the number of queries; for completeness, we provide the version of the algorithm with multiplicative thresholds and its analysis in \cref{app:simple-stableTSF}.

\begin{algorithm}[!t]
\caption{\textsc{Stable Additive Threshold Step Function (Stable-ATSF)}}
\label{alg:astsf}
\LinesNumbered
\KwIn{A preference profile $\succ=(\succ_M, \succ_W)$, and an $\varepsilon \in (0,1]$.}
\KwOut{A stable matching $\mu \in \mathcal{M}_s(\succ)$.}
\BlankLine

Let $\lambda = \lceil 2/\varepsilon \rceil - 1$, and $t_\ell = 1-\ell/(\lambda+1)$, for $\ell = 0, \cdots, \lambda$.\\
\tcc{The difference $\delta$ is $t_{\ell-1}-t_\ell$, which is $1/(\lambda+1)$ for any $\ell=1,\ldots,\lambda$.}
\BlankLine

\If{$n \leq 1+\lambda\lceil\log_2 n\rceil$}
{Perform $n$ queries per agent to learn the whole valuation profile $q^n(\bv)=\bv$. Let $\tilde{\bv} = \bv$.} \label{line:additive-query-all}
\Else{
Run Algorithm $\mathcal{A}_\text{gus}$ to compute, for each agent $a_i \in M \cup W$, the set $S_{a_i}$ of its stable partners (by \cref{lem:all-stable-pairs}), and construct the preference ranking $\succ_{a_i}^s$ based on $S_{a_i}$.
\BlankLine

\ForEach{agent $a_i \in A$, for both $A=M,W$}{

    Query agent $a_i$ for her favorite stable partner $f^s(a_i) \in \bar{A}$.\\
    Set $T_{a_i,0} \leftarrow \{f^s(a_i)\}$, and let $v_{a_i}^* = v_{a_i}(f^s(a_i))$.\\
    Set $\tilde v_{a_i}(f^s(a_i)) = v_{a_i}^*$.\\
    \tcc{Preserve the favorite's value. Earlier buckets take priority, so this value is never overwritten.}

    \For{$\ell = 1$ \KwTo $\lambda$}{
    \tcc{Use binary search to create buckets, where the $\ell$-th bucket contains all the stable partners of $a_i$ for whom $a_i$'s value is between $t_\ell$ and $t_{\ell-1}$ times its value for its favorite stable partner, excluding partners already assigned to earlier buckets.}
        Using binary search on $\succ^s_{a_i}$, compute \label{line:additive-binary-search}
        \[
        T_{a_i,\ell}
        \leftarrow
        \left\{
            a_j \in S_{a_i} : v_{a_i}(a_j) \in [t_\ell \cdot v_{a_i}^*,\, t_{\ell-1} \cdot v_{a_i}^*]
        \right\}
        \setminus \bigcup_{r=0}^{\ell-1} T_{a_i,r}
        \] 

        Set $\tilde v_{a_i}(a_j) = t_\ell \cdot v_{a_i}^*$ for every $a_j \in T_{a_i,\ell}$\\
        \tcc{Set the simulated value of agent $a_i$ for agent $a_j$ to be the lower boundary of the bucket.}
    }

    Set $T_{a_i} \leftarrow \bigcup_{\ell=0}^{\lambda} T_{a_i,\ell}$\\
    \tcc{The set of stable partners contained in the constructed buckets.}

    Set $\tilde v_{a_i}(a_j) = 0$ for every $a_j \in \bar{A} \setminus T_{a_i}$\\
    \tcc{Set the simulated value to $0$ for partners outside the constructed buckets, including stable partners whose values are below the last threshold.}
}
}

\BlankLine
\Return $\mu \in \arg \max_{\hat{\mu} \in \mathcal{M}_s(\succ)} \SW (\hat{\mu} \mid \tilde{\bv})$, using Algorithm $\mathcal{A}_\text{ilg}$.\\
\tcc{Return a stable matching that maximizes the simulated social welfare, using the algorithm of \citet{irving1987efficient}.}
\end{algorithm}

\begin{theorem}\label{thm:astsf-distortion}
For any $\varepsilon\in(0,1]$, the \textsc{Stable-ATSF} algorithm performs at most
\[\min\left\{n,\,1+\left(\left\lceil\frac{2}{\varepsilon}\right\rceil-1\right)\lceil\log_2 n\rceil\right\}\]
queries per agent and achieves distortion at most $1+\varepsilon$.
\end{theorem}

\begin{proof}
Consider any valuation profile $\bv$, and let $\mu^*(\bv)$ be the optimal stable matching on $\bv$. Let $\mu$ be the stable matching returned by the \textsc{Stable-ATSF} algorithm. If the algorithm learns the whole valuation profile, it returns an
optimal stable matching. Hence, for the distortion analysis, we may
assume that the algorithm executes the ``\textbf{else}'' branch. We will partition the agents in $M\cup W$ into two sets $S$ and $\bar S$, where $S$ will contain the agents whose partner in $\mu^*$ is in the set $T_{a_i}$, i.e., their partner falls into one of the first $\lambda+1$ buckets constructed by the algorithm. Formally, we have
\[
S:=\{a_i\in M\cup W\mid \mu^*(a_i)\in T_{a_i}\},
\qquad\text{and}\qquad
\bar S=(M\cup W)\setminus S.
\]
For ease of notation, let also $\bar S_M=\bar S\cap M$ and $\bar S_W=\bar S\cap W$, and observe that $\bar S_M\cup\bar S_W=\bar S$. We will bound the contribution to the social welfare of $\mu^*$ from agents in those two sets separately. We start with the agents in the set $\bar S$. We have:
\begin{align}
\sum_{a_i\in\bar S}v_{a_i}(\mu^*(a_i))
&=
\sum_{m_i\in\bar S_M}v_{m_i}(\mu^*(m_i))
+
\sum_{w_i\in\bar S_W}v_{w_i}(\mu^*(w_i))
\notag\\
&\leq
t_\lambda\cdot
\left(
\sum_{m_i\in\bar S_M}v_{m_i}(f^s(m_i))
+
\sum_{w_i\in\bar S_W}v_{w_i}(f^s(w_i))
\right)
\notag\\
&=
t_\lambda\cdot
\left(
\SW_M(\mu_M^*\mid\bv)
+
\SW_W(\mu_W^*\mid\bv)
-
\sum_{a_i\in S}v_{a_i}(f^s(a_i))
\right)
\notag\\
&\leq
2t_\lambda\cdot\SW(\mu\mid\bv)
-
t_\lambda\cdot\sum_{a_i\in S}v_{a_i}(f^s(a_i))
\label{eq:additive-stsf-negative-correction}\\
&=
\frac{2}{\lambda+1}\cdot\SW(\mu\mid\bv)
-
\frac{1}{\lambda+1}\cdot\sum_{a_i\in S}v_{a_i}(f^s(a_i)).
\notag
\end{align}
The first inequality follows from the definition of the set $\bar S$, since each agent $a_i$ in this set has value at most $t_\lambda\cdot v_{a_i}(f^s(a_i))$ for their assigned partner $\mu^*(a_i)$ in the optimal matching. The second equation follows from the fact that in the man-optimal (resp.\ woman-optimal) matching, each man (resp.\ woman) is matched with his (resp.\ her) favorite stable partner, and from $\bar S=(M\cup W)\setminus S$.

The last inequality follows from the fact that every agent $a_i$ has been queried for its favorite stable partner $f^s(a_i)$, and therefore
$\tilde v_{a_i}(f^s(a_i))=v_{a_i}(f^s(a_i))$.
Since the algorithm returns a matching that maximizes the simulated social welfare, and the simulated social welfare is always a lower bound on the actual social welfare of a stable matching, we have
\begin{align*}
\SW_M(\mu_M^*\mid\bv)+\SW_W(\mu_W^*\mid\bv)
&=
\sum_{a_i\in M\cup W}v_{a_i}(f^s(a_i))\\
&\leq
\SW(\mu_M^*\mid\tilde{\bv})
+
\SW(\mu_W^*\mid\tilde{\bv})\\
&\leq
2\SW(\mu\mid\tilde{\bv})
\leq
2\SW(\mu\mid\bv).
\end{align*}
The last equation follows by the definition of $t_\lambda$. We retain the negative correction term in \eqref{eq:additive-stsf-negative-correction}, since the additive error for the agents in $S$ will appear in the next bound. \medskip

\noindent Next, we bound the contribution to the social welfare from the agents in $S$. For $\ell=0,\ldots,\lambda$, let
\[
S_\ell=\{a_i\in S\mid \mu^*(a_i)\in T_{a_i,\ell}\}
\]
be the restriction of the agents in $S$ to those whose partner in the optimal stable matching is in the $\ell$-th bucket. For any $\ell\geq1$, any agent $a_i\in S_\ell$, and any agent $a_j\in T_{a_i,\ell}$, we have
\begin{align*}
v_{a_i}(a_j)
&\leq
t_{\ell-1}\cdot v_{a_i}(f^s(a_i))\\
&=
t_\ell\cdot v_{a_i}(f^s(a_i))
+
(t_{\ell-1}-t_\ell)\cdot v_{a_i}(f^s(a_i))\\
&=
\tilde v_{a_i}(a_j)
+
\delta\cdot v_{a_i}(f^s(a_i)),
\end{align*}
where $\delta=t_{\ell-1}-t_\ell$ for all $\ell=1,\ldots,\lambda$.
The first inequality above follows by the definition of the set $T_{a_i,\ell}$, and the last equation follows from the definition of the simulated valuation function. The same final inequality holds for $\ell=0$, where the simulated and true values coincide.
Therefore, we can bound the contribution to the social welfare of $\mu^*$ from the agents in $S$ as follows:
\begin{align}
\sum_{a_i\in S}v_{a_i}(\mu^*(a_i))
&=
\sum_{\ell=0}^{\lambda}
\sum_{a_i\in S_\ell}v_{a_i}(\mu^*(a_i))
\notag\\
&\leq
\sum_{\ell=0}^{\lambda}
\sum_{a_i\in S_\ell}
\left(
\tilde v_{a_i}(\mu^*(a_i))
+
\delta\cdot v_{a_i}(f^s(a_i))
\right)
\notag\\
&=
\SW(\mu^*\mid\tilde{\bv})
+
\delta\cdot\sum_{a_i\in S}v_{a_i}(f^s(a_i))
\notag\\
&\leq
\SW(\mu\mid\bv)
+
\delta\cdot\sum_{a_i\in S}v_{a_i}(f^s(a_i))
\label{eq:additive-stsf-positive-error}\\
&=
\SW(\mu\mid\bv)
+
\frac{1}{\lambda+1}\cdot\sum_{a_i\in S}v_{a_i}(f^s(a_i)).
\notag
\end{align}
Here the second equation uses the fact that agents in $\bar S$ have simulated value zero for their partners in $\mu^*$. The last inequality follows from the fact that $\mu$ maximizes the social welfare when calculated with the simulated valuation functions, and that these functions underestimate the true values on stable pairs.
The last equation follows from the definition of $\delta$, since
\[
\delta
=
t_{\ell-1}-t_\ell
=
\left(1-\frac{\ell-1}{\lambda+1}\right)
-
\left(1-\frac{\ell}{\lambda+1}\right)
=
\frac{1}{\lambda+1}.
\]

\noindent Putting everything together, the negative correction term in
\eqref{eq:additive-stsf-negative-correction} cancels the positive
additive-error term in \eqref{eq:additive-stsf-positive-error},
since $t_\lambda=\delta=1/(\lambda+1)$, and we have
\[
\SW(\mu^*\mid\bv)
\leq
\left(1+\frac{2}{\lambda+1}\right)\SW(\mu\mid\bv).
\]
Substituting the value of $\lambda$, we have that
\begin{align*}
\operatorname{dist}(\textsc{Stable-ATSF})
&\leq
1+\frac{2}{\lambda+1}\\
&=
1+\frac{2}{(\lceil2/\varepsilon\rceil-1)+1}\\
&=
1+\frac{2}{\lceil2/\varepsilon\rceil}
\leq
1+\varepsilon,
\end{align*}
where the last inequality holds for any $\varepsilon\in(0,1]$.
The welfare inequality also holds when all stable matchings have welfare zero.\medskip

\noindent Each binary search procedure in \cref{alg:astsf} requires at most $\lceil\log_2 n\rceil$ queries after the initial favorite query, and the procedure is called $\lceil2/\varepsilon\rceil-1$ times per agent. If $n\leq1+\lambda\lceil\log_2 n\rceil$, the algorithm queries the whole valuation profile using $n$ queries per agent, therefore the number of queries that it performs is upper bounded by
\[
\min\left\{
n,\;
1+
\left(\left\lceil\frac{2}{\varepsilon}\right\rceil-1\right)
\lceil\log_2 n\rceil
\right\}.
\]
This completes the proof.
\end{proof}

\begin{remark}\label{rem:additive-terms-non-needed}
We provide a remark regarding the role of the additive correction term $\sum_{a_i \in S} v_{a_i}(f^s(a_i))$ which appears in the bounds for the social welfare of agents in the sets $\bar{S}$ and $S$ (see \eqref{eq:additive-stsf-negative-correction} and \eqref{eq:additive-stsf-positive-error} respectively).  
Retaining this term in the analysis allows us to cancel out
its contribution to the welfare of agents in the two sets improving the constants
in the query complexity bound, but it does not affect the asymptotic bounds. Indeed, we could discard
the negative correction term and use the cruder bound
\[
\sum_{a_i\in S}v_{a_i}(f^s(a_i))
\leq
\sum_{a_i\in M\cup W}v_{a_i}(f^s(a_i))
\leq
2\SW(\mu\mid\bv)
\]
to bound the positive additive term separately. The resulting
distortion bound would be $1+4/(\lambda+1)$, rather than
$1+2/(\lambda+1)$. Thus, choosing
$\lambda=\lceil4/\varepsilon\rceil-1$ instead of
$\lambda=\lceil2/\varepsilon\rceil-1$ would still yield distortion
at most $1+\varepsilon$, with at most
\[
\min\left\{
n,\;
1+
\left(\left\lceil\frac{4}{\varepsilon}\right\rceil-1\right)
\lceil\log_2 n\rceil
\right\}
\]
queries per agent. This remains $O(\log n/\varepsilon)$, although
the number of binary searches is approximately doubled. We elected to incorporate the cancelation of the term in the analysis, to obtain a bound that is as close as possible to the lower bound given by \cref{thm:lower-bound-queries-adaptive-optimal}. 
\end{remark}

\SetKwFunction{SearchInterval}{SearchInterval}
\SetKwProg{SearchProcedure}{Procedure}{:}{}

\begin{algorithm}[!t]
\caption{\textsc{Stable Additive Threshold Step Function with Recursive Search (Stable-ATSF-Recursive})}
\label{alg:recursive-astsf}
\LinesNumbered
\tcc{Everything is the same as in \cref{alg:astsf}, except for the full-query condition and the replacement of the binary-search bucket-construction step by the subroutine call below.}

$\ldots$\\
\BlankLine

\If{$n \leq \lambda+1$}{
    $\ldots$\\
}
\Else{
    $\ldots$\\
    \BlankLine

    Compute
    \[
    (T_{a_i,0},\ldots,T_{a_i,\lambda})
    \leftarrow
    \textsc{Recursive-Bucket-Search}
    \left(
        a_i,\succ^s_{a_i},(t_\ell)_{\ell=0}^{\lambda},v_{a_i}^*
    \right)
    \]
    using \Cref{alg:recursive-bucket-search}.\\
    \tcc{Call once per agent, after querying the favorite and before the loop assigning simulated values. Remove the binary-search bucket-construction step from that loop.}

    $\ldots$\\
}

\BlankLine
Return $\mu \in \arg \max_{\hat{\mu} \in \mathcal{M}_s(\succ)} \SW (\hat{\mu} \mid \tilde{\bv})$, using Algorithm $\mathcal{A}_\text{ilg}$.\\
\end{algorithm}

\begingroup
\SetAlgorithmName{Subroutine}{subroutine}{List of Subroutines}
\crefalias{algocf}{subroutine}

\begin{algorithm}[!t]
\caption{\textsc{Recursive-Bucket-Search}}
\label{alg:recursive-bucket-search}
\LinesNumbered
\KwIn{An agent $a_i$; its stable-partner ranking $\succ^s_{a_i}$; the thresholds $(t_\ell)_{\ell=0}^{\lambda}$; and the already queried value $v_{a_i}^*=v_{a_i}(f^s(a_i))$.}
\KwOut{The disjoint buckets $(T_{a_i,0},\ldots,T_{a_i,\lambda})$.}
\BlankLine

Write $\succ^s_{a_i}$ as $b_1\succ^s_{a_i}\cdots\succ^s_{a_i}b_s$, where $b_1=f^s(a_i)$.\label{line:stable-partners}\\
Initialize $T_{a_i,0}\leftarrow\{b_1\}$ and $T_{a_i,\ell}\leftarrow\varnothing$ for $\ell=1,\ldots,\lambda+1$.\\
Introduce an artificial endpoint $b_{s+1}$.\\
Set $T_{a_i,\lambda+1}\leftarrow\{b_{s+1}\}$.\\
\tcc{The auxiliary bucket $T_{a_i,\lambda+1}$ contains partners below the last threshold. The artificial endpoint belongs to this bucket by convention and is never queried.}

\SearchInterval{$1,s+1$}\\
\Return $(T_{a_i,0},\ldots,T_{a_i,\lambda})$.\\

\BlankLine
\tcc{The following procedure shares the ranking, thresholds, favorite value, and buckets with the main part of the algorithm.}

\SearchProcedure{\SearchInterval{$L,R$}}{
    \tcc{Both endpoints have known bucket membership, and $1\leq L<R\leq s+1$.}

    \If{$R-L=1$}{
        \Return\\
    }

    \If{$b_L$ \emph{and} $b_R$ \emph{belong to the same bucket} $T_{a_i,\ell}$}{
        Set $T_{a_i,\ell}\leftarrow
        T_{a_i,\ell}\cup\{b_j:L<j<R\}$.\\
        \Return\\
    }

    Set $p\leftarrow\lfloor(L+R)/2\rfloor$.\label{line:midpoint}\\
    Query $v_{a_i}(b_p)$.\label{line:query}\\

    \If{$v_{a_i}(b_p)<t_\lambda\cdot v_{a_i}^*$}{
        Set $T_{a_i,\lambda+1}\leftarrow
        T_{a_i,\lambda+1}\cup\{b_p\}$.\label{line:samebucket}\\
        \tcc{If the revealed value is below the last threshold, assign $b_p$ to the last (auxiliary) bucket $T_{a_i,\lambda+1}$.}
    }
    \Else{
        Set
        \[
        \ell\leftarrow
        \min\left\{
            r\in\{1,\ldots,\lambda\}:
            v_{a_i}(b_p)\in
            [t_r\cdot v_{a_i}^*,\,t_{r-1}\cdot v_{a_i}^*]
        \right\}.
        \]
        Set $T_{a_i,\ell}\leftarrow T_{a_i,\ell}\cup\{b_p\}$.\\
        \tcc{Assign $b_p$ to the appropriate bucket $T_{a_i,\ell}$.}
    }

    \SearchInterval{$L,p$}\\
    \SearchInterval{$p,R$}\\
}
\end{algorithm}
\endgroup

\noindent While the use of additive thresholds allows us to achieve a query complexity bound of $O(\log n / \varepsilon)$ (as opposed to $O(\log n / \varepsilon^2)$ achieved by multiplicative thresholds, see \cref{app:simple-stableTSF} and the discussion in \cref{sec:our-techniques}), this bound is still not matching the lower bound proved in \cref{thm:lower-bound-queries-adaptive-optimal}. It turns out that we can exactly (asymptotically) match that bound, with a refined search procedure (see \Cref{alg:recursive-bucket-search}) rather than the ``vanilla' set of binary searches described in Line~\ref{line:additive-binary-search} of \cref{alg:astsf}; see \cref{alg:recursive-astsf}. We refer the reader to \cref{sec:our-techniques} for a detailed discussion of this refined search. The rest of the algorithm is identical to \cref{alg:astsf}; we provide the pseudocode for the new algorithm in \cref{alg:recursive-astsf}, omitting all the identical parts. The distortion/query complexity tradeoff achieved by the algorithm is captured by the theorem below.

\begin{theorem}
For any $\varepsilon\in(0,1]$, the \textsc{Stable-ATSF-Recursive} algorithm performs at most
\[\min\left\{n,\,\frac{4\log_2(1+n\varepsilon)}{\varepsilon}\right\}\]
queries per agent and achieves distortion at most $1+\varepsilon$.
\label{thm:recursive-astsf-distortion}
\end{theorem}

\begin{proof}
If the algorithm learns the whole valuation profile, it returns an
optimal stable matching. Otherwise, the \textsc{Recursive-Bucket-Search} subroutine (\Cref{alg:recursive-bucket-search}) constructs exactly the same buckets as the binary search
procedure of \cref{alg:astsf}, and hence defines exactly the same
simulated valuation functions. Therefore, the distortion bound follows directly
from the proof of \cref{thm:astsf-distortion}. Thus, we only need to establish the query complexity of the \textsc{Stable-ATSF-Recursive} algorithm. 


Notice that, when the algorithm executes the ``\textbf{else}'' branch, once the midpoint $p$ of an interval $(b_i, b_j)$ is queried, it becomes the right and left endpoint of the resulting intervals $(b_i, b_p)$ and $(b_p,b_j)$ respectively, and is never queried again. This establishes that no stable partner is queried more than once, for a total of at most $s \leq n$ queries, establishing the first part of the $\min$ expression in the stated bound. For the second part, notice that, since each call of the recursion splits each interval in half, at depth $d$ of the recursion the subroutine performs at most $2^d$ queries. Furthermore, the bucket indices are non-decreasing in the ranking of the agent's stable partners, and there are at most $\lambda+1$ transition points between distinct buckets, including the artificial endpoint. Every interval that requires a query contains at least one such transition. Since the intervals at any fixed depth have disjoint interiors, there are at most $\lambda+1$ possible queries at that depth. Therefore, the number $q$ of queries per agent performed by \Cref{alg:recursive-astsf}, including the initial query for $f^s(a_i)$, can be bounded as follows:
\begin{equation}
q\leq
1+\sum_{d=0}^{\lceil \log_2 n\rceil-1}\min\{2^d,\lambda+1\}.
\label{eq:recursive-astsf-query-sum}
\end{equation}
We consider two cases, depending on the relation between $n$ and $\lambda+1$.\medskip

\noindent \emph{Case 1:} $n \leq \lambda +1$. In this case, \cref{alg:recursive-astsf} learns the whole valuation profile using $q=n$ queries per agent. The second part of the bound then follows since in this case we have 
\begin{align*}
&n\varepsilon \leq (\lambda+1)\varepsilon
= \varepsilon\left\lceil\frac{2}{\varepsilon}\right\rceil
< \varepsilon\left(\frac{2}{\varepsilon}+1\right)
= 2+\varepsilon \leq 3,
\text{ and } \\[3pt]
&\log_2(1+x)\geq x/(1+x) \ \text{(e.g., see 
\citep[Problem~2.33, part~$1^\circ$, p.~67]{mitrinovic1964elementary})}
\end{align*}
and hence 
\[
q=n
\leq\frac{1+n\varepsilon}{\varepsilon}\log_2(1+n\varepsilon)
\leq\frac4\varepsilon\log_2(1+n\varepsilon).
\]

\noindent \emph{Case 2:} $n > \lambda +1$. In this case we have
\begin{align*}
q
&\leq 1+\sum_{d=0}^{\lceil \log_2 n\rceil-1}\min\{2^d,\lambda+1\} \\ 
&= 1+\sum_{d=0}^{\lceil\log_2(\lambda+1)\rceil-1}
\min\{2^d,\lambda+1\}
+\sum_{d=\lceil\log_2(\lambda+1)\rceil}^{\lceil\log_2 n\rceil-1}
\min\{2^d,\lambda+1\}\\
&=
1+\sum_{d=0}^{\lceil\log_2(\lambda+1)\rceil-1}2^d
+\sum_{d=\lceil\log_2(\lambda+1)\rceil}^{\lceil\log_2 n\rceil-1}
(\lambda+1)\\
&=
1+\sum_{d=0}^{\lceil\log_2(\lambda+1)\rceil-1}2^d
+(\lambda+1)
\bigl(\lceil\log_2 n\rceil-\lceil\log_2(\lambda+1)\rceil\bigr).\\
&=
2^{\lceil\log_2(\lambda+1)\rceil}
+(\lambda+1)
\bigl(\lceil\log_2 n\rceil-\lceil\log_2(\lambda+1)\rceil\bigr).
\end{align*}
The first inequality above is by \eqref{eq:recursive-astsf-query-sum} established earlier. The second equation is because 
\[
\min\{2^d,\lambda+1\}
=
\begin{cases}
2^d, & d<\lceil\log_2(\lambda+1)\rceil,\\
\lambda+1, & d\geq\lceil\log_2(\lambda+1)\rceil.
\end{cases}
\]
The last equation is by the geometric sum formula, namely 
\[
1+\sum_{d=0}^{\lceil\log_2(\lambda+1)\rceil-1}2^d
=
1+\left(2^{\lceil\log_2(\lambda+1)\rceil}-1\right)
=
2^{\lceil\log_2(\lambda+1)\rceil}.
\]
Next, we can expand $2^{\lceil\log_2(\lambda+1)\rceil}$ and upper bound it as follows:
\begin{align*}
2^{\lceil\log_2(\lambda+1)\rceil}
&=
2^{\log_2(\lambda+1)
+\lceil\log_2(\lambda+1)\rceil-\log_2(\lambda+1)}\\
&=
2^{\log_2(\lambda+1)}
\,2^{\lceil\log_2(\lambda+1)\rceil-\log_2(\lambda+1)}\\
&=
(\lambda+1)
\,2^{\lceil\log_2(\lambda+1)\rceil-\log_2(\lambda+1)}\\
&\leq
(\lambda+1)
\bigl(1+\lceil\log_2(\lambda+1)\rceil-\log_2(\lambda+1)\bigr),
\end{align*}
where in the last step we used that for $x \in [0,1]$, we have $2^x \leq 1+x$ (e.g., see \citep[Section~0.2, equation~(5), p.~15]{mitrinovic1964elementary}), as clearly $0\leq\lceil\log_2(\lambda+1)\rceil-\log_2(\lambda+1)<1$. Therefore, we now have
\begin{align*}
q
&\leq
(\lambda+1)
\bigl(1+\lceil\log_2(\lambda+1)\rceil-\log_2(\lambda+1)\bigr)+
(\lambda+1)
\bigl(\lceil\log_2 n\rceil-\lceil\log_2(\lambda+1)\rceil\bigr)\\
&=
(\lambda+1)
\bigl(1+\lceil\log_2 n\rceil-\log_2(\lambda+1)\bigr) \\
&\leq
(\lambda+1)
\left(2+\log_2\frac{n}{\lambda+1}\right) \\
&=(\lambda+1)\log_2\frac{4n}{\lambda+1},
\end{align*}
where the inequality follows by the fact that $\lceil\log_2 n\rceil\leq\log_2 n+1$, and the last equation follows from the known identity for logarithms. By the concavity of the logarithm (e.g., see \citep[Section~1.4.3, Definition~1, equation~(1), p.~15]{mitrinovic1970analytic}), we have
\[
x\log_2 y
\leq\log_2\bigl(x y+1-x\bigr)
\qquad(y>0,\ x\in[0,1]).
\]
Since
\[
0<\frac{\varepsilon(\lambda+1)}4
=\frac{\varepsilon\lceil2/\varepsilon\rceil}{4}
<\frac{2+\varepsilon}{4}\leq\frac34<1,
\]
we can take \(x=\varepsilon(\lambda+1)/4\) and \(y=4n/(\lambda+1)\), and we obtain
\[
\begin{aligned}
q
&\leq
\frac4\varepsilon
\left(\frac{\varepsilon(\lambda+1)}4
\log_2\frac{4n}{\lambda+1}\right)\\
&\leq
\frac4\varepsilon
\log_2\left(1+n\varepsilon-\frac{\varepsilon(\lambda+1)}4\right)\\
&\leq
\frac4\varepsilon\log_2(1+n\varepsilon).
\end{aligned}
\]
This completes the proof.
\end{proof}

\begin{remark}\label{rem:cumbersome-for-constants}
The second part of the proof of \cref{thm:recursive-astsf-distortion} (for the case when $n > \lambda+1$) employs several known mathematical inequalities and related properties, making it notationally somewhat cumbersome. This is merely to optimize the constant in the query complexity bound to $4$, to be as close as possible to the constant of the lower bound on \cref{thm:lower-bound-queries-adaptive-optimal}, i.e., $1/4$. With cruder (and thus simpler) bounding, it is possible to obtain exactly the same asymptotic bounds, at the expense of a larger constant, e.g., $6$.   
\end{remark}

\section{Conclusion and Future Work}\label{sec:conclusion}

In this work, we initiated the study of distortion in the stable matching problem. First, we proved that while the Deferred Acceptance algorithm, as well as any ordinal algorithm for the problem, has unbounded distortion, a randomized version of the algorithm achieves a distortion of $2$; furthermore, this distortion is best possible among all randomized stable algorithms for the problem. In terms of query-enhanced algorithms, we observed that a single-query (deterministic) variant of Deferred Acceptance achieves the same distortion bound of $2$, and proved that this is best possible for any algorithm that uses at most either $O(\log n)$ adaptive queries per agent, or $O(n)$ non-adaptive such queries. Furthermore, we studied the question of how many queries are necessary and sufficient to achieve a distortion of at most $1+\varepsilon$ for any $\varepsilon >0$, and proved that an asymptotically tight bound of $\Theta((\log(1+n\varepsilon)/\varepsilon)$ queries per agent. \medskip 

\noindent We identify some promising avenues for future work, which are outlined below.

\paragraph{Stable Fractional Matchings.} One could consider \emph{stable fractional matchings}, a notion of randomized stable matchings studied notably by \citet{caragiannis2021stable}. In such a matching $\mu$, a blocking pair is a (man,woman) pair such that each agent in the pair has a higher utility for its partner than its expected utility for the partners it is matched with by $\mu$; $\mu$ is stable fractional if there are no blocking pairs. As we suggested in \cref{rem:randomized-lower}, \cref{thm:lower-bound-ordinal-randomized} applies to these matchings as well, and hence $2$ is the best approximation one can hope for. Note however, that, similarly to above, the benchmark that we compare against is stronger; in fact, it can be inferred by the work of \citet{caragiannis2021stable} that the gap between the social welfare of the best fractional and the best integral stable matchings can be very large. This implies that the \RMOWO algorithm, although stable fractional, is not guaranteed to have a distortion of $2$ with respect to this stronger benchmark. 

\paragraph{Unreliable Cardinal Information.} It would be interesting to generalize our results to cases where the answers to the value queries may be imperfect, or even potentially completely erroneous. Concretely, first we may assume that the answer to each value query has an associated imprecision $\eta$; our results can be extended to this case as well, and presenting the bounds as functions of $\eta$ is an interesting technical exercise. In a more intriguing investigation, we may assume that the answers to some of the queries could be entirely \emph{unreliable}, and the distortion of query-enhanced algorithms would need to be robust to this possibility. This sort of question is usually studied in the literature of \emph{predictions} \citep{lykouris2021competitive}, in terms of tradeoffs between \emph{consistency} and \emph{robustness}. The distortion of algorithms with cardinal predictions was studied recently by \citet{filos2025utilitarian} for single-winner voting and one-sided matching. Very recently, \citet{mccauley2026stable} considered stable matchings with predictions, but not in the context of distortion, and with the prediction being on the ordinal rankings of the agents, rather than their cardinal values.   

\section*{Use of AI tools and Discussion}

The first arXiv version of this paper \citep{filosratsikas2026distortionv1}, uploaded on 16 February 2026, was developed and written entirely by the authors, without any use of AI tools. In that version, we introduced the problem, developed the setting, formulated the research questions, and provided most of the results that are also present in this version. In particular, that version already contained the lower bounds of $\Omega(n)$ and $\Omega(\log n)$ non-adaptive and adaptive queries respectively, for a distortion better than $2$ (see \cref{sec:more-queries-impossibility} and \cref{thm:lower-bound-queries-non-adaptive,thm:lower-bound-queries-adaptive} of the present version), as well as the \textsc{Stable-MTSF} algorithm (therein referred to as ``\textsc{Stable-TSF}'') which achieves a distortion of $1+\varepsilon$ with $O(\log n / \varepsilon^2)$ queries per agent, as well as its analysis, as presented in \cref{app:simple-stableTSF} of the present version (see \cref{alg:stsf} and \cref{thm:stsf-distortion}). This version also contained an algorithm that achieves the same distortion bound with $O(\log n / \varepsilon)$ queries, when the Hasse diagram of the rotation poset is a path. 

For this version we used AI tools (namely OpenAI’s GPT-6 Astra and GPT-6.1 Sol) to obtain an improved tight $\Theta((\log(1+n\varepsilon)/\varepsilon)$ bound on the number of queries that are necessary and sufficient in order to achieve distortion $1+\varepsilon$. The AI models operated with our constructions provided as input, via means of well-written and fully developed proofs. For the lower bound proof (\cref{thm:lower-bound-queries-adaptive-optimal}), they employed the same construction (the same cyclic-shift profile, the same proof structure via the uninformed regions etc), and generalized the argument from arity $2$ (dichotomous valuation profiles) used for the lower bound of $2$ in our first version, to a larger arity required for the refined bound of $1+\varepsilon$ in the present version. For the upper bound, the AI models identified that the use of additive rather than multiplicative thresholds is sufficient to improve the query complexity bound by a factor of $1/\varepsilon$, by a virtually identical algorithm (\cref{alg:astsf}) and proof (\cref{thm:astsf-distortion}); we deliberately chose to present our previous algorithm and proof in \cref{app:simple-stableTSF}, to emphasize the close connection. To achieve the precisely tight bound of $O(\log(1+n\varepsilon)/\varepsilon)$, the AI models identified the refined search procedure that is used in \cref{alg:recursive-astsf}, see \cref{thm:recursive-astsf-distortion}. Interestingly, this procedure was considered in the past by the first author in the context of \citep{amanatidis2022few}, but it was dismissed because it does not provide asymptotic improvements in that setting. The proof approaches provided by the AI models were markedly different from the ones presented in the paper. The proofs have been verified and simplified by the authors, as well as placed in the right context in relation to the other results of the paper. The authors assume full responsibility for all the results of the paper, including those obtained with the aid of the AI models.

\bibliographystyle{plainnat}
\bibliography{references}

@article{gale1962college,
  title={College admissions and the stability of marriage},
  author={Gale, David and Shapley, Lloyd S},
  journal={The American mathematical monthly},
  volume={69},
  number={1},
  pages={9--15},
  year={1962},
  publisher={Taylor \& Francis}
}

@article{amanatidis2021peeking,
  title={Peeking behind the ordinal curtain: Improving distortion via cardinal queries},
  author={Amanatidis, Georgios and Birmpas, Georgios and Filos-Ratsikas, Aris and Voudouris, Alexandros A},
  journal={Artificial Intelligence},
  volume={296},
  pages={103488},
  year={2021},
  publisher={Elsevier}
}

@article{irving1986complexity,
  title={The complexity of counting stable marriages},
  author={Irving, Robert W and Leather, Paul},
  journal={SIAM Journal on Computing},
  volume={15},
  number={3},
  pages={655--667},
  year={1986},
  publisher={SIAM}
}

@book{gusfield1989stable,
  title={The stable marriage problem: structure and algorithms},
  author={Gusfield, Dan and Irving, Robert W},
  year={1989},
  publisher={MIT press}
}

@inproceedings{anshelevich2021distortion,
  title={Distortion in Social Choice Problems: The First 15 Years and Beyond},
  author={Anshelevich, Elliot and Filos-Ratsikas, Aris and Shah, Nisarg and Voudouris, Alexandros A},
  booktitle={30th International Joint Conference on Artificial Intelligence},
  pages={4294--4301},
  year={2021},
  organization={International Joint Conferences on Artificial Intelligence Organization}
}

@article{aziz2020justifications,
  title={Justifications of welfare guarantees under normalized utilities},
  author={Aziz, Haris},
  journal={ACM SIGecom Exchanges},
  volume={17},
  number={2},
  pages={71--75},
  year={2020},
  publisher={ACM New York, NY, USA}
}

@article{amanatidis2022few,
  title={A few queries go a long way: Information-distortion tradeoffs in matching},
  author={Amanatidis, Georgios and Birmpas, Georgios and Filos-Ratsikas, Aris and Voudouris, Alexandros A},
  journal={Journal of Artificial Intelligence Research},
  volume={74},
  pages={227--261},
  year={2022}
}

@inproceedings{ebadian2025every,
  title={Every bit helps: Achieving the optimal distortion with a few queries},
  author={Ebadian, Soroush and Shah, Nisarg},
  booktitle={Proceedings of the AAAI Conference on Artificial Intelligence},
  volume={39},
  number={13},
  pages={13788--13795},
  year={2025}
}

@article{amanatidis2024dont,
author = {Amanatidis, Georgios and Birmpas, Georgios and Filos-Ratsikas, Aris and Voudouris, Alexandros A.},
title = {Don’t Roll the Dice, Ask Twice: The Two-Query Distortion of Matching Problems and Beyond},
journal = {SIAM Journal on Discrete Mathematics},
volume = {38},
number = {1},
pages = {1007-1029},
year = {2024},
doi = {10.1137/23M1545677},

URL = { 
    
        https://doi.org/10.1137/23M1545677
    
    

},
eprint = { 
    
        https://doi.org/10.1137/23M1545677
    
    

}
}

@article{gusfield1987three,
  title={Three fast algorithms for four problems in stable marriage},
  author={Gusfield, Dan},
  journal={SIAM Journal on Computing},
  volume={16},
  number={1},
  pages={111--128},
  year={1987},
  publisher={SIAM}
}

@article{birkhoff1946tres,
  title={Tres observaciones sobre el algebra lineal},
  author={Birkhoff, Garrett},
  journal={Univ. Nac. Tucuman, Ser. A},
  volume={5},
  pages={147--154},
  year={1946}
}

@book{trotter2002combinatorics,
  title={Combinatorics and partially ordered sets},
  author={Trotter, William T},
  year={2002},
  publisher={Johns Hopkins University Press}
}

@article{roth1993stable,
  title={Stable matchings, optimal assignments, and linear programming},
  author={Roth, Alvin E and Rothblum, Uriel G and Vande Vate, John H},
  journal={Mathematics of operations research},
  volume={18},
  number={4},
  pages={803--828},
  year={1993},
  publisher={INFORMS}
}

@article{caragiannis2021stable,
  title={Stable fractional matchings},
  author={Caragiannis, Ioannis and Filos-Ratsikas, Aris and Kanellopoulos, Panagiotis and Vaish, Rohit},
  journal={Artificial Intelligence},
  volume={295},
  number={103416},
  pages={103416},
  year={2021}
}

@article{aziz2019random,
  title={Random matching under priorities: stability and no envy concepts},
  author={Aziz, Haris and Klaus, Bettina},
  journal={Social Choice and Welfare},
  volume={53},
  number={2},
  pages={213--259},
  year={2019},
  publisher={Springer}
}

@inproceedings{procaccia2006distortion,
  title={The distortion of cardinal preferences in voting},
  author={Procaccia, Ariel D. and Rosenschein, Jeffrey S.},
  booktitle={Proceedings of the 10th International Workshop on Cooperative Information Agents ({CIA})},
  pages={317--331},
  year={2006},
}

@article{boutilier2015optimal,
  title={Optimal social choice functions: A utilitarian view},
  author={Boutilier, Craig and Caragiannis, Ioannis and Haber, Simi and Lu, Tyler and Procaccia, Ariel D. and Sheffet, Or},
  journal={Artificial Intelligence},
  volume={227},
  pages={190--213},
  year={2015},
}

@article{caragiannis2017subset,
	title={Subset selection via implicit utilitarian voting},
	author={Caragiannis, Ioannis and Nath, Swaprava and Procaccia, Ariel D. and Shah, Nisarg},
	journal={Journal of Artificial Intelligence Research},
	volume={58},
	pages={123--152},
	year={2017}
}

@inproceedings{filos2025utilitarian,
  title={Utilitarian Distortion with Predictions},
  author={Filos-Ratsikas, Aris and Kalantzis, Georgios and Voudouris, Alexandros A},
  booktitle={Proceedings of the 26th ACM Conference on Economics and Computation},
  pages={254--271},
  year={2025}
}

@article{roth1992two,
  title={Two-sided matching},
  author={Roth, Alvin E and Sotomayor, Marilda},
  journal={Handbook of game theory with economic applications},
  volume={1},
  pages={485--541},
  year={1992},
  publisher={Elsevier}
}

@book{manlove2013algorithmics,
  title={Algorithmics of matching under preferences},
  author={Manlove, David},
  volume={2},
  year={2013},
  publisher={World Scientific}
}

@article{roth1984evolution,
  title={The evolution of the labor market for medical interns and residents: a case study in game theory},
  author={Roth, Alvin E},
  journal={Journal of political Economy},
  volume={92},
  number={6},
  pages={991--1016},
  year={1984},
  publisher={The University of Chicago Press}
}

@article{roth1999redesign,
  title={The redesign of the matching market for American physicians: Some engineering aspects of economic design},
  author={Roth, Alvin E and Peranson, Elliott},
  journal={American economic review},
  volume={89},
  number={4},
  pages={748--780},
  year={1999},
  publisher={American Economic Association}
}

@article{abdulkadirouglu2005new,
  title={The new york city high school match},
  author={Abdulkadiro{\u{g}}lu, Atila and Pathak, Parag A and Roth, Alvin E},
  journal={American Economic Review},
  volume={95},
  number={2},
  pages={364--367},
  year={2005},
  publisher={American Economic Association}
}

@article{irving1987efficient,
  title={An efficient algorithm for the “optimal” stable marriage},
  author={Irving, Robert W and Leather, Paul and Gusfield, Dan},
  journal={Journal of the ACM (JACM)},
  volume={34},
  number={3},
  pages={532--543},
  year={1987},
  publisher={ACM New York, NY, USA}
}

@inproceedings{feder1989new,
  title={A new fixed point approach for stable networks stable marriages},
  author={Feder, Tom{\'a}s},
  booktitle={Proceedings of the twenty-first annual ACM symposium on Theory of computing},
  pages={513--522},
  year={1989}
}

@article{anshelevich2013anarchy,
  title={Anarchy, stability, and utopia: creating better matchings},
  author={Anshelevich, Elliot and Das, Sanmay and Naamad, Yonatan},
  journal={Autonomous Agents and Multi-Agent Systems},
  volume={26},
  number={1},
  pages={120--140},
  year={2013},
  publisher={Springer}
}

@book{knuth1997stable,
  title={Stable marriage and its relation to other combinatorial problems: An introduction to the mathematical analysis of algorithms},
  author={Knuth, Donald Ervin},
  volume={10},
  year={1997},
  publisher={American Mathematical Soc.}
}

@inproceedings{caragiannis2024beyond,
  title={Beyond the worst case: Distortion in impartial culture electorates},
  author={Caragiannis, Ioannis and Fehrs, Karl},
  booktitle={International Conference on Web and Internet Economics},
  pages={420--437},
  year={2024},
  organization={Springer}
}

@inproceedings{ashlagi2025stable,
  title={Stable Matching with Interviews},
  author={Ashlagi, Itai and Chen, Jiale and Roghani, Mohammad and Saberi, Amin},
  booktitle={16th Innovations in Theoretical Computer Science Conference (ITCS 2025)},
  pages={12--1},
  year={2025},
  organization={Schloss Dagstuhl--Leibniz-Zentrum f{\"u}r Informatik}
}

@inproceedings{rastegari2016preference,
  title={Preference elicitation in matching markets via interviews: a study of offline benchmarks},
  author={Rastegari, B and Goldberg, P and Manlove, D},
  booktitle={Proceedings of the International Joint Conference on Autonomous Agents and Multiagent Systems, AAMAS},
  year={2016},
  organization={International Foundation for Autonomous Agents and Multiagent Systems}
}

@inproceedings{emek2015price,
  title={The price of matching with metric preferences},
  author={Emek, Yuval and Langner, Tobias and Wattenhofer, Roger},
  booktitle={Algorithms-ESA 2015: 23rd Annual European Symposium, Patras, Greece, September 14-16, 2015, Proceedings},
  pages={459--470},
  year={2015},
  organization={Springer}
}

@article{arkin2009geometric,
  title={Geometric stable roommates},
  author={Arkin, Esther M and Bae, Sang Won and Efrat, Alon and Okamoto, Kazuya and Mitchell, Joseph SB and Polishchuk, Valentin},
  journal={Information Processing Letters},
  volume={109},
  number={4},
  pages={219--224},
  year={2009},
  publisher={Elsevier}
}

@article{caragiannis2011voting,
  title={Voting almost maximizes social welfare despite limited communication},
  author={Caragiannis, Ioannis and Procaccia, Ariel D},
  journal={Artificial Intelligence},
  volume={175},
  number={9-10},
  pages={1655--1671},
  year={2011},
  publisher={Elsevier}
}

@inproceedings{rastegari2013two,
  title={Two-sided matching with partial information},
  author={Rastegari, Baharak and Condon, Anne and Immorlica, Nicole and Leyton-Brown, Kevin},
  booktitle={Proceedings of the fourteenth ACM conference on Electronic Commerce},
  pages={733--750},
  year={2013}
}

@inproceedings{drummond2014preference,
  title={Preference elicitation and interview minimization in stable matchings},
  author={Drummond, Joanna and Boutilier, Craig},
  booktitle={Proceedings of the AAAI Conference on Artificial Intelligence},
  volume={28},
  number={1},
  year={2014}
}

@article{gonczarowski2019stable,
  title={A stable marriage requires communication},
  author={Gonczarowski, Yannai A and Nisan, Noam and Ostrovsky, Rafail and Rosenbaum, Will},
  journal={Games and Economic Behavior},
  volume={118},
  pages={626--647},
  year={2019},
  publisher={Elsevier}
}

@article{ng1990lower,
  title={Lower bounds for the stable marriage problem and its variants},
  author={Ng, Cheng and Hirschberg, Daniel S},
  journal={SIAM Journal on Computing},
  volume={19},
  number={1},
  pages={71--77},
  year={1990},
  publisher={SIAM}
}

@article{segal2007communication,
  title={The communication requirements of social choice rules and supporting budget sets},
  author={Segal, Ilya},
  journal={Journal of Economic Theory},
  volume={136},
  number={1},
  pages={341--378},
  year={2007},
  publisher={Elsevier}
}

@inproceedings{koutsoupias1999worst,
  title={Worst-case equilibria},
  author={Koutsoupias, Elias and Papadimitriou, Christos},
  booktitle={Annual symposium on theoretical aspects of computer science},
  pages={404--413},
  year={1999},
  organization={Springer}
}

@article{mccauley2026stable,
  title={Stable Matching with Predictions: Robustness and Efficiency under Pruned Preferences},
  author={McCauley, Samuel and Moseley, Benjamin and Niaparast, Helia and Singh, Shikha},
  journal={arXiv preprint arXiv:2602.02254},
  year={2026}
}

@article{lykouris2021competitive,
  title={Competitive caching with machine learned advice},
  author={Lykouris, Thodoris and Vassilvitskii, Sergei},
  journal={Journal of the ACM (JACM)},
  volume={68},
  number={4},
  pages={1--25},
  year={2021},
  publisher={ACM New York, NY}
}

@article{babaioff2026efficient,
  title={Efficient Interview Scheduling for Stable Matching},
  author={Babaioff, Moshe and Gil, Rotem and Assaf, Romm},
  journal={arXiv preprint arXiv:2602.20358},
  year={2026}
}

@book{mitrinovic1964elementary,
  author    = {Mitrinovi{\'c}, Dragoslav S.},
  title     = {Elementary Inequalities},
  publisher = {P. Noordhoff},
  address   = {Groningen},
  year      = {1964}
}

@book{mitrinovic1970analytic,
  author    = {Mitrinovi{\'c}, Dragoslav S.},
  title     = {Analytic Inequalities},
  publisher = {Springer-Verlag},
  address   = {Berlin Heidelberg},
  year      = {1970},
  series    = {Grundlehren der mathematischen Wissenschaften},
  volume    = {165},
  doi       = {10.1007/978-3-642-99970-3}
}

@misc{filosratsikas2026distortionv1,
  author       = {Aris Filos-Ratsikas and Georgios Kalantzis},
  title        = {The Distortion of Stable Matching},
  year         = {2026},
  howpublished = {arXiv preprint arXiv:2602.14961v1},
  note         = {Version 1, 16 February 2026},
  url          = {https://arxiv.org/abs/2602.14961v1}
}

\newpage 
\appendix
\crefalias{section}{appendix}

\section*{Appendix}

\section{Missing proofs from \cref{app:background-stable-matchings}}\label{app: appendix-structure-stable-matching}

\begin{proof}[Proof of \cref{prop:shared-agent-precedence}]
    We prove the statement when the common agent is a man $m\in M$; the proof for a woman is symmetric. Suppose that $m$ appears in both rotations $\rho$ and $\rho'$. Following \cref{def:elimination-of-rotation} when $\rho$ is exposed, $m$ is matched to some woman $w_\rho$, and eliminating $\rho$ moves $m$ to a less preferred woman. Similarly, when $\rho'$ is exposed, $m$ is matched to some woman $w_{\rho'}$, and eliminating $\rho'$ moves $m$ to a less preferred woman. We first note that $w_\rho\neq w_{\rho'}$. Indeed, if $w_\rho=w_{\rho'}$, then the pair $(m,w_\rho)$ appears in both rotations $\rho$ and $\rho'$, contradicting \cref{lem: pair-in-one-rotation}. Thus either $w_\rho\succ_m w_{\rho'}$ or $w_{\rho'}\succ_m w_\rho$. Assume without loss of generality that $w_\rho\succ_m w_{\rho'}$. By \cref{def:elimination-of-rotation} along any sequence of rotation eliminations, the partner of $m$ can only move downward in $m$'s preference order, and it changes only when the eliminated rotation contains $m$. Therefore, before $m$ can be matched to $w_{\rho'}$, it must have already passed through the partner $w_\rho$ and hence the rotation $\rho$ must already have been eliminated. Consequently, every stable matching in which $\rho'$ is exposed has $\rho$ already eliminated. By the \cref{def:rotation-precedes} of the rotation precedence relation, this means $\rho\prec\rho'$. If instead $w_{\rho'}\succ_m w_\rho$, the same argument gives $\rho'\prec\rho$. Hence $\rho$ and $\rho'$ are comparable. 
\end{proof}

\noindent By applying the above proposition we will also make use of the following corollary.

\begin{proof}[Proof of \cref{lem:fixed-agent-chain}]
    We prove the claim for a fixed man $m\in M$; the proof for a fixed woman is identical. Let $\rho,\rho'\in\mathcal R(m)$ be distinct rotations. Since both rotations contain $m$, they share an agent. By \cref{prop:shared-agent-precedence}, the two rotations are comparable in the rotation poset. Hence either $\rho\prec\rho'$ or $\rho'\prec\rho$. Therefore every two distinct rotations in $\mathcal R(m)$ are comparable, so $\mathcal R(m)$ is a chain. The same argument applies to $\mathcal R(w)$ for every fixed woman $w\in W$. 
\end{proof}

\section{A Version of the Algorithm with Multiplicative Thresholds} \label{app:simple-stableTSF}

In this section of the appendix, we present a version of the \textsc{Stable-ATSF} algorithm introduced in \cref{sec:more-queries-positive}, which uses multiplicative thresholds rather than additive; we refer to this algorithm as \textsc{Stable Multiplicative Threshold Step Function (Stable-MTSF)}, and present its pseudocode in \cref{alg:stsf}. Multiplicative thresholds were used by \citet{amanatidis2022few} to obtain distortion/query complexity tradeoffs for one-sided matching (without stability). We remark that \textsc{Stable Multiplicative Threshold Step Function (Stable-MTSF)} already takes advantage of the inherent properties of stable matching and exhibits two significant novelties compared to the ``vanilla'' version of \citet{amanatidis2022few}. This enables the $O(\log n /\varepsilon^2)$ query complexity bound (rather than the $O(\log^2 n/ \varepsilon)$ bound proven in that work for their setting).

\begin{algorithm}[!t]
\caption{\textsc{Stable Multiplicative Threshold Step Function (Stable-MTSF)}}
\label{alg:stsf}
\LinesNumbered
\KwIn{A preference profile $\succ=(\succ_M, \succ_W)$, and an $\varepsilon \in (0,1]$.}
\KwOut{A stable matching $\mu \in \mathcal{M}_s(\succ)$.}
\BlankLine

Let $\lambda = \left\lceil \frac{4\ln (4/\varepsilon)}{\varepsilon} \right\rceil$, and $t_\ell = (\varepsilon/4)^{\ell/\lambda}$, for $\ell = 0, \cdots, \lambda$.\\
\tcc{The thresholds satisfy $t_0=1$ and $t_\lambda=\varepsilon/4$, and $t_{\ell-1}/t_\ell=(4/\varepsilon)^{1/\lambda}$ for every $\ell=1,\ldots,\lambda$.}
\BlankLine

\If{$n \leq 1+\lambda\lceil\log_2 n\rceil$}
{Perform $n$ queries per agent to learn the whole valuation profile $q^n(\bv)=\bv$. Let $\tilde{\bv} = \bv$.} \label{line:query-all}
\Else{
Run Algorithm $\mathcal{A}_\text{gus}$ to compute, for each agent $a_i \in M \cup W$, the set $S_{a_i}$ of its stable partners (by \cref{lem:all-stable-pairs}), and construct the preference ranking $\succ_{a_i}^s$ based on $S_{a_i}$.
\BlankLine

\ForEach{agent $a_i \in A$, for both $A=M,W$}{

    Query agent $a_i$ for her favorite stable partner $f^s(a_i) \in \bar{A}$.\\
    Set $T_{a_i,0} \leftarrow \{f^s(a_i)\}$, and let $v_{a_i}^* = v_{a_i}(f^s(a_i))$.\\
    Set $\tilde v_{a_i}(f^s(a_i)) = v_{a_i}^*$.\\
    \tcc{Preserve the favorite's value. Earlier buckets take priority, so this value is never overwritten.}

    \For{$\ell = 1$ \KwTo $\lambda$}{
    \tcc{Use binary search to create buckets, where the $\ell$-th bucket contains all the stable partners of $a_i$ for whom $a_i$'s value is between $t_\ell$ and $t_{\ell-1}$ times its value for its favorite stable partner, excluding partners already assigned to earlier buckets.}
        Using binary search on $\succ^s_{a_i}$, compute \label{line:binary_search}
        \[
        T_{a_i,\ell}
        \leftarrow
        \left\{
            a_j \in S_{a_i} : v_{a_i}(a_j) \in [t_\ell \cdot v_{a_i}^*,\, t_{\ell-1} \cdot v_{a_i}^*]
        \right\}
        \setminus \bigcup_{r=0}^{\ell-1} T_{a_i,r}
        \]

        Set $\tilde v_{a_i}(a_j) = t_\ell \cdot v_{a_i}^*$ for every $a_j \in T_{a_i,\ell}$\\
        \tcc{Set the simulated value of agent $a_i$ for agent $a_j$ to be the lower boundary of the bucket.}
    }

    Set $T_{a_i} \leftarrow \bigcup_{\ell=0}^{\lambda} T_{a_i,\ell}$\\
    \tcc{The set of stable partners contained in the constructed buckets.}

    Set $\tilde v_{a_i}(a_j) = 0$ for every $a_j \in \bar{A} \setminus T_{a_i}$\\
    \tcc{Set the simulated value to $0$ for partners outside the constructed buckets, including stable partners whose values are below the last threshold.}
}
}

\BlankLine
\Return $\mu \in \arg \max_{\hat{\mu} \in \mathcal{M}_s(\succ)} \SW (\hat{\mu} \mid \tilde{\bv})$, using Algorithm $\mathcal{A}_\text{ilg}$.\\
\tcc{Return a stable matching that maximizes the simulated social welfare, using the algorithm of \citet{irving1987efficient}.}
\end{algorithm}

We state the theorem regarding the number of queries required to achieve a distortion of $1+\varepsilon$ using \cref{alg:stsf}. The proof is almost identical to that of \cref{thm:astsf-distortion}; we present it in full here for completeness.

\begin{theorem}\label{thm:stsf-distortion}
For any $\varepsilon\in(0,1]$, the \textsc{Stable-MTSF} algorithm performs at most
\[\min\left\{n,\,1+\left\lceil\frac{4\ln(4/\varepsilon)}{\varepsilon}\right\rceil\lceil\log_2 n\rceil\right\}\]
queries per agent and achieves distortion at most $1+\varepsilon$.
\end{theorem}

\begin{proof}
Consider any valuation profile $\bv$, and let $\mu^*(\bv)$ be the optimal stable matching on $\bv$. Let $\mu$ be the stable matching returned by the \textsc{Stable-MTSF} algorithm. If the algorithm learns the whole valuation profile, it returns an optimal stable matching. Hence, for the distortion analysis, we may
assume that the algorithm executes the ``\textbf{else}'' branch.  We will partition the agents in $M\cup W$ into two sets $S$ and $\bar S$, where $S$ will contain the agents whose partner in $\mu^*$ is in the set $T_{a_i}$, i.e., their partner falls into one of the first $\lambda+1$ buckets constructed by the algorithm. Formally, we have
\[
S:=\{a_i\in M\cup W\mid \mu^*(a_i)\in T_{a_i}\},
\qquad\text{and}\qquad
\bar S=(M\cup W)\setminus S.
\]
For ease of notation, let also $\bar S_M=\bar S\cap M$ and $\bar S_W=\bar S\cap W$, and observe that $\bar S_M\cup\bar S_W=\bar S$. We will bound the contribution to the social welfare of $\mu^*$ from agents in those two sets separately. We start with the agents in the set $\bar S$. We have:
\begin{align}
\sum_{a_i\in\bar S}v_{a_i}(\mu^*(a_i))
&=
\sum_{m_i\in\bar S_M}v_{m_i}(\mu^*(m_i))
+
\sum_{w_i\in\bar S_W}v_{w_i}(\mu^*(w_i))
\notag\\
&\leq
t_\lambda\cdot
\left(
\sum_{m_i\in\bar S_M}v_{m_i}(f^s(m_i))
+
\sum_{w_i\in\bar S_W}v_{w_i}(f^s(w_i))
\right)
\notag\\
&\leq
t_\lambda\cdot
\left(
\SW_M(\mu_M^*\mid\bv)
+
\SW_W(\mu_W^*\mid\bv)
\right)
\notag\\
&\leq
2t_\lambda\cdot\SW(\mu\mid\bv)
=
\frac{\varepsilon}{2}\cdot\SW(\mu\mid\bv).
\label{eq:stsf-tail}
\end{align}
The first inequality follows from the definition of the set $\bar S$, since each agent $a_i$ in this set has value at most $t_\lambda\cdot v_{a_i}(f^s(a_i))$ for their assigned partner $\mu^*(a_i)$ in the optimal matching. The second inequality follows from the fact that in the man-optimal (resp.\ woman-optimal) matching, each man (resp.\ woman) is matched with his (resp.\ her) favorite stable partner.

The last inequality follows from the fact that every agent $a_i$ has been queried for its favorite stable partner $f^s(a_i)$, and therefore
$\tilde v_{a_i}(f^s(a_i))=v_{a_i}(f^s(a_i))$.
Since the algorithm returns a matching that maximizes the simulated social welfare, and the simulated social welfare is always a lower bound on the actual social welfare of a stable matching, we have
\begin{align*}
\SW_M(\mu_M^*\mid\bv)+\SW_W(\mu_W^*\mid\bv)
&=
\sum_{a_i\in M\cup W}v_{a_i}(f^s(a_i))\\
&\leq
\SW(\mu_M^*\mid\tilde{\bv})
+
\SW(\mu_W^*\mid\tilde{\bv})\\
&\leq
2\SW(\mu\mid\tilde{\bv})
\leq
2\SW(\mu\mid\bv).
\end{align*}
The last equation in \eqref{eq:stsf-tail} follows by the definition of $t_\lambda$. \medskip 

\noindent Next, we bound the contribution to the social welfare from the agents in $S$. For $\ell=0,\ldots,\lambda$, let
\[
S_\ell=\{a_i\in S\mid \mu^*(a_i)\in T_{a_i,\ell}\}
\]
be the restriction of the agents in $S$ to those whose partner in the optimal stable matching is in the $\ell$-th bucket. For any $\ell\geq1$, any agent $a_i\in S_\ell$, and any agent $a_j\in T_{a_i,\ell}$, we have
\[
v_{a_i}(a_j)
\leq
t_{\ell-1}\cdot v_{a_i}(f^s(a_i))
=
\frac{t_{\ell-1}}{t_\ell}\cdot
t_\ell\cdot v_{a_i}(f^s(a_i))
=
\frac{1}{\delta}\cdot\tilde v_{a_i}(a_j),
\]
where
\[
\delta=\frac{t_\ell}{t_{\ell-1}}
=(\varepsilon/4)^{1/\lambda}
\qquad\text{for all }\ell=1,\ldots,\lambda.
\]
The first inequality above follows by the definition of the set $T_{a_i,\ell}$, and the last equation follows from the definition of the simulated valuation function.
The same final inequality holds for $\ell=0$, where the simulated and true values coincide, since $\delta\leq1$.
Therefore, we can bound the contribution to the social welfare of $\mu^*$ from the agents in $S$ as follows:
\begin{align}
\sum_{a_i\in S}v_{a_i}(\mu^*(a_i))
&=
\sum_{\ell=0}^{\lambda}
\sum_{a_i\in S_\ell}v_{a_i}(\mu^*(a_i))
\notag\\
&\leq
\frac{1}{\delta}\cdot
\sum_{\ell=0}^{\lambda}
\sum_{a_i\in S_\ell}\tilde v_{a_i}(\mu^*(a_i))
\notag\\
&=
\frac{1}{\delta}\cdot\SW(\mu^*\mid\tilde{\bv})
\notag\\
&\leq
\frac{1}{\delta}\cdot\SW(\mu\mid\bv)
=
(4/\varepsilon)^{1/\lambda}\cdot\SW(\mu\mid\bv).
\label{eq:stsf-represented}
\end{align}
Here the second equation uses the fact that agents in $\bar S$ have simulated value zero for their partners in $\mu^*$. The last inequality follows from the fact that $\mu$ maximizes the social welfare when calculated with the simulated valuation functions, and that these functions underestimate the true values on stable pairs.
The last equation follows from the definition of $\delta$. \medskip 

\noindent Putting everything together, and substituting the value of $\lambda$, we have that
\begin{align*}
\operatorname{dist}(\textsc{Stable-MTSF})
&\leq
\frac{\varepsilon}{2}+(4/\varepsilon)^{1/\lambda}\\
&=
\frac{\varepsilon}{2}
+
\exp\left(\frac{\ln(4/\varepsilon)}{\lambda}\right)\\
&=
\frac{\varepsilon}{2}
+
\exp\left(
\frac{\ln(4/\varepsilon)}
{\left\lceil4\ln(4/\varepsilon)/\varepsilon\right\rceil}
\right)\\
&\leq
\frac{\varepsilon}{2}+e^{\varepsilon/4}
\leq
1+\varepsilon,
\end{align*}
where the last inequality holds for any $\varepsilon\in(0,1]$: by convexity,
\[
e^{\varepsilon/4}
\leq
1+\varepsilon(e^{1/4}-1)
\leq
1+\frac{\varepsilon}{2}.
\]
The welfare inequality also holds when all stable matchings have welfare zero.\medskip 

\noindent Each binary search procedure in \cref{alg:stsf} requires at most $\lceil\log_2 n\rceil$ queries after the initial favorite query, and the procedure is called $\lceil4\ln(4/\varepsilon)/\varepsilon\rceil$ times per agent. If $n\leq1+\lambda\lceil\log_2 n\rceil$, the algorithm queries the whole valuation profile using $n$ queries per agent, therefore the number of queries that it performs is upper bounded by
\[
\min\left\{
n,\;
1+
\left\lceil\frac{4\ln(4/\varepsilon)}{\varepsilon}\right\rceil
\lceil\log_2 n\rceil
\right\}.
\]
This completes the proof.
\end{proof}

\begin{remark}\label{rem:additive-terms-not-applied}
As the reader can observe, the proofs of \cref{thm:stsf-distortion} and \cref{thm:astsf-distortion} presented in \cref{sec:more-queries-positive} are nearly identical, modulo the use of the different values for $\lambda$ and $t_\ell$, for $\ell = 0, \ldots, \lambda$. Another inconsequential difference in the analysis is in the role of the additive correction term $\sum_{a_i \in S} v_{a_i}(f^s(a_i))$ which appears in the bounds for the social welfare of agents in the sets $\bar{S}$ and $S$ (see \eqref{eq:additive-stsf-negative-correction} and \eqref{eq:additive-stsf-positive-error} respectively in the proof of \cref{thm:astsf-distortion}). As we mentioned in \cref{rem:additive-terms-non-needed}, one could omit these terms from the proof of \cref{thm:astsf-distortion} as well, and upper bound the social welfare contributions appropriately, at the expense of worse constants, but the same asymptotic number of queries. This applies to the proof of \cref{thm:stsf-distortion} above as well, where a more refined analysis that incorporates these terms can achieve an improved bound of $\min\left\{n,\,1+\left\lceil\frac{2\ln(2/\varepsilon)}{\varepsilon}\right\rceil\lceil\log_2 n\rceil\right\}$ in the number of queries. We also remark that the actual bound proven in the proof is in fact $O\!\left(\frac{\log n\,\ln(4/\varepsilon)}{\varepsilon}\right)$; we elected to state the slightly higher bound instead, as an easier point of reference. 
\end{remark}

\end{document}